\documentclass[twocolumn,superscriptaddress,prb]{revtex4-1}
\usepackage{lipsum}
\usepackage{amsfonts}
\usepackage{color}
\usepackage{amsmath,amssymb,mathtools,amsthm,mathrsfs}

\usepackage{braket}
\setcitestyle{numbers,square}
\usepackage{overpic}
\usepackage{booktabs}
\usepackage{tabularx}
\usepackage{makecell}

\usepackage[table]{xcolor}

\usepackage{graphicx} 
\usepackage[table]{xcolor}

\usepackage[normalem]{ulem}

\usepackage{array}

\usepackage[caption=false]{subfig}

\usepackage[colorlinks,bookmarks=false,citecolor=blue,linkcolor=red,urlcolor=blue]{hyperref}

\usepackage{comment}
\definecolor{DictHeader}{RGB}{220,235,250}
\definecolor{DictRowA}{RGB}{248,251,255}
\definecolor{DictRowB}{RGB}{242,247,252}
\definecolor{DictRule}{RGB}{120,150,185}

\usepackage[table]{xcolor}
\usepackage{array}
\usepackage{graphicx}

\newcommand{\RNum}[1]{\uppercase\expandafter{\romannumeral #1\relax}}

\newtheorem{theorem}{Theorem}

\def \be {\begin{equation}}
\def \ee {\end{equation}}
\def \ba {\begin{array}}
\def \ea {\end{array}}
\def \bea {\begin{eqnarray}}
\def \eea {\end{eqnarray}}

\def \ble {\begin{widetext}\begin{equation}}
\def \ele {\end{equation}\end{widetext}}
\def \blea {\begin{widetext}\begin{eqnarray}}
\def \elea {\end{eqnarray}\end{widetext}}

\def \Tr {{\mathrm{Tr}}}

\def \and {{\mathrm{and}}}

\begin{document}


\title{Universal Crossovers of Stabilizer Entropy Beyond Criticality}

\author{Reyhaneh~Khasseh}
\thanks{These authors contributed equally to this work.}
\affiliation{Theoretical Physics III, Center for Electronic Correlations and Magnetism, Institute of Physics, University of Augsburg, D-86135 Augsburg, Germany}
\author{E.~A.~Ramirez~Trino}
\thanks{These authors contributed equally to this work.}
\author{M.~A.~Rajabpour}
\affiliation{Instituto de Física, Universidade Federal Fluminense, Av.~Gal.~Milton Tavares de Souza s/n, Gragoatá, 24210-346, Niterói, RJ, Brazil}

\begin{abstract}
Stabilizer R\'enyi entropy has emerged as a probe of nonstabilizerness in quantum many-body systems, but its scaling structure beyond critical points remains poorly understood compared with entanglement entropy. Recent field-theory approaches indicate that stabilizer entropy contains universal critical data and boundary-sensitive terms, raising the question of how these structures extend into massive and crossover regimes. We address this problem for a broad class of finite-range spin chains at R\'enyi index one-half. We derive exact finite-size formulas for both full periodic chains and finite intervals of the infinite chain, making the universal crossover from critical to noncritical behavior analytically accessible. In periodic geometry, the entropy obeys a volume law away from criticality and exhibits a universal finite-size crossover controlled by the competition between system size and correlation length. We also show that the large-scale SRE density develops a cusp across the field-tuned critical line, while the XX endpoint is governed by a distinct scaling regime associated with the saturation point. In the subsystem geometry, the interval entropy separates bulk critical behavior from boundary contributions generated by the way the finite region cuts the infinite chain. The crossover from critical to massive behavior is then encoded in boundary constants and universal functions controlled by the correlation length. Through exact stabilizer-entropy correspondences, the scaling theory extends to internal XY reductions, Finite-range spin chains, and Cluster--Ising representatives. Our results provide an exact lattice benchmark for the emerging QFT description of stabilizer entropy beyond isolated conformal points.
\end{abstract}

\maketitle

\tableofcontents

\section{Introduction}

\subsection{Motivation and background}

The characterization of quantum many-body states increasingly relies on
quantities that go beyond local order parameters. Entanglement has played a particularly important role in this direction. It provides a quantitative measure of non-classical correlations, detects universal structures at quantum critical points, and connects microscopic lattice systems with conformal field theory (CFT)~\cite{Holzhey1994,Vidal2003,CalabreseCardy2004,CalabreseCardy2009,Laflorencie2016}. However, entanglement does not exhaust the quantum resources contained in a many-body state. From the perspective of quantum computation, another fundamental resource is nonstabilizerness, namely, the part of a state that cannot be generated within the stabilizer formalism. Since stabilizer states and Clifford circuits are classically simulable, nonstabilizerness is a necessary resource for universal quantum computation \cite{Gottesman1997,BravyiKitaev2005,HowardCampbell2017}.

Stabilizer Rényi entropies (SRE) provide a useful family of probes of this resource~\cite{Leone2022a}. Constructed from the expansion of a quantum state in the Pauli basis, these entropies quantify the spread of Pauli expectation values. In this sense, they are analogous to Shannon--Rényi (SR) or participation entropies, but in operator space rather than in a fixed wave-function basis~\cite{Stephan2009,Stephan2011,AlcarazRajabpour2014,Stephan2014,LuitzAletLaflorencie2014,LuitzLaflorencie2017}. They vanish on stabilizer states and become nonzero when Pauli expectation values spread beyond the stabilizer structure. SREs therefore provide a bridge between quantum information theory and the structure of many-body wave functions \cite{HaugPiroli2023,LeoneBittel2024}.

A central question is whether SRE should be viewed only as a measure of a computational resource, or also as a universal many-body diagnostic of phases and critical points. This question is especially relevant in one-dimensional systems. In such systems, entanglement entropies are controlled by CFT at criticality, while SR entropies display universal logarithmic corrections, boundary contributions, and basis-dependent critical data \cite{CalabreseCardy2004,CalabreseCardy2009,Stephan2009,Stephan2011,AlcarazRajabpour2014,Stephan2014}. Recent works have shown that nonstabilizerness can act as a diagnostic of many-body structure, criticality, gauge-theory dynamics, tensor-network complexity, and topological or nonlocal magic responses \cite{White2021,Leone2022b,Tarabunga2023,Tarabunga2024MPS,Frau2025,Falcao2025Gauge,DingYan2025,Viscardi2026SpinModels,Nehra2025TopologicalMagic,Collura2026,Hallam2026}. At the same time, field-theory approaches indicate that SREs contain universal critical data, with subleading terms related to boundary conditions, doubled Hilbert spaces, defects, and boundary-condition-changing structures \cite{HoshinoOshikawaAshida2026,RamirezRajabpour2026,HoshinoAshida2025Defects}. Recent finite-temperature studies of open critical chains further support this picture, showing that stabilizer entropy can reveal boundary- or defect-like conformal data beyond ordinary thermodynamic probes \cite{KhassehRajabpour2026ThermalSRE}.

However, the microscopic origin of these universal terms is still not fully understood. Compared with entanglement entropy, where both critical and massive regimes are well developed, much less is known for SRE. Existing CFT results mainly clarify the universal structure at criticality, including the role of boundary conditions, doubled Hilbert spaces, defects, and Rényi-index-dependent effects. In particular, the appearance of special behavior at certain Rényi indices, such as the boundary-sector transition at\(\alpha_c=4\), shows that SRE has a richer field-theory structure than a simple entropy density. Related numerical evidence for such an \(\alpha_c\) transition has also been found in~\cite{Dalmonte2026}. This transition is also natural from the SRE--SR entropy correspondence, since the corresponding SR entropy of critical free fermions is known to undergo a boundary transition at the free-fermion point \cite{Stephan2011,RamirezRajabpour2026}. Recent numerical and quantum Monte Carlo studies further show that SRE and its volume-law terms can display singular responses at quantum critical points \cite{DingYan2025}. Knowing the SRE at an isolated critical point is not sufficient. A complete scaling description should also explain its behavior in the massive regime, its approach to criticality, the finite-size and finite-interval crossover functions connecting these limits, and the changes that arise when one considers a subsystem rather than the full chain. These questions determine whether SRE can be treated as an exact scaling observable, in the same spirit as entanglement entropy or participation entropy.

This leads to the main question of this work: can SRE be computed exactly in microscopic spin chains in a way that reveals its massive, critical, crossover, endpoint, and boundary-sensitive regimes? We answer this question for a broad class of quantum spin chains, focusing on the SRE at index \(\alpha=1/2\). This value belongs to the replica regime in which a field-theory description is expected to be accessible \cite{HoshinoOshikawaAshida2026,RamirezRajabpour2026}, while still leading to a highly nontrivial microscopic problem. Because this regime lies below the boundary-sector transition expected at \(\alpha_c\), the structures found here provide a natural reference point for understanding other Rényi indices in the same sector, where the universal constants and crossover functions are expected to become \(\alpha\)-dependent \cite{Stephan2011,RamirezRajabpour2026,HoshinoOshikawaAshida2026}.

The physical models considered here include the anisotropic spin chain in a transverse field and its transverse-field Ising and XX limits \cite{Lieb1961,Katsura1962,Pfeuty1970,BarouchMcCoyDresden1970,BarouchMcCoy1971}, zero-field reductions related to Ising blocks, folded finite-range string Hamiltonians, and Cluster--Ising representatives with three-spin cluster interactions \cite{KramersWannier1941a,Smacchia2011}. The common structure behind these models is that, after the Jordan--Wigner transformation, they are described by finite-range quadratic fermionic Hamiltonians. In the spin representation, this includes ordinary nearest-neighbor chains as well as longer-range string interactions. The symbol language provides a compact way to encode these Hamiltonians, their Bogoliubov spectra, and their Gaussian correlation matrices \cite{JinKorepin2004,ItsJinKorepin2005,Its:2008,PeschelEisler2009,Verresen:2018,Verresen:2019}.

The primary challenge is that SRE is not an ordinary Gaussian observable. Even when the state is Gaussian, the Pauli amplitudes entering the SRE are highly nonlocal functions of the state. Thus, the usual solvability of quadratic Hamiltonians does not by itself make the SRE easy to compute. At \(\alpha=1/2\), the computation can be viewed as a large combinatorial sum over Pauli amplitudes; for generic matrices, the corresponding absolute-minor problem is computationally intractable \cite{Dyer1998,Gover2016}. The exact solvability found here is therefore not automatic from Gaussianity alone. It relies on a special algebraic structure in the spin chains under consideration.

To access both global finite-size data and boundary-sensitive subsystem data, we study two geometries. The first is a finite periodic chain of length \(L\), where the stabilizer entropy probes the full many-body ground state in a fixed fermion-parity sector \cite{DePasquale2008,DePasquale2009}. The second is a finite interval of length \(\ell\) embedded in the infinite chain, where the SRE probes a reduced Gaussian density matrix \cite{PeschelEisler2009}. These two settings are physically different. The periodic chain measures the global finite-size response of the state, while the subsystem geometry also contains boundary information associated with the way the interval cuts the infinite chain. This distinction is central to the results below.

\begin{figure*}[!ht]
\centering
\includegraphics[width=\linewidth]{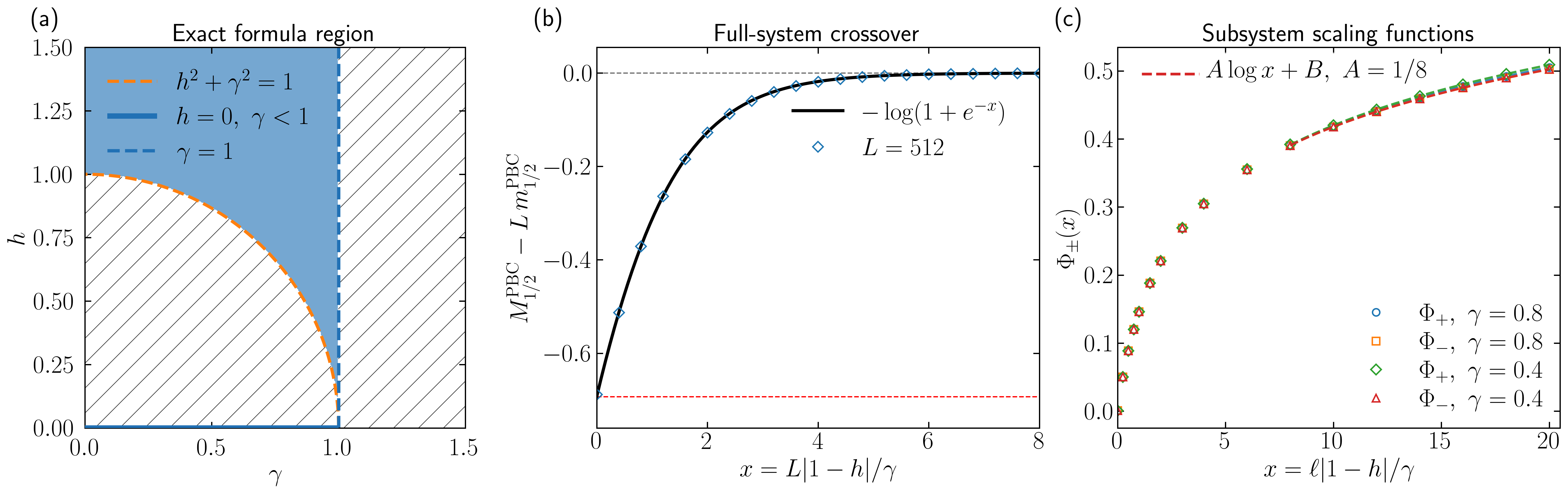}
\caption{
Universal scaling structures of the stabilizer Rényi entropy.
(a) Exact-formula region for the periodic-chain and subsystem analysis in the $(\gamma,h)$ plane. The orange dashed curve denotes $h^2+\gamma^2=1$, the black dashed line denotes $\gamma=1$, and the solid blue line along $h=0$ indicates the additional exact result for $\gamma<1$. The shaded blue region indicates the exact-formula regime, while the hatched regions are outside this regime.
(b) Full-system finite-size crossover. After subtracting the volume-law contribution $L\,m_{1/2}^{\rm PBC}$, the periodic-chain data for $L=512$ follow the universal crossover function $-\log(1+e^{-x})$ with $x=L|1-h|/\gamma$. The horizontal dashed lines mark the critical value $-\log 2$ and the massive limit $0$.
(c) Subsystem scaling functions. The extrapolated functions $\Phi_\pm(x)=\lim_{\ell\to\infty}\left[
\widetilde M_{\frac12}^{\rm XY,(\infty)}(\ell;h_\ell^\pm(x),\gamma)
-\widetilde M_{\frac12}^{\rm XY,(\infty)}(\ell;1,\gamma)
\right]$ are shown for two anisotropies, $\gamma=0.8$ and $\gamma=0.4$, as functions of $x=\ell|1-h|/\gamma$. The dashed curves show logarithmic-tail fits of the form $A\log x+B$ performed for $x>8$, with the expected asymptotic coefficient $A=1/8$ indicated in the panel.
}
\label{fig:main_results}
\end{figure*}

\subsection{Overview of results}

We summarize the main results of this work, with Fig.~\ref{fig:main_results} providing a compact overview of the exact parameter region, the full-system finite-size crossover, and the subsystem scaling functions. The first result is an exact finite-size solution for the SRE of the full periodic chain. Although the calculation is implemented through an elementary anisotropic free-fermion block, the result is not restricted to a single Hamiltonian. Through exact SRE correspondences, the same solution controls the Ising and XX limits, zero-field decimated chains, folded finite-range string Hamiltonians, and Cluster--Ising representatives. In the periodic geometry, the stabilizer entropy becomes a product over independent Bogoliubov-mode pairs. This turns a many-body Pauli-amplitude problem into an exactly solvable finite-size expression.

This exact formula gives direct access to the large-\(L\) scaling regimes beyond the critical point. In the massive regime, the SRE is governed by a volume-law density, while the finite-size corrections are exponentially small and controlled by the correlation length. This makes the departure from criticality analytically visible at the lattice level. More importantly, when the system is close to the field-tuned critical line, the finite-size behavior is governed by the scaling variable \(x=L|1-h|/\gamma\). The subleading contribution then becomes the universal crossover function \(-\ln(1+e^{-x})\), which describes how a finite periodic chain resolves the passage from the critical regime to the gapped regime. Its large-\(x\) limit gives the massive behavior, where the finite part vanishes, while its \(x\to0\) limit recovers the universal critical constant. Thus, the periodic result not only reproduces the known critical structure; it gives an exact interpolation between critical and noncritical scaling. This provides a lattice benchmark for the emerging QFT description of SRE and is consistent with the SRE--SR entropy correspondence for critical free-fermion chains \cite{HoshinoOshikawaAshida2026,RamirezRajabpour2026}.

The same exact formula also reveals a nonanalytic response of the large-scale SRE density. The density is continuous across the anisotropic critical line, but its derivative with respect to the transverse field jumps at \(h=1\). Thus, the SRE density develops a cusp at the transition. This gives a direct analytic signature of criticality at the level of the stabilizer entropy itself, not only in its finite-size corrections. Related singular responses have recently appeared in numerical and quantum Monte Carlo studies of many-body SRE, where derivatives and volume-law corrections reveal critical behavior \cite{DingYan2025}. Here, the cusp is obtained analytically from the exact lattice formula.

The endpoint \((h,\gamma)=(1,0)\), where the field-tuned critical line reaches the XX saturation point, is a nonuniform limit of the periodic scaling theory. This endpoint has long been recognized as a special crossover region of the XY chain~\cite{VaidyaTracy1978Crossover,dosSantosStinchcombe1981}, and finite-ring XX studies also emphasize the special role of this limit in periodic systems~\cite{DePasquale2009}. In the present setting, it is not obtained by taking the fixed-\(\gamma\) critical expansion and then sending \(\gamma\to0\); rather, the limits \(L\to\infty\) and \(\gamma\to0\) do not commute. The endpoint is resolved by a distinct double-scaling window in which \(\mu=L^2(h-1)\) and \(\eta=L^2\gamma\) are kept fixed. The \(L^{-2}\) scale reflects the quadratic dispersion at the saturation point, in contrast with the fixed-\(\gamma\) crossover controlled by \(L|1-h|/\gamma\). Thus, the point \((h,\gamma)=(1,0)\) defines a separate endpoint crossover regime. Under the Kramers--Wannier map for the Cluster--Ising representative, the cluster critical line is mapped to \(h=1\), \(\gamma=1-g_c\), so its midpoint \(g_c=1\) probes precisely this XX endpoint. This relation connects our finite-size scaling result at \(\alpha=1/2\) with recent transfer-matrix and infinite-MPS studies of second SRE in quantum spin chains, where the XX point also requires special care \cite{Hallam2026}.

The second main result concerns finite intervals of the infinite chain. In this setting, one first takes the large-scale limit and only then restricts the state to a finite subsystem. Therefore, the answer is not obtained by replacing \(L\) with \(\ell\) in the periodic formula. The interval has endpoints, and these endpoints carry boundary information. We derive an exact finite-\(\ell\) formula for the subsystem stabilizer entropy and use it to obtain the large-\(\ell\) regimes.

In the subsystem geometry, the leading term is again governed by a bulk volume-law density, but the finite part has a different physical origin: it is controlled by the endpoints of the interval. The exact finite-\(\ell\) object is a Pfaffian formula associated with a block-Toeplitz structure. The scalar XY symbol is used only to identify the common bulk density, the correlation length, and the local Fisher--Hartwig singularity; the boundary constants and crossover functions remain interval data of the Pfaffian problem. Closely related critical-to-massive crossover problems in the XY chain, such as emptiness formation probabilities and full counting statistics, have been analyzed through Toeplitz and block-Toeplitz asymptotics, leading in those cases to Painlev\'e V descriptions~\cite{AresViti2020EFP,AresRajabpourViti2021FCS}. In the massive regime, these endpoint effects produce a finite boundary contribution, with corrections exponentially suppressed by the correlation length. As the field-tuned critical line is approached, this boundary contribution becomes singular and develops the matching behavior \(\frac14\ln\xi^{-1}(h,\gamma)\), up to finite universal terms. Thus, away from criticality, the critical logarithm is not simply lost; it is reorganized into a correlation-length-dependent boundary contribution. More importantly, when the interval length is comparable to the correlation length, the subsystem behavior is governed by the scaling variable \(x=\ell/\xi\). Universal crossover functions then interpolate between the critical boundary regime and the massive boundary regime, giving an exact lattice description of how a finite interval resolves the passage from critical to noncritical scaling. The leading crossover is described by a single function \(\Phi_{\rm XY}(x)\), whose large-\(x\) tail is \(\frac18\ln x\). The dependence on the side of the transition and on the anisotropy does not enter this leading function; it appears only in the finite-size approach to the scaling limit, through the branch-dependent subleading correction. Finally, the subsystem correspondences are boundary-sensitive: shifted interval representatives have the same bulk density and critical point, but different boundary constants and crossover functions. Hence, the interval problem reveals boundary data that are invisible in the periodic geometry, in agreement with the boundary-sensitive CFT description of SRE \cite{HoshinoOshikawaAshida2026,HoshinoAshida2025Defects}.

The third main result is the extension of this scaling theory to other quantum spin chains through exact SRE correspondences. For folded finite-range string Hamiltonians, lattice permutations and gauge transformations decompose the Gaussian correlation matrix into independent elementary blocks, so the SRE is inherited additively from the solved representative. For the Cluster--Ising family, the mechanism is instead Kramers--Wannier duality: the relevant Pauli algebra is mapped to that of a dual quantum spin chain, preserving the Pauli probability distribution that defines the SRE. After the appropriate sector or interval-algebra matching, the Cluster--Ising SRE is therefore obtained from the corresponding dual problem. This connection complements recent studies of SRE in spin models and infinite matrix product states, where Cluster--Ising-type chains provide natural settings for nonstabilizerness, criticality, and subleading spectral data \cite{Viscardi2026SpinModels,Hallam2026}.

An important refinement of these correspondences appears when one compares periodic chains with the subsystem. The subsystem correspondences are more delicate than the periodic ones. On a ring, shifts of the free-fermion symbol can be absorbed by cyclicity and sector matching. On an interval, by contrast, such shifts change how the endpoints are embedded in the infinite-chain Gaussian kernel. Consequently, different interval embeddings of the same bulk free-fermion kernel can have the same volume-law density and the same critical logarithmic coefficient, but different boundary constants and crossover functions. We describe this effect through shifted interval representatives. These representatives are not new bulk universality classes; they encode how the subsystem boundary cuts the Gaussian degrees of freedom. This provides a microscopic realization of the boundary sensitivity predicted by field-theory and defect-based approaches to SRE \cite{HoshinoOshikawaAshida2026,HoshinoAshida2025Defects}.

Taken together, these results show that stabilizer Rényi entropy can be treated as an exact scaling observable in a broad family of solvable many-body spin chains. The exact finite-size formulas provide a lattice benchmark beyond purely numerical or field-theory approaches. They determine the massive, critical, and crossover regimes in both periodic and subsystem geometries, and identify the universal finite-size, boundary, normalization, and endpoint structures that distinguish these regimes. Through exact correspondences, the same scaling theory also extends to folded finite-range string Hamiltonians and Cluster--Ising representatives. In this way, the work connects microscopic quantum spin chain solvability with the emerging QFT understanding of stabilizer entropy \cite{HoshinoOshikawaAshida2026,RamirezRajabpour2026,HoshinoAshida2025Defects}.

The remainder of the paper is organized as follows. In Sec.~\ref{sec:gaussian-reduction}, we introduce the stabilizer Rényi entropy conventions and the Gaussian reduction used throughout the paper. In Sec.~\ref{sec:hamiltonians-symbols}, we present the spin Hamiltonians, their symbol description, and the two geometries considered here. In Sec.~\ref{sec:generating-functions-correspondences}, we establish the SRE correspondences generated by decimation and Kramers--Wannier duality. In Sec.~\ref{sec:pbc-exact-results}, we solve the periodic geometry and derive the large-\(L\) regimes. In Sec.~\ref{sec:subsystem-exact-results}, we solve the subsystem geometry and derive the large-\(\ell\) regimes, including the shifted interval representatives. We conclude with a discussion of the relation to QFT predictions and possible extensions.

\section{Gaussian reduction of stabilizer Rényi entropies}
\label{sec:gaussian-reduction}

In this section, we establish the stabilizer Rényi entropy conventions used throughout the paper and reduce the problem, for the real fermionic Gaussian states considered here, to a minor expansion of a single Gaussian correlation matrix. The reduction takes a particularly simple form at Rényi index \(\alpha=1/2\): the entropy is governed by the sum of the absolute values of all minors of the Gaussian \(G\)-matrix. Although this absolute-minor functional is computationally intractable for a generic matrix, the XY correlation matrices studied below have additional algebraic structure that allows exact evaluation in the regimes identified in subsequent sections. Recent numerical approaches to the nonstabilizerness of fermionic Gaussian states use sampling and determinantal structures to access large systems~\cite{Tarabunga2024MPS,DingYan2025,Tarighi2024,Rajabpour2025a,Collura2026}; here we instead exploit special XY Toeplitz and finite-periodic structures to obtain exact formulas.

For an arbitrary \(N\)-qubit density matrix \(\rho\), we use the Pauli expansion \(\rho=2^{-N}\sum_{P\in\mathcal P_N}\Tr(P\rho)\,P\), with \(\mathcal P_N=\{I,X,Y,Z\}^{\otimes N}\). By Pauli orthogonality, \(\mathbb P_\rho(P)=|\Tr(P\rho)|^2/[2^N\Tr(\rho^2)]\) is a normalized probability distribution. For pure states, the standard stabilizer Rényi entropy is the Rényi entropy of this Pauli distribution shifted by the stabilizer value \(N\ln2\)~\cite{Leone2022a}. For Rényi index \(\alpha>0\), \(\alpha\neq1\), we use the purity-normalized extension
\begin{equation}
\label{eq:StabRenyi_def}
M_{\alpha}(\rho)=\frac{1}{1-\alpha}
\ln\left[\frac{\sum_{P\in\mathcal P_N}
|\Tr(P\rho)|^{2\alpha}}{2^N\Tr(\rho^2)}\right].
\end{equation}
The limit \(\alpha\to1\) defines the corresponding stabilizer Shannon entropy.

When \(\rho\) is pure, Eq.~\eqref{eq:StabRenyi_def} reduces to the standard pure-state SRE. The normalization by \(\Tr(\rho^2)\) also removes the trivial contribution associated with the rank of stabilizer-code projectors. Indeed, if \(\rho^{\rm stab}=\Pi_S/\Tr\Pi_S\), with nonzero Pauli expectation values uniformly supported on a stabilizer subgroup, then \(M_\alpha(\rho^{\rm stab})=0\). This property is important for the limits encountered below, including fully polarized states and monomial-symbol limits in which the Gaussian subsystem matrix reduces to a signed partial permutation. Although mixed-state SRE admits inequivalent resource-theoretic extensions~\cite{HaugPiroli2023,LeoneBittel2024}, Eq.~\eqref{eq:StabRenyi_def} is the natural convention for the present Gaussian problem: it agrees with the pure-state definition, vanishes on the stabilizer limits realized by the chain, and preserves the correlator/minor formulation developed below.

We now specialize to the real fermionic Gaussian states that arise in the quantum spin chains studied in this work. Their relevant two-point data are encoded in the real \(G\)-matrix: \(G_{jk}:=\Tr\!\left[\rho\,(c_j^\dagger-c_j)(c_k^\dagger+c_k)\right]\). This description applies to pure Gaussian ground states of real quadratic Hamiltonians, and also to reduced or mixed Gaussian density matrices whenever the restricted contractions retain the same real Gaussian structure. For a pure Gaussian state, \(G_{jk}\) is simply the corresponding expectation value in the Gaussian state vector. A standard paired parametrization of such states is reviewed in Refs.~\cite{Tarighi2024,Rajabpour2025a}.

In this setting, Wick's theorem expresses the Pauli-string expectation values entering Eq.~\eqref{eq:StabRenyi_def} as determinants of submatrices of \(\mathbf G\). Thus the Gaussian stabilizer Rényi entropies are controlled by sums of powers of minors of \(\mathbf G\). For an \(N\times N\) matrix \(\mathbf M\), we define
\begin{equation}
\label{eq:Det_beta_def}
\mathbf{Det}_{\beta}(\mathbf M)
:=
\sum_{k=0}^{N}
\sum_{\substack{I,J\subseteq[N]\\ |I|=|J|=k}}
\left|\det\mathbf M[I,J]\right|^\beta ,
\end{equation}
where \([N]=\{1,\ldots,N\}\), and \(\mathbf M[I,J]\) denotes the submatrix with row set \(I\) and column set \(J\). For the Gaussian states described above, the purity-normalized stabilizer Rényi entropy becomes~\cite{Leone2022b}
\begin{equation}
\label{eq:M_alpha_ratio}
M_\alpha(\mathbf G)=\frac{1}{1-\alpha}
\ln\left[\frac{\mathbf{Det}_{2\alpha}(\mathbf G)}{\mathbf{Det}_{2}(\mathbf G)}
\right],\qquad\alpha>0,\quad \alpha\neq1 .
\end{equation}
The denominator is the normalization of the Pauli probability distribution. By the Cauchy--Binet identity~\cite{Knill2014CauchyBinet}, it can be written as \(\mathcal N:=\mathbf{Det}_{2}(\mathbf G)=\det(\mathbf I+\mathbf G^\top\mathbf G)\). For pure real Gaussian states on the full system, \(\mathbf G^\top\mathbf G=\mathbf I\), and therefore \(\mathcal N=2^N\). For reduced or mixed Gaussian states, this simplification is absent, and the full normalization must be kept.

The main exact results of this work concern the special value \(\alpha=1/2\). At this index, Eq.~\eqref{eq:M_alpha_ratio} becomes
\begin{equation}
\label{eq:M_half_mixed}
M_{\frac12}(\mathbf G)=2\ln \mathbf{Det}_{1}(\mathbf G)-2\ln\mathcal N.
\end{equation}

Thus, the central quantity in the remainder of the paper is \(\mathbf{Det}_{1}(\mathbf G)\), the sum of the absolute values of all minors of \(\mathbf G\). This functional is highly nontrivial: for a generic rational matrix \(\mathbf A\), computing \(\mathbf{Det}_{1}(\mathbf A)\) exactly is \(\#P\)-hard. The hardness follows from its relation to exact zonotope-volume computation, where the volume is expressed as a sum of absolute values of maximal minors~\cite{Dyer1998,Gover2016}. A formal statement and proof are given in Appendix~\ref{app:proof-theorem1}. The exact results below, therefore, rely on additional algebraic structure of the correlation matrices: in the periodic geometry, the absolute-minor sum reduces to a finite-mode product, while for infinite-chain subsystems, it admits a Pfaffian representation.

\section{Hamiltonians, symbols, and geometries}\label{sec:hamiltonians-symbols}

We study stabilizer R\'enyi entropies of spin-chain Gaussian states in two geometries, shown in Fig.~\ref{fig:setups}. The first is a finite periodic spin chain of even length \(L\), where the SRE probes the full many-body ground state. The second is a finite interval \(\Lambda_\ell=\{1,\dots,\ell\}\) embedded in the infinite chain, where the SRE probes the reduced Gaussian density matrix of a subsystem. These two setups are built from the same bulk spin-chain data, but they lead to different finite Gaussian matrices. In the periodic geometry, the symbol is sampled on a sector-dependent momentum grid. In the subsystem geometry, the continuum phase defines a Toeplitz kernel that is truncated to the interval. This distinction is important throughout the paper: the periodic geometry captures global finite-size effects, while the interval geometry also contains boundary data associated with the endpoints of the subsystem.

\subsection{Spin Hamiltonians and Laurent symbols}
We organize the finite-range spin chains considered in this work through their
Laurent symbols. Consider the spin Hamiltonian~\cite{Its:2008,Verresen:2018,Verresen:2019}
\begin{equation}
\label{eq:Hspin-general}
\mathrm{H}_{\hat f}=-\sum_{j=1}^{L}
\sum_{m=-R}^{R}A_m\,O_j^{(m)},
\end{equation}
where, for the finite ring, we take \(L\) to be even and identify the sites modulo \(L\). The operators \(O_j^{(m)}\) are defined in Fig.~\ref{fig:laurent-symbol-dictionary}. The \(Z\)-strings appearing in \(O_j^{(m)}\) are the Jordan--Wigner strings: they make the spin Hamiltonian string-local for \(|m|>1\), while the corresponding fermionic Hamiltonian remains quadratic. For the infinite-chain subsystem geometry, the same bulk Hamiltonian is understood directly on \(\mathbb Z\), and the interval \(\Lambda_\ell\) is introduced only after taking the scaling limit.

\begin{figure}[!ht]
    \centering
        \includegraphics[width=0.8\linewidth]{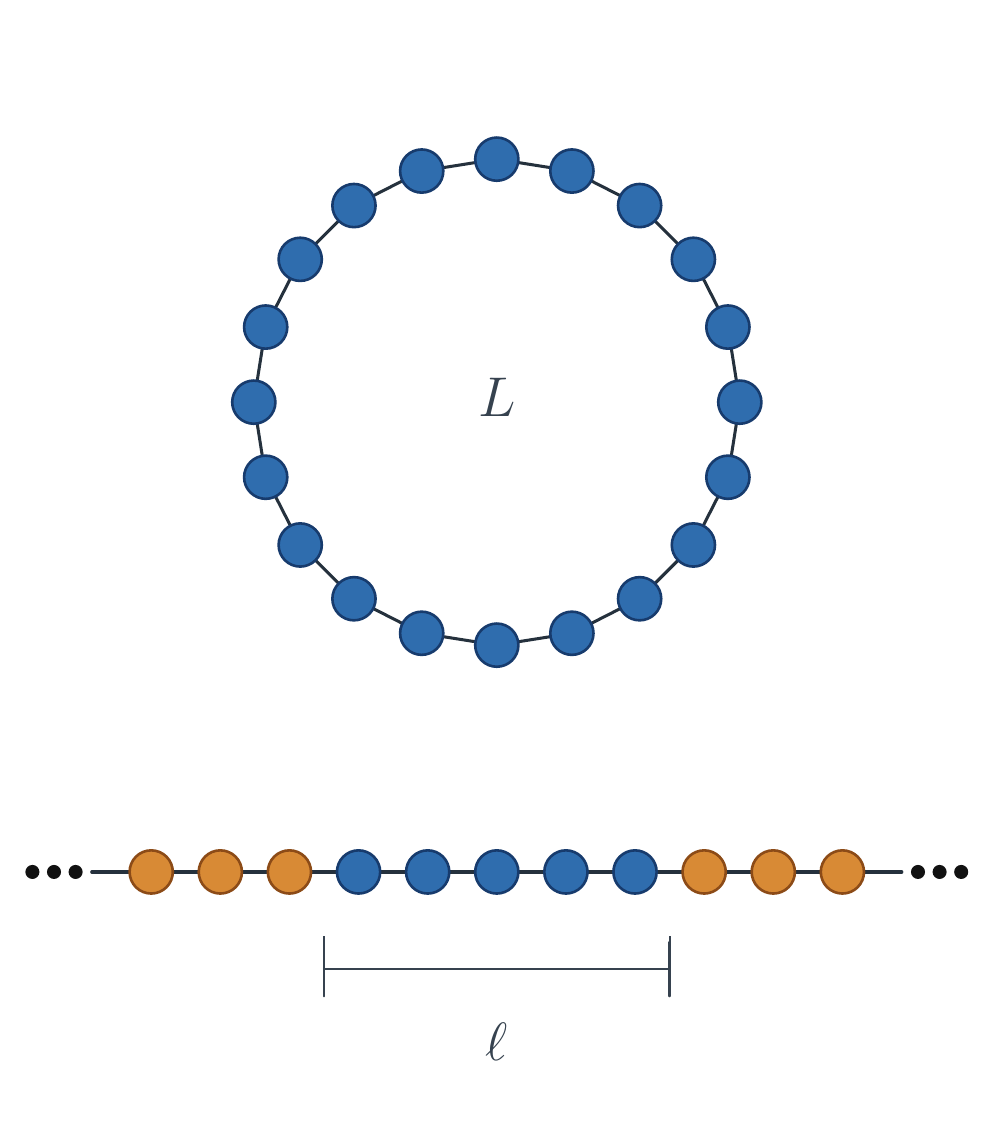}
        \put(-180,185){(a)}
        \put(-180,70){(b)}
    \caption{Lattice configurations. (a) Finite periodic chain with even length \(L\). (b) Infinite chain restricted to a contiguous subsystem of length \(\ell\).}
    \label{fig:setups}
\end{figure}

\begin{figure*}[!ht]
\centering
\includegraphics[width=\linewidth]{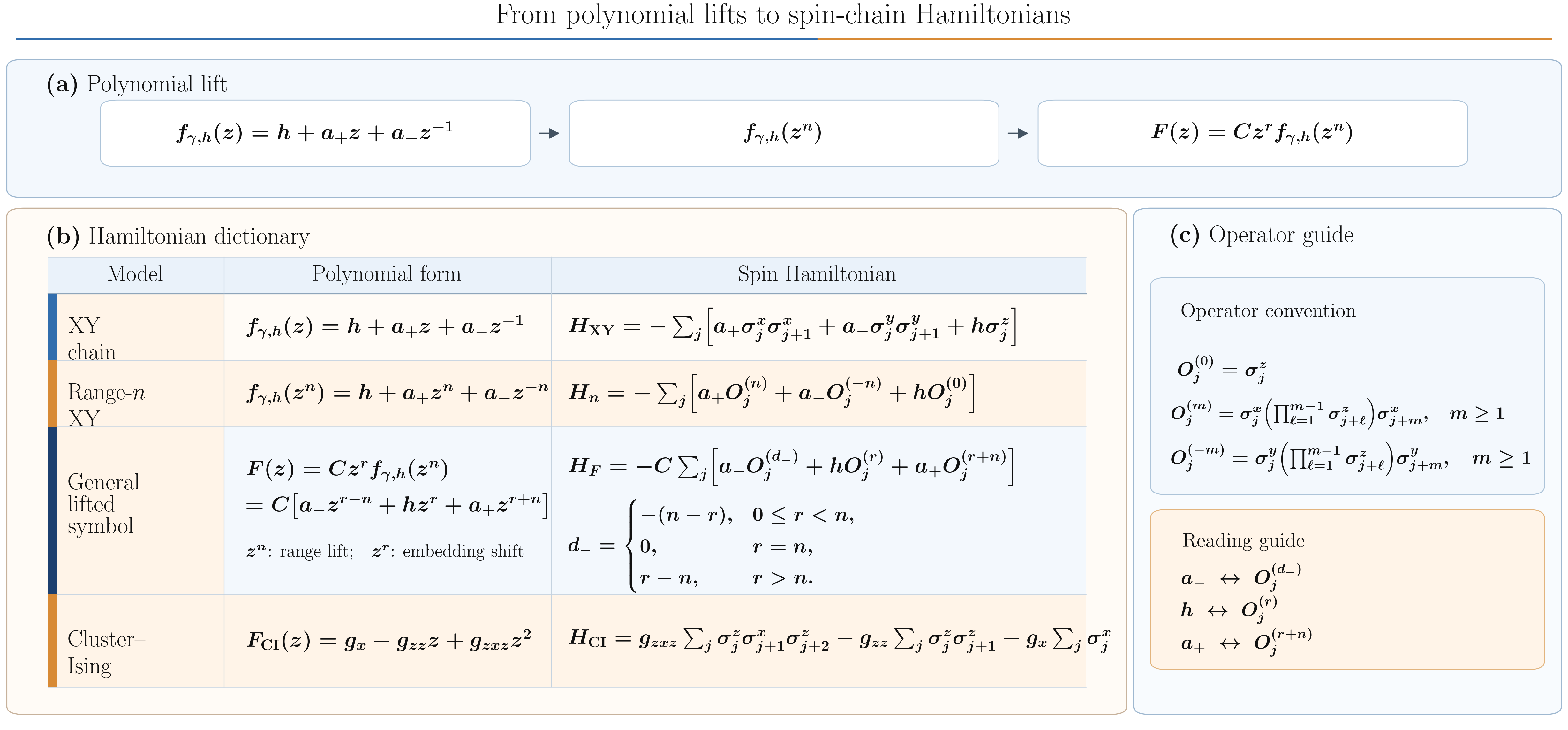}
\caption{Dictionary between Laurent-polynomial symbols and spin-chain Hamiltonians. Panel (a) shows the successive transformations \(f_{\gamma,h}(z)\mapsto f_{\gamma,h}(z^n)\mapsto
F(z)=Cz^r f_{\gamma,h}(z^n)\). Panel (b) lists the representative symbol families used in the paper and their corresponding spin-chain Hamiltonians,
with \(a_{\pm}=(1\pm\gamma)/2\). Panel (c) defines the string operators \(O_j^{(m)}\): positive indices denote \(x\)-strings, negative indices denote \(y\)-strings, and \(O_j^{(0)}=\sigma_j^z\).}
\label{fig:laurent-symbol-dictionary}
\end{figure*}

After fixing the fermion-parity sector, the Jordan--Wigner, Fourier, and Bogoliubov transformations diagonalize the Hamiltonian as \(\mathrm{H}_\eta=\sum_{q\in K}\varepsilon_{\hat f}(q)(\eta_q^\dagger\eta_q-\frac12)\), with \(\varepsilon_{\hat f}(q)=2|\hat f(e^{iq})|\). With the convention of Eq.~\eqref{eq:Hspin-general}, the Laurent symbol is~\cite{Its:2008,Verresen:2018,Verresen:2019}
\begin{equation}
\label{eq:general-laurent-symbol}
\hat f(z)=\sum_{m=-R}^{R}t_m z^m,
\qquad z=e^{iq}.
\end{equation}
The Gaussian ground state is fixed by the phase \(\omega_{\hat f}(q):=\hat f(e^{iq})/|\hat f(e^{iq})|\), and zeros of \(\hat f\) on the unit circle correspond to gapless points. The dictionary between the symbols and the spin Hamiltonians used in this work is summarized in Fig.~\ref{fig:laurent-symbol-dictionary}. For a detailed dictionary of the full models used in the work, see Appendix~\ref{app:hamiltonian-symbol-dictionary}.

\subsection{Geometry-dependent Gaussian matrices}

We now define the Gaussian matrices associated with the two geometries. In both cases, the input for the stabilizer calculation is the Gaussian correlation matrix determined by the phase \(\omega_{\hat f}\), but the finite matrix is constructed differently in the periodic and subsystem setups.

For periodic spin boundary conditions, the JW fermions split into two parity sectors. Equivalently, one works with the Neveu--Schwarz grid \(K_{\rm NS}=\{2\pi(k+\frac12)/L\}_{k=0}^{L-1}\) or the Ramond grid \(K_{\rm R}=\{2\pi k/L\}_{k=0}^{L-1}\). Once the sector \(\nu\in\{{\rm NS},{\rm R}\}\) is fixed, the finite periodic Gaussian matrix is
\begin{equation}
\label{eq:G-PBC-general}
G_{nm}^{(\hat f),\nu}=\frac{1}{L}
\sum_{q\in K_\nu}\omega_{\hat f}(q)\,e^{iq(n-m)},
\qquad 1\le n,m\le L .
\end{equation}
Thus, the NS/R distinction enters the correlation matrix only through the momentum grid.


For the infinite chain, there is no finite-size sector choice. The continuum phase defines the Fourier coefficients~\cite{Its:2008,PeschelEisler2009}
\begin{equation}\label{eq:infinite-fourier-coefficients}
g_r^{\hat f}=\frac{1}{2\pi}\int_0^{2\pi}\omega_{\hat f}(q)e^{-irq}\,dq,
\end{equation}
and the subsystem matrix is the Toeplitz truncation
\begin{equation}
\label{eq:G-infinite-general}
\mathbf G_\ell^{(\hat f),(\infty)}=\left[g_{a-b}^{\hat f}\right]_{1\le a,b\le \ell}.
\end{equation}
Equations~\eqref{eq:G-PBC-general} and~\eqref{eq:G-infinite-general} are the two Gaussian matrices used throughout the paper: the first for the finite periodic chain and the second for the infinite-chain subsystem.

\subsection{The anisotropic XY chain}
The main exact results presented below are obtained for the nearest-neighbor anisotropic XY chain~\cite{Lieb1961,Katsura1962,Pfeuty1970,BarouchMcCoyDresden1970,BarouchMcCoy1971}. Within the Laurent-symbol convention introduced in Eq.~\eqref{eq:general-laurent-symbol}, this model is specified by \(R=1\), \(t_0=h\), and \(t_{\pm1}=(1\pm\gamma)/2\), as illustrated in Fig.~\ref{fig:laurent-symbol-dictionary}. Accordingly, we denote the associated phase symbol by 
\[\omega_{\gamma,h}(q):=\frac{f_{\gamma,h}(q)}{|f_{\gamma,h}(q)|}.\]
Several important limits and regions of the XY parameter space will be used throughout the analysis. The transverse-field Ising and XX chains correspond to the special lines \(\gamma=1\) and \(\gamma=0\), respectively. The phase structure in the \((\gamma,h)\)-plane is summarized in Fig.~\ref{fig:xy_phase_diagram}: the lines \(h=\pm1\) with \(\gamma\neq0\) are the anisotropic critical lines, the segment \(\gamma=0\), \(|h|<1\), is the XX critical line, and the curve \(h^2+\gamma^2=1\) is the disorder circle separating the corresponding regions.
\begin{figure*}[!ht]
    \centering
    \includegraphics[width=1\textwidth]{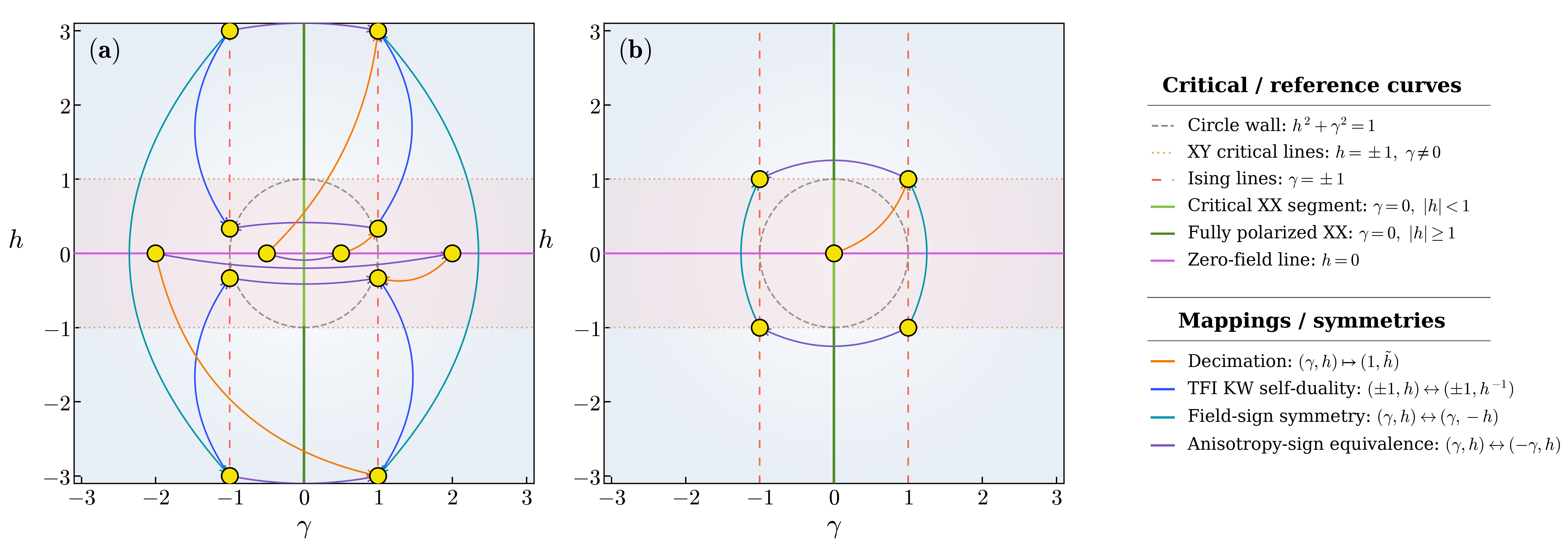}
    \caption{Parameter space of the anisotropic XY chain. The dashed curve
    \(h^2+\gamma^2=1\) separates oscillatory and ordered regimes. The lines
    \(h=\pm1\), \(\gamma=\pm1\), and \(\gamma=0\), \(|h|<1\), denote the
    anisotropic critical, Ising, and critical XX lines, respectively. The
    zero-field line is special: parity decimation maps \((\gamma,0)\) to the
    TFI point \((1,\widetilde h)\), with
    \(\widetilde h=(1-\gamma)/(1+\gamma)\). Arrows denote decimation,
    Kramers--Wannier duality, and unitary sign equivalences. The marked points
    organize two orbits. All points within a given orbit have the same SRE,
    with the appropriate sector or interval-algebra matching.}
    \label{fig:xy_phase_diagram}
\end{figure*}

For this XY building block, we use the following shorthand for the two
Gaussian matrices introduced above. In the finite periodic geometry, with even
length \(L=2M\), we write
\[
\mathbf G^{\rm PBC}(h,\gamma)
:=\mathbf G^{(f_{\gamma,h}),{\rm NS}},
\]
where the right-hand side is the sector-resolved matrix of
Eq.~\eqref{eq:G-PBC-general} evaluated in the NS sector. In the subsystem
geometry, we write
\[
\mathbf G_\ell^{(\infty)}(h,\gamma)
:=\mathbf G_\ell^{(f_{\gamma,h}),(\infty)},
\]
where the right-hand side is the Toeplitz interval matrix of
Eq.~\eqref{eq:G-infinite-general} for the infinite-chain XY ground state.

For the periodic XY chain, the finite-size ground state is selected by
comparing the NS and R sector vacua~\cite{DePasquale2009,DePasquale2008}. For
even \(L\), the ground state is NS throughout \(h^2+\gamma^2>1\). Inside the
disorder circle the sector is size dependent; on \(h=0\), including the XX
point, it is NS for \(L=0 \!\!\mod 4\) and R for \(L=2 \!\!\mod 4\). Unless
stated otherwise, the periodic product formulas below are written in the NS
sector. For \(L=2M\), the NS momenta are paired as \(\pm\alpha_j\), with \(\alpha_j=\frac{(2j-1)\pi}{L}\), and \(j=1,\dots,M\).

\section{SRE correspondences across quantum spin chains}
\label{sec:generating-functions-correspondences}

This section explains how the XY building block controls a broader class of
quantum spin chains. The exact SRE formulas derived in the next two sections
can be transferred to these chains through two mechanisms. The first is a
matrix-level decimation. After a lattice permutation and diagonal gauge
transformations, the fermionic correlation matrix \(\mathbf G\) can split into
independent blocks. The absolute-minor generating function then factorizes
over these blocks, and the SRE becomes additive. The second mechanism is
Kramers--Wannier duality~\cite{KramersWannier1941a,Kogut1979}. This duality
does not arise from a block decomposition of \(\mathbf G\). Instead, it acts
directly on the Pauli algebra and preserves the Pauli probability distribution
that defines the SRE, provided the periodic sectors or interval algebras are
matched appropriately.

These two mechanisms generate the SRE correspondences used throughout the
paper. Matrix decimation relates lifted or folded finite-range symbols to
ordinary XY blocks, while Kramers--Wannier duality relates the Cluster--Ising
and TFI representatives to the appropriate XY representatives. The relevant
spin-chain representatives and parameter-space orbits were introduced in
Figs.~\ref{fig:laurent-symbol-dictionary} and~\ref{fig:xy_phase_diagram}.
Detailed block constructions are given in Appendix~\ref{app:xy_decimation};
the cluster--Ising/Kramers--Wannier map and sector matching are discussed in
Appendix~\ref{app:cluster_XY_duality}; and numerical checks are collected in
Appendix~\ref{app:numerical-decimation-checks}.

A further point, crucial for the subsystem results, is that the same symbolic
correspondence can be implemented differently in the two geometries. In the
periodic geometry, monomial shifts of the Laurent symbol can be absorbed by
cyclicity, gauge transformations, and sector matching. In the interval
geometry, there is no cyclic identification of sites, so the endpoints retain
the shift. Consequently, correspondences that are equivalent on a ring can
produce distinct boundary-sensitive interval representatives.


\subsection{Generating functions and block additivity}
\label{subsec:generating-functions-additivity}

Let \(\mathbf G\) be an \(L\times L\) real Gaussian correlation matrix. For
\(\beta>0\), we introduce the minor-generating polynomial
\begin{equation}
\label{eq:Pbeta-generating-function}
P_{\beta,\mathbf G}(u):=\sum_{p=0}^{L} u^p
\sum_{\substack{I,J\subseteq[L]\\ |I|=|J|=p}}
\left|\det \mathbf G[I,J]\right|^\beta .
\end{equation}
Thus \(P_{\beta,\mathbf G}(1)=\mathbf{Det}_{\beta}(\mathbf G)\). The
correspondences below can therefore be established at the level of the
generating polynomial \(P_{\beta,\mathbf G}\), before specializing to the SRE
index \(\alpha=1/2\).

The decimation mechanism is purely matrix-theoretic. Suppose that, after a
lattice permutation \(\mathbf P\) and a diagonal unitary gauge \(\mathbf D\),
the correlation matrix takes the block form
\begin{equation}
\label{eq:block-reduction-main}
\mathbf D^{-1}\mathbf P^\top\mathbf G\,\mathbf P\mathbf D
=
\mathbf G_1\oplus\mathbf G_2\oplus\cdots\oplus\mathbf G_k .
\end{equation}
Permutations only relabel minors, while diagonal gauges multiply minors by
phases of unit modulus. Since \(P_{\beta,\mathbf G}\) depends only on absolute
values of minors, these operations leave it invariant. The direct sum then
factorizes the generating polynomial:
\begin{equation}
\label{eq:Pbeta-block-factorization}
P_{\beta,\mathbf G}(u)
=
\prod_{a=1}^{k}P_{\beta,\mathbf G_a}(u).
\end{equation}
Evaluating at \(u=1\) gives the additivity of the Gaussian SRE,
\begin{equation}
\label{eq:SRE-block-additivity}
M_\alpha(\mathbf G)
=
\sum_{a=1}^{k}M_\alpha(\mathbf G_a),
\qquad
\alpha>0,\quad \alpha\neq1 .
\end{equation}
In particular, if the reduction produces \(k\) identical blocks \(\widetilde{\mathbf G}\), then \(M_\alpha(\mathbf G)=k\,M_\alpha(\widetilde{\mathbf G})\). This is the algebraic core of all decimation identities in Fig.~\ref{fig:xy-decimation-summary}.

\subsection{SRE correspondences}

\paragraph{Periodic-chain correspondences.}
We first consider the ring geometry. The nearest-neighbor XY chain is encoded by \(f_{\gamma,h}(z)\), while the finite-range representatives used below are generated by lifted symbols \(F(z)=Cz^r f_{\gamma,h}(z^n)\). On a periodic chain, the folding \(z\mapsto z^n\) reorganizes the lattice into residue classes. After the corresponding permutation and diagonal gauge, the correlation matrix decomposes into identical XY blocks. Hence the SRE follows from block additivity, Eq.~\eqref{eq:SRE-block-additivity}, rather than from a new calculation.

The zero-field XY line gives the simplest example. The identity \(f_{\gamma,0}(z)=\frac{1+\gamma}{2}z^{-1}f_{1,\widetilde h}(z^2)\), with \(\widetilde h=(1-\gamma)/(1+\gamma)\), shows that the line \(h=0\) is a parity lift of TFI. With the compatible periodic-sector matching,
\begin{equation}
\label{eq:SRE-PBC-zero-field-XY-TFI}
M_\alpha^{{\rm XY},{\rm PBC}}(L;0,\gamma)=2M_\alpha^{{\rm TFI},{\rm PBC}}\left(\frac L2;1,\widetilde h\right).
\end{equation}
More generally, for \(F(z)=Cz^r f_{\gamma,h}(z^n)\) and \(g=\gcd(L,n)\), the periodic block decomposition gives
\begin{equation}
\label{eq:SRE-PBC-lifted-symbol}
M_\alpha^{F,{\rm PBC}}(L)
=
g\,M_\alpha^{{\rm XY},{\rm PBC}}
\left(\frac Lg;h,\gamma\right).
\end{equation}

The cluster--Ising representative is transferred by Kramers--Wannier duality rather than by a \(\mathbf G\)-block decomposition~\cite{Smacchia2011,Hallam2026}. With \(F_{\rm CI}(z)=(g_x+g_{zxz})z f_{\gamma,-h}(z^{-1})\), \(\gamma=(g_x-g_{zxz})/(g_x+g_{zxz})\), and \(h=g_{zz}/(g_x+g_{zxz})\), the Pauli probability distribution is preserved after the appropriate sector matching. Therefore
\begin{equation}
\label{eq:SRE-PBC-cluster-XY}
M_\alpha^{{\rm CI},{\rm PBC}}(L;g_{zxz},g_{zz},g_x)
=
M_\alpha^{{\rm XY},{\rm PBC}}(L;h,\gamma).
\end{equation}
The TFI self-duality is the \(\gamma=1\) specialization of the same Pauli-algebra mechanism. The preservation of SRE under Pauli-string dualities, including the Kramers--Wannier map \(h\mapsto h^{-1}\), was noted in Ref.~\cite{Tarabunga2023}. In the present notation, using \(f_{1,h}(z)=hz f_{1,h^{-1}}(z^{-1})\), one obtains
\begin{equation}
\label{eq:SRE-PBC-TFI-self-duality}
M_\alpha^{{\rm TFI},{\rm PBC}}(L;1,h)
=
M_\alpha^{{\rm TFI},{\rm PBC}}
\left(L;1,\frac1h\right).
\end{equation}
The sector matching and duality conventions behind Eqs.~\eqref{eq:SRE-PBC-cluster-XY} and~\eqref{eq:SRE-PBC-TFI-self-duality} are detailed in Appendix~\ref{app:cluster_XY_duality}.

\begin{figure*}[t]
\centering
\includegraphics[width=0.75\textwidth]{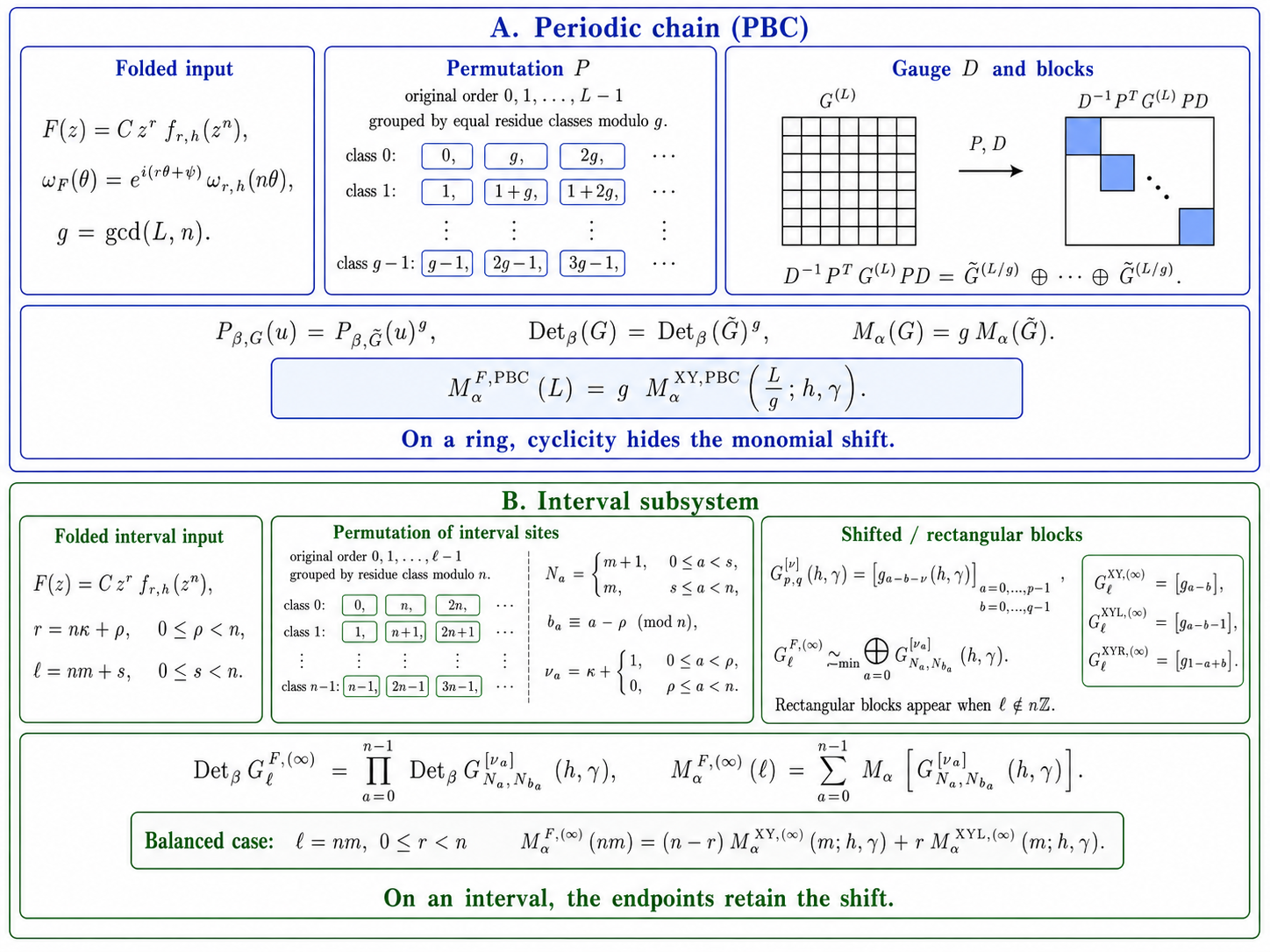}
\caption{Summary of the exact decimation structure and SRE correspondences used throughout the work. \textbf{A. Periodic chain:} correlation matrices decompose into independent \(\mathbf G\)-blocks after site permutations and diagonal gauge transformations, leading to factorization of the generating function and additive SRE. \textbf{B. Interval subsystem:} the same folding acts on Toeplitz truncations, but the endpoints retain the monomial shift, producing ordinary, shifted, and, when \(\ell\notin n\mathbb Z\), rectangular blocks; consequently, the determinant sums factorize, and the subsystem SRE becomes additive.}
\label{fig:xy-decimation-summary}
\end{figure*}

\paragraph{Subsystem correspondences.}
For an infinite chain restricted to an interval, there are no NS/R sectors and no cyclic identification of sites. The relevant objects are finite Toeplitz truncations of the infinite-volume kernel. This changes the role of monomial shifts: on a ring they can be absorbed by cyclicity and sector matching, whereas on an interval the endpoints remember the shift. We therefore keep the Hamiltonian labels, such as \({\rm XY}\), \({\rm TFI}\), and \({\rm CI}\), separate from the interval-block labels \(X\in\{{\rm XY},{\rm XYR},{\rm XYL}\}\), writing \(M_\alpha^{X,(\infty)}(\ell;h,\gamma):=M_\alpha[\mathbf G_\ell^{X,(\infty)}(h,\gamma)]\). The ordinary interval is \((\mathbf G_\ell^{{\rm XY},(\infty)})_{ab}=g_{a-b}(h,\gamma)\), while the shifted representatives are \((\mathbf G_\ell^{{\rm XYL},(\infty)})_{ab}=g_{a-b-1}(h,\gamma)\) and \((\mathbf G_\ell^{{\rm XYR},(\infty)})_{ab}=g_{1-a+b}(h,\gamma)\), with \(a,b=0,\ldots,\ell-1\).

For the zero-field XY line, parity decimation gives two TFI interval blocks rather than two identical periodic blocks. For even \(\ell=2m\),
\begin{equation}
\label{eq:SRE-subsystem-zero-field-XY-TFI}
\begin{split}
M_\alpha^{{\rm XY},(\infty)}(2m;0,\gamma)
=&\,
M_\alpha^{{\rm TFI},(\infty)}(m;1,\widetilde h)\\
&+
M_\alpha^{{\rm TFI},(\infty)}
\left(m;1,\frac1{\widetilde h}\right).
\end{split}
\end{equation}
The second block follows from \(g_r(\widetilde h^{-1},1)=g_{1-r}(\widetilde h,1)\). For odd \(\ell\), the two parity subsectors have unequal sizes and produce rectangular unbalanced blocks; the complete construction is given in Appendix~\ref{app:xy_decimation}.

For a lifted symbol \(F(z)=Cz^r f_{\gamma,h}(z^n)\), the folding \(z\mapsto z^n\) decomposes the interval into residue classes, while the monomial factor fixes which classes are shifted. If \(\ell=nm\) and \(\bar r\equiv r\ {\rm mod}\ n\), with \(0\le\bar r<n\), then
\begin{equation}
\label{eq:SRE-subsystem-lifted-symbol}
\begin{split}
M_\alpha^{F,(\infty)}(nm)
=&\,
(n-\bar r)M_\alpha^{{\rm XY},(\infty)}(m;h,\gamma)\\
&+
\bar r\,M_\alpha^{{\rm XYL},(\infty)}(m;h,\gamma).
\end{split}
\end{equation}
The unshifted case \(\bar r=0\) gives \(n\) ordinary XY interval blocks. The cases \(\ell\notin n\mathbb Z\), where the residue classes have unequal sizes, are listed in Appendix~\ref{app:xy_decimation}.

For the cluster--Ising subsystem, Kramers--Wannier duality selects the right-shifted interval representative. With the same map \(\gamma=(g_x-g_{zxz})/(g_x+g_{zxz})\) and \(h=g_{zz}/(g_x+g_{zxz})\),
\begin{equation}
\label{eq:SRE-subsystem-cluster-XYR}
M_\alpha^{{\rm CI},(\infty)}
(\ell;g_{zxz},g_{zz},g_x)
=
M_\alpha^{{\rm XYR},(\infty)}(\ell;h,\gamma).
\end{equation}
Finally, TFI self-duality in the interval geometry is represented by the same right-shifted block:
\begin{equation}
\label{eq:SRE-subsystem-TFI-self-duality}
M_\alpha^{{\rm TFI},(\infty)}
\left(\ell;1,\frac1h\right)
=
M_\alpha^{{\rm XYR},(\infty)}(\ell;h,1).
\end{equation}
Thus, in the periodic geometry, shifts are absorbed by cyclicity, gauge transformations, and sector matching, whereas in the subsystem geometry they survive as the interval representatives \({\rm XYR}\) and \({\rm XYL}\). The Pfaffian formulas used to evaluate the corresponding \(P_\ell^{(\infty)}(u)\) are developed in Sec.~\ref{sec:subsystem-exact-results} and Appendix~\ref{app:subsystem_toeplitz_pfaffian}.

The next two sections use these correspondences in the two geometries separately. Section~\ref{sec:pbc-exact-results} solves the periodic XY block through a finite-mode product, while Sec.~\ref{sec:subsystem-exact-results} solves the interval representatives through Pfaffian formulas.
\section{Finite periodic chain: exact \(M_{\frac12}\) and scaling regimes}\label{sec:pbc-exact-results}

This section gives the finite-size solution for the full periodic XY chain at \(\alpha=1/2\). We work in the even-length NS sector and focus on the principal stability chamber where the absolute-minor generating function of the Gaussian matrix \(\mathbf G^{\mathrm{PBC}}(h,\gamma)\) factorizes into independent Bogoliubov-mode contributions. The product formula is the exact finite-\(L\) input. We then use it to discuss the physical response of the stabilizer entropy near the anisotropic critical line \(h=1\), before extracting the massive, critical, and crossover scaling regimes. The endpoint \((h,\gamma)=(1,0)\), where the anisotropic line meets the XX saturation point, requires a separate double-scaling analysis.

\subsection{Stability chamber and exact product formula}
\paragraph{Stability chamber.}
Consider a periodic chain of length \(L=2M\) in the NS sector, with Gaussian matrix \(\mathbf G^{\mathrm{PBC}}(h,\gamma)\). At \(\alpha=1/2\), the relevant object is the absolute-minor generating polynomial \(P_{1,\mathbf{G}^{\mathrm{PBC}}}(u):=P_L^{\mathrm{PBC}}(u;h,\gamma)\) defined in Eq.~\eqref{eq:Pbeta-generating-function}. Because of the absolute values, this polynomial is analytic only after a sign branch for the relevant minors has been fixed.

For finite \(L\), the principal product branch is detected by the patterned minor
\(\mathcal D_M(h,\gamma):=\det\mathbf G^{\mathrm{PBC}}[I_M,J_M]\), with
\(I_M=\{1,\dots,M\}\) and \(J_M=\{1,M+2,\dots,2M\}\). The equation
\(\mathcal D_M(h,\gamma)=0\) may contain several components. We denote by
\(\gamma=\gamma_*(h;L)\) the component first reached when increasing \(\gamma\) from the small-\(\gamma\) side of the principal chamber outside the circle \(h^2+\gamma^2=1\). Other zero components correspond to different sign chambers and are not part of the product branch used below.

We define the corresponding finite-size stability chamber as
\begin{equation}
\label{eq:PBC-stability-region-results}
\mathcal R_L^{\mathrm{PBC}}=
\Big\{(h,\gamma)\in\mathbb R_{\ge0}^{2}:h^2+\gamma^2>1,\;0\le \gamma<\gamma_*(h;L)\Big\}.
\end{equation}
The circle \(h^2+\gamma^2=1\) is the physical lower boundary of the XY symbol, whereas the selected curve
\(\gamma=\gamma_*(h;L)\) is a finite-size sign wall of the absolute-minor polynomial, not a new Hamiltonian critical line. Inside \(\mathcal R_L^{\mathrm{PBC}}\), the fixed sign pattern is coherent, and the absolute-minor polynomial reduces to the product formula below. The corresponding exact-formula region is summarized in Fig.~\ref{fig:main_results}(a). The zero-locus analysis, chamber interpretation, and numerical checks are given in Appendix~\ref{app:pbc-stability}. As \(L\) increases, this selected sign wall drifts toward the line \(\gamma=1\).

\paragraph{Product formula.} Let \(\alpha_j=\frac{\pi}{L}(2j-1)\), \(j=1,\dots,M\), be the positive NS momenta, and define \(\Lambda_j(h,\gamma):=\frac{h+\gamma+(1-\gamma)\cos\alpha_j}{[(h+\cos\alpha_j)^2+\gamma^2\sin^2\alpha_j]^{1/2}}\). Then, throughout \(\mathcal R_L^{\mathrm{PBC}}\),
\begin{equation}
\label{eq:PBC-product-results}
P_L^{\mathrm{PBC}}(u;h,\gamma)=
\prod_{j=1}^{M}\left[1+u^2+2u\,\Lambda_j(h,\gamma)\right].
\end{equation}
Thus, the exponentially large absolute-minor sum collapses to a product over the \(M=L/2\) paired Bogoliubov modes.

Evaluating Eq.~\eqref{eq:PBC-product-results} at \(u=1\) and using \(L=2M\), the prefactor \(2^M\) is absorbed by the pure-state normalization of \(M_{\frac12}\). Hence, throughout \(\mathcal R_L^{\mathrm{PBC}}\),
\begin{equation}
\label{eq:Mhalf-PBC-product-results}
M_{\frac12}^{\mathrm{PBC}}(L;h,\gamma)=2\sum_{j=1}^{M}
\ln\left[\frac{1+\Lambda_j(h,\gamma)}{2}\right] .
\end{equation}
Equivalent rapidity, determinant, companion-matrix, and exterior-power formulations of Eq.~\eqref{eq:PBC-product-results} are collected in Appendix~\ref{app:pbc-equivalent-forms}.

\subsection{Field response and Large-scale cusp}

The product formula gives direct access to the finite-size entropy density \(m_L(h,\gamma):=M_{\frac12}^{\mathrm{PBC}}(L;h,\gamma)/L\) and to its field response \(\chi_L(h,\gamma):=\partial_h m_L(h,\gamma)\). At finite \(L\), both are smooth inside a fixed sign chamber, so the critical behavior is rounded.


The nonanalyticity appears only after taking the scaling limit. We denote the limiting density by
\begin{equation}\label{eq:pbc-continuum-density}
m_{\frac12}^{\mathrm{PBC}}(h,\gamma)=\frac1\pi\int_0^\pi
\ln\left[\frac{1+\Lambda_{h,\gamma}(q)}{2}\right]dq.
\end{equation}
This function is continuous at \(h=1\), but its field derivative jumps:
\begin{equation}
\label{eq:pbc-cusp-jump}
\partial_h m_{\frac12}^{\mathrm{PBC}}(1^+,\gamma)
-
\partial_h m_{\frac12}^{\mathrm{PBC}}(1^-,\gamma)
=
-\frac{1}{\gamma}.
\end{equation}
Thus the SRE density develops a cusp, rather than a discontinuity, whose sharpness grows as \(\gamma\to0^+\).

The cusp is controlled by the endpoint \(q=\pi\). Writing \(h=1+\delta\) and \(q=\pi-t\), the singular part of the integrand behaves as \(-\frac12\log(\delta^2+\gamma^2t^2)\). Integration over \(t\) then gives the nonanalytic term \(-|h-1|/(2\gamma)\), which is continuous but has different left and right derivatives. The formal statement and proof are given in Appendix~\ref{app:pbc-density-cusp}. The crossover term \(-\log(1+e^{-x})\), derived below, is the finite-\(L\) remnant of this Large-scale cusp.


\subsection{Large-\(L\) regimes}

We now summarize the large-\(L\) behavior following from the product formula. We keep \(0<\gamma\le1\) fixed and remain on the stable product branch; the singular endpoint \(\gamma=0\) is treated separately below. The continuum density is Eq.~\eqref{eq:pbc-continuum-density}. Near the anisotropic critical line \(h=1\), the relevant length scale is the correlation length \(\xi(h,\gamma)\), which diverges at criticality and is finite away from it. In the present normalization, its inverse behaves as \(\xi^{-1}(h,\gamma)\sim |1-h|/\gamma\) close to \(h=1\). Thus the natural finite-size variable is
\(x_L=L/\xi(h,\gamma)\simeq L|1-h|/\gamma\), which compares the ring size with the correlation length. For the crossover window we use \(h_{L,\sigma}(x)=1+\sigma\gamma x/L\), with \(x=O(1)\) and \(\sigma=\pm1\).

\paragraph{SRE large regime formula.} The periodic-chain SRE has three fixed-\(\gamma\) regimes:
\begin{widetext}
\begin{equation}
\label{eq:pbc-Mhalf-three-regimes}
M_{\frac12}^{\mathrm{PBC}}(L;h,\gamma)=
\begin{cases}
L\,m_{\frac12}^{\mathrm{PBC}}(h,\gamma)+O(e^{-x_L}),
& h\neq1\ \mathrm{fixed}\quad (x_L\to\infty),\\[2mm]
L\,m_{\frac12}^{\mathrm{PBC}}(1,\gamma)
-\ln2+\dfrac{\pi}{12L}
+\dfrac{a_3(\gamma)}{L^3}+O(L^{-5}),
& h=1\quad (x_L=0),\\[2mm]
L\,m_{\frac12}^{\mathrm{PBC}}(h_{L,\sigma}(x),\gamma)
-\ln(1+e^{-x})
+\dfrac{B_1^{(\sigma)}(x;\gamma)}{L}
+O(L^{-2}),
& h=h_{L,\sigma}(x),\quad x=O(1).
\end{cases}
\end{equation}
\end{widetext}

The first line is the gapped, or massive, regime: \(h\neq1\) is kept fixed as \(L\to\infty\), so the correlation length remains finite and \(L\gg\xi\), equivalently \(x_L\to\infty\). The entropy is a volume law with density \(m_{\frac12}^{\mathrm{PBC}}(h,\gamma)\). Since the system size is much larger than the correlation length, the finite part is exponentially suppressed; in the crossover language, \(-\ln(1+e^{-x_L})\to0\).

The second line is the anisotropic critical line. Here \(\xi\to\infty\), so the system is scale invariant in the scaling limit. The endpoint singularity at \(q=\pi\) gives the finite part \(-\ln2\), with no \(\ln L\) term, in agreement with the CFT prediction \(c_\alpha=\frac{\ln\alpha}{2(\alpha-1)}\) for \(\alpha\le4\)~\cite{HoshinoOshikawaAshida2026,RamirezRajabpour2026}. The correction \(a_3(\gamma)/L^3\) is kept because \(a_3(\gamma)\propto\gamma^{-2}\), signaling that the fixed-\(\gamma\) expansion is not uniform as \(\gamma\to0\).

The third line is the finite-size crossover regime. In this window, \(h\) approaches the critical line as \(1/L\), so the correlation length grows proportionally to the system size and \(x=L/\xi=O(1)\). The scaling function \(-\ln(1+e^{-x})\) interpolates between the critical finite part \(-\ln2\) at \(x=0\) and the gapped finite part \(0\) in the large-\(x\) tail. This universal collapse is shown in Fig.~\ref{fig:main_results}(b). Thus the crossover describes how the finite-size system resolves the approach from a critical, scale-invariant regime to a gapped regime with finite correlation length.

The explicit expressions for \(a_3(\gamma)\), \(B_1^{(\sigma)}(x;\gamma)\), and the higher-order terms are given in Appendix~\ref{app:pbc-regimes-Mhalf}.

\paragraph{XX endpoint.}
The fixed-\(\gamma\) expansion is not uniform at the XX saturation endpoint \((h,\gamma)=(1,0)\). At this point \(\Lambda_j(1,0)=1\) for every finite NS momentum, so \(M_{\frac12}^{\mathrm{PBC}}(L;1,0)=0\). Thus, the endpoint has the expected stabilizer value. By contrast, for any fixed \(0<\gamma\le1\), the anisotropic critical expansion along \(h=1\) gives \(M_{\frac12}^{\mathrm{PBC}}(L;1,\gamma)=L m_{\frac12}^{\mathrm{PBC}}(1,\gamma)-\ln2+o(1)\), where \(m_{\frac12}^{\mathrm{PBC}}\) denotes the purity-normalized density. Hence, after subtracting the extensive term, the finite part jumps from \(-\ln2\) on the fixed-\(\gamma\) critical line to \(0\) at the XX endpoint. This is not a contradiction: it reflects the noncommutativity of the limits \(L\to\infty\) and \(\gamma\to0\).

The endpoint must therefore be resolved in a separate double-scaling window. We zoom into \((h,\gamma)=(1,0)\) by keeping \(\mu=L^2(h-1)\), \(\eta=L^2\gamma\) fixed, equivalently \(h=1+\mu/L^2\) and \(\gamma=\eta/L^2\). The \(L^{-2}\) scale follows from the quadratic endpoint dispersion at \(\gamma=0\): near \(q=\pi-t\), one has \(h+\cos(\pi-t)\simeq h-1+t^2/2\). This differs from the fixed-\(\gamma\) anisotropic scaling, where the endpoint region is controlled by \(\sqrt{(h-1)^2+\gamma^2t^2}\). Thus the XX endpoint is not described by the fixed-\(\gamma\) massive, critical, or crossover regimes of Eq.~\eqref{eq:pbc-Mhalf-three-regimes}; it defines a distinct endpoint crossover associated with the saturation point.

For \(\mu>0\), or for \(\mu=0\) with \(\eta>0\), the trajectory stays above the lower chamber wall \(h^2+\gamma^2=1\). On the stable product branch considered here, the purity-normalized SRE has the endpoint scaling form
\begin{equation}\label{eq:pbc-XX-endpoint-scaling}
M_{\frac12}^{\mathrm{PBC}}\left(L;1+\frac{\mu}{L^2},\frac{\eta}{L^2}\right)=\mathcal F_{\rm XX}(\mu,\eta)-\frac{\eta}{2L}+O(L^{-2}),
\end{equation}
where
\begin{equation}
\mathcal F_{\rm XX}(\mu,\eta)=2\ln\left[\frac{\cosh\left(\frac{\sqrt{2(\mu+\eta)}}{2}\right)}{\cosh\left(\frac{\sqrt{2\mu}}{2}\right)}\right].
\end{equation}
It vanishes for \(\eta=0\), consistently with the stabilizer saturation value. Along the anisotropic critical line \(h=1\), where \(\mu=0\), it reduces to \(\mathcal F_{\rm XX}(0,\eta)=2\ln\cosh\left(\frac{\sqrt{2\eta}}{2}\right)\). Therefore, the XX endpoint is exceptional: although the massive, critical, and fixed-\(\gamma\) crossover regimes have no \(\ln L\) term, their finite parts do not describe the endpoint itself. The endpoint is controlled instead by the \(L^{-2}\) window where the lower branch wall meets the XX saturation point.

Finally, this endpoint has a direct implication for the cluster--Ising representative. Under the Kramers--Wannier correspondence discussed above, the cluster critical line maps to \(h=1\), \(\gamma=1-g_c\). Hence, its midpoint \(g_c=1\) is precisely the XX endpoint. In particular, approaching the midpoint as \(g_c=1-\eta/L^2\) gives \(M_{\frac12}^{\mathrm{CI,PBC}}(L;g_c)=2\ln\cosh\left(\frac{\sqrt{2\eta}}{2}\right)
-\frac{\eta}{2L}+O(L^{-2})\) on the matched periodic sector. This explains why the finite part near \(g_c=1\) should not be compared with the fixed-\(\gamma\) Ising value \(-\ln2\). Rather, the cluster midpoint probes the nonuniform XX endpoint scaling. This provides a finite-size counterpart to the cluster-chain critical behavior studied in Ref.~\cite{Hallam2026}.

\paragraph{Lifted symbols and cluster representatives.}
The periodic regimes above transfer to the lifted symbols discussed in Sec.~\ref{sec:generating-functions-correspondences}. For \(F(z)=Cz^r f_{\gamma,h}(z^n)\), let \(g_F=\gcd(L,n)\), \(L_F=L/g_F\), and \(x_F=L_F\mu(h,\gamma)\simeq L_F|1-h|/\gamma\) near \(h=1\). The block identity gives \(M_{\frac12}^{F,\mathrm{PBC}}(L;h,\gamma)=g_F M_{\frac12}^{\mathrm{XY},\mathrm{PBC}}(L_F;h,\gamma)\). Defining \(h_{F,\sigma}(x):=1+\sigma\gamma x/L_F\), one obtains
\begin{widetext}
\begin{equation}\label{eq:pbc-F-three-regimes}
M_{\frac12}^{F,\mathrm{PBC}}(L;h,\gamma)=\begin{cases}
L\,m_{\frac12}^{\mathrm{PBC}}(h,\gamma)
-g_F\ln(1+e^{-x_F})+\cdots,
& h\neq1\ \mathrm{fixed}\quad (x_F\to\infty),\\[2mm]
L\,m_{\frac12}^{\mathrm{PBC}}(1,\gamma)
-g_F\ln2
+\dfrac{g_F^2\pi}{12L}
+\dfrac{g_F^4a_3(\gamma)}{L^3}
+O(L^{-5}),
& h=1,\\[2mm]
L\,m_{\frac12}^{\mathrm{PBC}}(h_{F,\sigma}(x),\gamma)
-g_F\ln(1+e^{-x})
+\dfrac{g_F^2B_1^{(\sigma)}(x;\gamma)}{L}
+O(L^{-2}),
& h=h_{F,\sigma}(x),\quad x=O(1).
\end{cases}
\end{equation}
\end{widetext}
Thus, the volume-law density is unchanged, while the finite universal terms are multiplied by the number \(g_F\) of periodic blocks. In particular, the critical finite part is \(-g_F\ln2\), and the crossover constant is \(-g_F\ln(1+e^{-x})\). The algebraic corrections acquire the powers of \(g_F\) shown in Eq.~\eqref{eq:pbc-F-three-regimes} because each XY block has length \(L_F=L/g_F\).

The cluster--Ising representative is not obtained by a \(\mathbf G\)-block decimation. Its periodic correspondence follows instead from Kramers--Wannier duality, together with the sector matching reviewed in Appendix~\ref{app:cluster_XY_duality}. Under the map
\(h=g_{zz}/(g_x+g_{zxz})\) and \(\gamma=(g_x-g_{zxz})/(g_x+g_{zxz})\), the periodic SRE is identified with the corresponding XY result,
\(M_{\frac12}^{\mathrm{CI},\mathrm{PBC}}(L;g_{zxz},g_{zz},g_x)
=
M_{\frac12}^{\mathrm{PBC}}(L;h,\gamma)\).
Along the cluster critical line \((g_{zxz},g_{zz},g_x)=(g_c,2,2-g_c)\), this gives \(h=1\) and \(\gamma=1-g_c\). Therefore, for fixed \(g_c\neq1\), the cluster asymptotics are obtained directly from the fixed-\(\gamma\) critical line of Eq.~\eqref{eq:pbc-Mhalf-three-regimes}.

The point \(g_c=1\) is exceptional. It maps to the XX endpoint \((h,\gamma)=(1,0)\), where the fixed-\(\gamma\) expansion is nonuniform. The correct scaling is therefore inherited from Eq.~\eqref{eq:pbc-XX-endpoint-scaling}, with
\(\mu=L^2\left[g_{zz}/(g_x+g_{zxz})-1\right]\) and
\(\eta=L^2(g_x-g_{zxz})/(g_x+g_{zxz})\). In particular, along the cluster critical line one has \(\mu=0\) and \(\eta=L^2(1-g_c)\). Thus a nontrivial endpoint limit is obtained by taking \(g_c=1-\eta/L^2\), for which the leading finite part becomes \(\mathcal F_{\rm XX}(0,\eta)=2\ln\cosh(\sqrt{2\eta}/2)\). This separates the ordinary cluster critical regime at fixed \(g_c\neq1\) from the XX-endpoint scaling regime reached as \(g_c\to1\).
\section{Infinite-chain subsystem: exact \(M_{\frac12}\) and scaling regimes}\label{sec:subsystem-exact-results}
We now turn to finite intervals of the infinite XY chain. In contrast with the periodic geometry, where one fixes a finite ring and a fermionic sector from the outset, the subsystem problem is defined by first taking the large-scale-limit ground state and only then restricting its Gaussian kernel to a finite interval. The relevant data are the Toeplitz coefficients \(g_r^{f_{\gamma,h}}:=g_r(h,\gamma)\) defined in Eq.~\eqref{eq:infinite-fourier-coefficients}, and the ordinary interval block is the principal truncation \(\mathbf G_\ell^{\rm XY,(\infty)}=[g_{a-b}(h,\gamma)]_{a,b=0}^{\ell-1}\). The absolute-minor sum entering \(M_{\frac12}\) is therefore not a product over modes; it is represented instead by a finite Pfaffian.

We first derive the ordinary \({\rm XY}\) interval formula and its large-\(\ell\) regimes. The shifted representatives required by the subsystem correspondences of Sec.~\ref{sec:generating-functions-correspondences} are then treated as boundary embeddings of the same infinite-volume kernel. This separation is useful because the bulk density and the critical logarithm are common, while the finite constants and crossover functions depend on how the interval endpoints cut the Gaussian degrees of freedom.

\subsection{Stability chamber and Pfaffian formula}
\paragraph{Stability chamber.}
Because the absolute-minor sum contains absolute values, it is analytic only after a sign branch for the relevant minors has been fixed. For even \(\ell=2m\), we select the principal branch using the patterned Toeplitz minor
\(\mathcal D_m^{\rm Pf}(h,\gamma):=\det\mathbf G_{2m}^{\rm XY,(\infty)}[I_m,J_m]\), with
\(I_m=\{0,\dots,m-1\}\) and \(J_m=\{m-1,\dots,2m-2\}\) in the zero-based Toeplitz convention. Among the possible zero components of \(\mathcal D_m^{\rm Pf}\), we denote by \(\gamma=\gamma_*^{\rm Pf}(h;\ell)\) the first component reached from the small-\(\gamma\) side of the principal branch outside the circle \(h^2+\gamma^2=1\). 

We define the corresponding subsystem stability chamber by
\begin{equation}
\label{eq:subsystem-stability-region-results}
\mathcal R_\ell^{\rm Pf}=\left\{(h,\gamma)\in\mathbb R_{\geq0}^2:
h^2+\gamma^2>1,\quad 0\leq\gamma<\gamma_*^{\rm Pf}(h;\ell)
\right\}.
\end{equation}
The circle \(h^2+\gamma^2=1\) is the physical lower boundary of the XY symbol, whereas the selected curve
\(\gamma=\gamma_*^{\rm Pf}(h;\ell)\) is a finite-size sign wall of the absolute-minor polynomial, not a new Hamiltonian critical line. Inside \(\mathcal R_\ell^{\rm Pf}\), the polynomial remains on the principal analytic branch. This region is summarized in Fig.~\ref{fig:main_results}(a). For odd \(\ell\), where no balanced split exists, the same branch is obtained by direct minor tracking or by continuity from neighboring even truncations. The zero-minor walls, their finite-\(\ell\) drift toward the Ising line, and the numerical Pfaffian checks are discussed in Appendix~\ref{app:subsystem_toeplitz_pfaffian}.

\paragraph{Pfaffian formula.} For the ordinary interval, the absolute-minor generating function is the specialization of Eq.~\eqref{eq:Pbeta-generating-function} to \(\beta=1\) and \(\mathbf G=\mathbf G_\ell^{\rm XY,(\infty)}\). We denote it by \(P_\ell^{(\infty)}(u;h,\gamma)\), so that \(P_\ell^{(\infty)}(1;h,\gamma)=\mathbf{Det}_1\!\left(\mathbf G_\ell^{\rm XY,(\infty)}(h,\gamma)\right)\). On the sign-stable branch connected to the large-field region, the full absolute-minor generating function is encoded by the finite Pfaffian
\begin{equation}
\label{eq:main-subsystem-pfaffian-XY}
P_\ell^{(\infty)}(u;h,\gamma)=u^\ell\operatorname{Pf}\mathbf K_\ell^{\rm XY}(u;h,\gamma).
\end{equation}
The explicit \(2\ell\times2\ell\) skew kernel \(\mathbf K_\ell^{\rm XY}\) is given in Appendix~\ref{app:subsystem_toeplitz_pfaffian}. Equation~\eqref{eq:main-subsystem-pfaffian-XY} is the subsystem analogue of the periodic product formula: the exponentially large sum over minors is reduced to a structured finite object, but in the interval geometry this object is a Pfaffian rather than a product over discrete Bogoliubov momenta. Its large-\(\ell\) behavior is therefore controlled by Toeplitz/Pfaffian asymptotics rather than by a finite-mode Euler--Maclaurin expansion. In the smooth massive regime the relevant mechanism is of Szegő--Widom type, whereas at \(h=1\) the endpoint \(q=\pi\) produces a Fisher--Hartwig singularity~\cite{FisherHartwig1968,Widom1974,Widom1976,BottcherSilbermann1999,DeiftItsKrasovsky2011,JinKorepin2004,ItsJinKorepin2005,Its:2008,PeschelEisler2009}. Related block-Toeplitz determinants with discontinuous \(2\times2\) symbols also arise in finite-range fermionic chains with broken reflection or charge-conjugation symmetry, where the symbol discontinuities generate logarithmic contributions to interval entanglement asymptotics~\cite{AresQueiroz2015}. Closely related critical-to-massive crossover problems in the XY chain, including emptiness formation probabilities and full counting statistics, have been analyzed through Toeplitz and block-Toeplitz asymptotics and lead in those cases to Painlev\'e V/Fredholm descriptions ~\cite{AresViti2020EFP,AresRajabpourViti2021FCS}. The present observable is different at the microscopic level, since a completed Pfaffian kernel encodes the absolute-minor sum, but the asymptotic mechanism is analogous.

\paragraph{Reference and purity normalizations.}
Before turning to the large-\(\ell\) behavior, we fix the normalization conventions. This point is important for the interpretation of subsystem results. From the microscopic point of view, the natural object is the purity-normalized SRE, because the interval is described by a reduced density matrix. However, this normalization also includes the Gaussian mixedness of the interval and can cancel part of the critical contribution carried by the absolute-minor numerator~\cite{HoshinoOshikawaAshida2026}. For this reason, it is useful to keep track separately of a \(2^\ell\)-reference normalization, which isolates the numerator and makes the Fisher--Hartwig/CFT logarithm visible.

For a Gaussian subsystem matrix \(\mathbf G_\ell\), we define, for arbitrary \(\alpha>0\), \(\alpha\neq1\), the \(2^\ell\)-normalized numerator
\begin{equation}
\label{eq:main-subsystem-Mtilde-general}
\widetilde M_{\alpha}(\mathbf G_\ell)=\frac{1}{1-\alpha}\ln\left[\frac{\mathbf{Det}_{2\alpha}\!\left(\mathbf G_\ell\right)}{2^\ell}\right],
\end{equation}
and the corresponding purity-normalized subsystem SRE is \(M_{\alpha}(\mathbf G_\ell)=\widetilde M_{\alpha}(\mathbf G_\ell)+\frac{S_2(\rho_\ell)}{1-\alpha}\). Thus, the \(2^\ell\)-normalized quantity is the natural object for exposing the Fisher--Hartwig logarithm carried by the absolute-minor numerator, whereas the purity-normalized SRE is the physical mixed-state quantity associated with the reduced Gaussian density matrix. We therefore first analyze the scaling of
\begin{equation}\label{eq:main-subsystem-SRE-XY}
\widetilde M_{\frac12}^{\rm XY,(\infty)}(\ell;h,\gamma)=2\ln\operatorname{Pf}\mathbf K_\ell^{\rm XY}(1;h,\gamma)-2\ell\ln2 .
\end{equation}
This expression follows by setting \(u=1\) in Eq.~\eqref{eq:main-subsystem-pfaffian-XY} and inserting the result into Eq.~\eqref{eq:main-subsystem-Mtilde-general}. We then discuss how this scaling is modified when the second Rényi entanglement entropy \(S_2(\rho_\ell)\) is added to obtain the purity-normalized subsystem SRE.

\subsection{Large-\(\ell\) regimes}
We now summarize the large-\(\ell\) behavior of the Pfaffian formula for the SRE~\eqref{eq:main-subsystem-SRE-XY}. In this limit, the interval becomes locally translationally invariant away from its endpoints, so the leading extensive term is fixed by the continuum XY symbol. This is the reason for using the same bulk density that appears in the periodic scaling limit: endpoints can change \(O(1)\) terms, but they cannot change the entropy density of a long interval. We denote the corresponding bulk density by \(m_{\frac12}^{\rm bulk}(h,\gamma):=m_{\frac12}^{\rm PBC}(h,\gamma)\), with \(m_{\frac12}^{\rm PBC}\) defined in Eq.~\eqref{eq:pbc-continuum-density}. This bulk term carries the Large-scale cusp at the anisotropic critical line \(h=1\), and therefore its field derivative has the same one-sided jump derived in Sec.~\ref{sec:pbc-exact-results}. The genuinely interval-dependent information starts at order one: a massive boundary constant, a Fisher--Hartwig logarithm at criticality, and a finite crossover function when the interval length is comparable to the correlation length.

Although the exact subsystem formula is Pfaffian, its bulk entries are built from the same infinite-volume XY Fourier coefficients as the continuum symbol. We use the scalar XY symbol only to identify the local bulk density, the correlation length, and the Fisher--Hartwig singularity at \(q=\pi\). The constants and crossover functions below are still defined by the interval Pfaffian problem; they are not obtained by replacing the Pfaffian formula by a scalar product formula.

In the gapped region, the relevant length scale is the correlation length \(\xi(h,\gamma)\), which is finite away from the field-tuned critical line and diverges as \(h\to1\). For fixed \(0<\gamma\le1\), its inverse behaves near criticality as \(\xi^{-1}(h,\gamma)\sim |1-h|/\gamma\). The interval crossover is therefore controlled by \(x=\ell/\xi(h,\gamma)\). We distinguish three large-\(\ell\) regimes: the off-critical regime \(h\neq1\) fixed, where \(x\to\infty\); the critical regime \(h=1\), where \(x=0\); and the crossover window \(x=O(1)\), parametrized by \(h_\ell^\sigma(x)=1+\sigma\gamma x/\ell\), with \(\sigma=\pm1\).

\paragraph{SRE large regime formula.} For fixed \(0<\gamma\le1\), the \(2^\ell\)-normalized subsystem SRE behaves as
\begin{widetext}
\begin{equation}
\label{eq:main-subsystem-XY-three-regimes}
\widetilde M_{\frac12}^{\rm XY,(\infty)}(\ell;h,\gamma)
=\begin{cases}
\ell m_{\frac12}^{\rm bulk}(h,\gamma)+2C_{\rm I}^{P,{\rm XY}}(h,\gamma)+O(e^{-\ell/\xi(h,\gamma)}),
& h\neq1\ \text{fixed},\\[0.45em]
\ell m_{\frac12}^{\rm bulk}(1,\gamma)-\dfrac14\ln\ell+2C_{\rm II}^{P,{\rm XY}}(\gamma)+o(1),
& h=1,\\[0.45em]
\ell m_{\frac12}^{\rm bulk}(h_\ell^\sigma(x),\gamma)-\dfrac14\ln\ell+2C_{\rm II}^{P,{\rm XY}}(\gamma)
+2\Phi_{\rm XY}(x)+O(\ell^{-1/2}),
& h=h_\ell^\sigma(x).
\end{cases}
\end{equation}
\end{widetext}

The three subleading structures in Eq.~\eqref{eq:main-subsystem-XY-three-regimes} have distinct roles. The constant \(2C_{\rm I}^{P,{\rm XY}}\) is the finite boundary contribution of a gapped interval. It is the subsystem analogue of a boundary entropy: it depends on the finite cut of the infinite Gaussian state, but it does not grow with \(\ell\). In the smooth-symbol regime, it comes from the Szegő--Widom part of the block-Toeplitz/Pfaffian asymptotics; its explicit representation is given in Appendix~\ref{sec:subsystem_asymptotics}.

As the critical line is approached, this constant becomes singular in precisely the way needed to match the Fisher--Hartwig regime: \(2C_{\rm I}^{P,{\rm XY}}(h,\gamma)=\frac14\ln \xi^{-1}(h,\gamma)+2C_{\rm II}^{P,{\rm XY}}(\gamma)+o(1)\), as \(h\to1^\sigma\).

At \(h=1\), the correlation length diverges, and the endpoint \(q=\pi\) of the symbol becomes Fisher--Hartwig singular~\cite{FisherHartwig1968,DeiftItsKrasovsky2011,JinKorepin2004,ItsJinKorepin2005,Its:2008}. This produces the universal numerator logarithm \(-\frac14\ln\ell\)~\cite{HoshinoOshikawaAshida2026}. From the CFT point of view, this is the boundary logarithm seen by the \(2^\ell\)-normalized numerator: it is the part of the subsystem SRE where the critical scaling is most visible before the purity normalization is included. The corresponding constant \(2C_{\rm II}^{P,{\rm XY}}\) contains the smooth background, the Fisher--Hartwig normalization, and the anisotropic metric contribution. Its explicit form, together with the definition of the regular Fourier coefficients entering it, is given in Appendix~\ref{sec:subsystem_asymptotics}.

The third line describes the crossover regime. Here, the interval length and the correlation length are comparable, \(x=\ell/\xi=O(1)\), so neither the purely critical nor the purely off-critical expansion is sufficient. The scaling function \(2\Phi_{\rm XY}(x)\) connects the Fisher--Hartwig critical regime to the off-critical Szegő boundary regime. It is normalized by \(\Phi_{\rm XY}(0)=0\), and for \(x\to\infty\) behaves as
\begin{equation}
\Phi_{\rm XY}(x)=\frac18\ln x+o(1).
\end{equation}
Therefore the combination \(-\frac14\ln\ell+2\Phi_{\rm XY}(x)\) becomes \(\frac14\ln\xi^{-1}\) plus a finite boundary term in the off-critical limit. This is the matching condition between the critical interval and the massive interval: the crossover function does not modify the bulk density, but it reorganizes the boundary contribution as \(\ell/\xi\) is varied. The derivation and the full expressions for the constants and crossover functions are given in Appendices~\ref{sec:subsystem_asymptotics} and~\ref{app:subsystem-Mhalf-scaling}.

At finite \(\ell\), however, the extracted crossover does not immediately coincide with the limiting function. In the ordinary \(XY\) block one observes two finite-size branches, associated with the two approaches \(h_\ell^\sigma\) to the critical line. Thus the remainder in the third line of Eq.~\eqref{eq:main-subsystem-XY-three-regimes} can be resolved more explicitly as \(2\Phi_{\rm XY}(x)+O(\ell^{-1/2})=2\Phi_{\rm XY}(x)+\frac{2\Psi_{\rm XY}^{\sigma}(x;\gamma)}{\sqrt{\ell}}+o(\ell^{-1/2})\). 
Equivalently, if \(\Phi_{\ell,{\rm XY}}^\sigma(x;\gamma)\) denotes the finite-\(\ell\) crossover function obtained after subtracting the moving bulk, the critical logarithm, and the critical constant, then \(\Phi_{\ell,{\rm XY}}^\sigma(x;\gamma)=\Phi_{\rm XY}(x)+\frac{\Psi_{\rm XY}^{\sigma}(x;\gamma)}{\sqrt{\ell}}+o(\ell^{-1/2})\). The leading function \(\Phi_{\rm XY}(x)\) is independent of the side of the transition, while the branch label \(\sigma\) and the anisotropy \(\gamma\) enter the first finite-size correction. The extrapolated subsystem scaling functions and their logarithmic massive tail are shown in Fig.~\ref{fig:main_results}(c).

\paragraph{Purity normalization and CFT interpretation.}
Equation~\eqref{eq:main-subsystem-XY-three-regimes} describes the Pfaffian numerator with a reference \(2^\ell\) normalization. The purity-normalized SRE also includes \(2S_2(\rho_\ell)\). At the anisotropic critical line, the Gaussian Rényi-2 entanglement entropy contributes \(+\frac14\ln\ell+O(1)\), which cancels the numerator logarithm \(-\frac14\ln\ell\). Hence, the purity-normalized subsystem SRE has no residual \(\ln\ell\) term at \(h=1\).

This cancellation should not be interpreted as the disappearance of critical information. Rather, the universal logarithm of the numerator is absorbed by the Gaussian mixedness of the interval, and the remaining critical data are stored in finite constants and crossover functions. This agrees with the CFT description of stabilizer Rényi entropy~\cite{HoshinoOshikawaAshida2026,RamirezRajabpour2026}: a numerator-like object displays a boundary logarithm, while the mixed-state purity normalization subtracts the Rényi-2 contribution of the reduced Gaussian state. In the notation used below, the surviving data are \(C_{\rm II}^{P,\chi}\), the boundary shifts of the interval representatives, and the crossover functions \(\Phi_\chi^{(\sigma)}\).

\paragraph{Shifted interval representatives.}
We now introduce the shifted Toeplitz representatives required by the subsystem correspondences of Sec.~\ref{sec:generating-functions-correspondences}. We use \({\rm XYR}\) for the Kramers--Wannier/cluster-dual shifted interval and \({\rm XYL}\) for the direct cluster-oriented shifted interval. In terms of the same Fourier coefficients \(g_r(h,\gamma)\), these blocks are
\(\mathbf G_\ell^{\rm XYR,(\infty)}=[g_{1-a+b}]\) and
\(\mathbf G_\ell^{\rm XYL,(\infty)}=[g_{a-b-1}]\), with \(a,b=0,\ldots,\ell-1\). They are not new bulk theories. Rather, they are different finite-interval embeddings of the same infinite-chain XY kernel, and therefore encode different choices of how the interval endpoints cut the Gaussian degrees of freedom.

This distinction is also relevant for the finite-\(\ell\) stability problem. Since the absolute-minor sum is analytic only after a sign branch has been fixed, the ordinary and shifted embeddings must be followed on their corresponding sign-stable branches. These branches have the same physical lower wall \(h^2+\gamma^2=1\), the same correlation length \(\xi(h,\gamma)\), and the same Fisher--Hartwig line \(h=1\); what changes is the finite minor-zero problem that selects the analytic branch. For \({\rm XYR}\), one uses the same interlaced Pfaffian convention as for the ordinary \({\rm XY}\) interval, with the replacement \(g_{a-b}\mapsto g_{1-a+b}\). For \({\rm XYL}\), the natural orientation is different: the shifted block is written as an ordinary \({\rm XY}\) core of size \(\ell-1\) together with an endpoint factor. The detailed construction of these shifted chambers is given in Appendix~\ref{app:subsystem_toeplitz_pfaffian}.

For \({\rm XYR}\), the Pfaffian formula has the same form as Eq.~\eqref{eq:main-subsystem-pfaffian-XY}, with the coefficient replacement \(g_{a-b}\mapsto g_{1-a+b}\):
\begin{equation}
\label{eq:main-subsystem-pfaffian-XYR}
P_\ell^{\rm XYR,(\infty)}(u;h,\gamma)
=
u^\ell\operatorname{Pf}\mathbf K_\ell^{\rm XYR}(u;h,\gamma).
\end{equation}
On the Ising line, the identity \(g_r(h^{-1},1)=g_{1-r}(h,1)\) identifies this shifted block with the ordinary TFI interval at the dual field. Away from \(\gamma=1\), however, it is a genuinely shifted \({\rm XYR}\) block and should not be replaced by an ordinary XY interval.

For \({\rm XYL}\), the most useful main-text expression is the Schur-reduced formula
\begin{equation}
\label{eq:main-subsystem-plus-schur}
\ln P_\ell^{\rm XYL,(\infty)}(u;h,\gamma)
=
\ln P_{\ell-1}^{(\infty)}(u;h,\gamma)
+
\ln\mathcal B_\ell^{\rm XYL}(u;h,\gamma).
\end{equation}
Here \(\mathcal B_\ell^{\rm XYL}\) is the boundary Pfaffian--Schur factor defined in Appendix~\ref{app:subsystem_toeplitz_pfaffian}; the terminology refers to the standard Schur-complement reduction of a block matrix~\cite{ZhangSchur2005}, applied here to the bordered Pfaffian form. Equation~\eqref{eq:main-subsystem-plus-schur} shows that the \({\rm XYL}\) representative contains the ordinary XY Pfaffian problem at size \(\ell-1\), dressed by an additional endpoint factor. Consequently, the bulk sign wall of \({\rm XYL}\) is inherited from the ordinary \({\rm XY}\) core at size \(\ell-1\), while the endpoint factor controls only the finite boundary branch. In the massive core chamber, the numerical checks of Appendix~\ref{app:subsystem_toeplitz_pfaffian} find no additional zero of this factor.

The large-\(\ell\) regimes of \({\rm XYR}\) and \({\rm XYL}\) have the same
structure as Eq.~\eqref{eq:main-subsystem-XY-three-regimes}. These shifted
representatives are different interval embeddings of the same infinite-volume
XY kernel. Therefore, they have the same bulk density, correlation length,
critical point, and Fisher--Hartwig logarithmic coefficient:
\[
m_{\frac12}^{\rm bulk}(h,\gamma),\qquad
\xi(h,\gamma),\qquad
h_c=1,\qquad
-\frac14\ln\ell .
\]
The shift affects only the endpoint-sensitive part of the expansion.
Equivalently, in the off-critical and critical regimes the ordinary XY
constants are replaced by
\[
\begin{aligned}
{\rm XYR}:\qquad&
C_{\rm I}^{P,{\rm XY}}\to C_{\rm I}^{P,{\rm XYR}},
&
C_{\rm II}^{P,{\rm XY}}\to C_{\rm II}^{P,{\rm XYR}},
\\
{\rm XYL}:\qquad&
C_{\rm I}^{P,{\rm XY}}\to C_{\rm I}^{P,{\rm XYL}},
&
C_{\rm II}^{P,{\rm XY}}\to C_{\rm II}^{P,{\rm XYL}} .
\end{aligned}
\]
In the crossover regime, the same scaling variable \(x=\ell/\xi\) controls all
three representatives, while the finite endpoint contribution may be shifted in
a representative-dependent way. These boundary-entropy-like data distinguish
how the interval endpoints cut the Gaussian degrees of freedom and are
essential for the subsystem SRE correspondences to lifted chains,
cluster--Ising intervals, and the dual TFI interval. The complete formulas for
\(C_{\rm I}^{P,\chi}\), \(C_{\rm II}^{P,\chi}\), and
\(\Phi_\chi^{(\sigma)}\) are collected in
Appendices~\ref{sec:subsystem_asymptotics}
and~\ref{app:subsystem-Mhalf-scaling}.

\paragraph{Lifted and cluster subsystem representatives.}
The interval representatives introduced above also determine the subsystem SREs of lifted and cluster chains. No new asymptotic analysis is required: once the subsystem correspondence has been established, the large-\(\ell\) regimes follow by inserting the appropriate interval blocks into Eq.~\eqref{eq:main-subsystem-XY-three-regimes}. Throughout this paragraph we use the \(2^\ell\)-normalized quantity \(\widetilde M_{\frac12}\).

Consider first a lifted symbol \(F(z)=Cz^r f_{\gamma,h}(z^n)\). Let \(\bar r\equiv r\!\!\mod n\), with \(0\le \bar r<n\), and take \(\ell=ns\). The interval decimation of Sec.~\ref{sec:generating-functions-correspondences} decomposes the subsystem into \(n-\bar r\) ordinary \({\rm XY}\) blocks and \(\bar r\) left-shifted \({\rm XYL}\) blocks. Hence 
\[\begin{aligned}
\widetilde M_{\frac12}^{F,(\infty)}(ns;h,\gamma)=&(n-\bar r)\widetilde M_{\frac12}^{{\rm XY},(\infty)}(s;h,\gamma)\\
&\qquad+\bar r\,\widetilde M_{\frac12}^{{\rm XYL},(\infty)}(s;h,\gamma).    
\end{aligned}\]
Thus, the bulk density is still \(m_{\frac12}^{\rm bulk}(h,\gamma)\), while the finite boundary data are the corresponding sums over the decimated blocks:
\[\begin{gathered}
C_{\rm I}^{P,F}=(n-\bar r)C_{\rm I}^{P,{\rm XY}}+\bar r\,C_{\rm I}^{P,{\rm XYL}},\\ C_{\rm II}^{P,F}=(n-\bar r)C_{\rm II}^{P,{\rm XY}}+\bar r\,C_{\rm II}^{P,{\rm XYL}}
\end{gathered}\]
and \(\Phi_F^{(\sigma)}=(n-\bar r)\Phi_{\rm XY}^{(\sigma)}+\bar r\,\Phi_{\rm XYL}^{(\sigma)}\). The common crossover function of XY and XYR is \(\Phi_{\rm XY}(x)\), then for the lifted model we have
\[\Phi_F(x)=n\Phi_{\rm XY}(x).\]
In particular, on the anisotropic critical line,
\begin{equation}
\label{eq:main-subsystem-F-critical}
\widetilde M_{\frac12}^{F,(\infty)}(\ell;1,\gamma)=\ell\,m_{\frac12}^{\rm bulk}(1,\gamma)-\frac n4\ln s+2C_{\rm II}^{P,F}(\gamma)+o(1).
\end{equation}
Equivalently, one may write the logarithm in terms of \(\ell=ns\); the difference \(\frac n4\ln n\) is then absorbed into the finite critical constant. The important point is that the Fisher--Hartwig logarithm is multiplied by the number of decimated interval blocks.


For the cluster--Ising representative, the subsystem correspondence has a
different origin. It is not a decomposition into several blocks. Instead,
Kramers--Wannier duality selects the right-shifted interval representative.
With
\[
h=\frac{g_{zz}}{g_x+g_{zxz}},
\qquad
\gamma=\frac{g_x-g_{zxz}}{g_x+g_{zxz}},
\]
one has
\[
\widetilde M_{\frac12}^{{\rm CI},(\infty)}
(\ell;g_{zxz},g_{zz},g_x)
=
\widetilde M_{\frac12}^{{\rm XYR},(\infty)}
(\ell;h,\gamma).
\]
Therefore, all large-\(\ell\) subsystem data of the cluster--Ising interval are
inherited from the \({\rm XYR}\) block. In particular,
\[
\begin{gathered}
C_{\rm I}^{P,{\rm CI}}=C_{\rm I}^{P,{\rm XYR}},
\qquad
C_{\rm II}^{P,{\rm CI}}=C_{\rm II}^{P,{\rm XYR}},
\\
\Phi_{\rm CI}^{(\sigma)}(x)=\Phi_{\rm XYR}^{(\sigma)}(x).
\end{gathered}
\]
Along the cluster critical line
\((g_{zxz},g_{zz},g_x)=(g_c,2,2-g_c)\), the map gives
\(h=1\) and \(\gamma=1-g_c\), up to the orientation convention fixed in
Sec.~\ref{sec:generating-functions-correspondences}. Thus fixed
\(g_c\neq1\) is described by the critical \({\rm XYR}\) regime. By contrast,
the approach \(g_c\to1\) reaches the XX endpoint \((h,\gamma)=(1,0)\), where
the fixed-\(\gamma\) expansion is no longer uniform.

\section{Conclusions}
\label{sec:conclusions}

We have shown that stabilizer Rényi entropy can be treated as an exact scaling observable in a broad class of solvable one-dimensional spin chains. The central question addressed in this work was whether the universal structures suggested by recent field-theory approaches to stabilizer entropy can be followed beyond isolated conformal points and into massive, crossover, endpoint, and boundary-sensitive regimes. For the stabilizer Rényi entropy at \(\alpha=1/2\), the exact finite-size formulas derived here provide a positive answer.

Our main results can be summarized as follows.

\begin{itemize}

\item \textbf{Exact finite-size formulas in two geometries.}
We obtained exact formulas for the stabilizer Rényi entropy at \(\alpha=1/2\) for both full periodic chains and finite intervals of the infinite chain. In the periodic geometry, the many-body Pauli-amplitude problem reduces to a finite product over paired Bogoliubov modes. In the interval geometry, the corresponding quantity is represented by a finite Pfaffian. These two formulas are the microscopic input that makes the large-size scaling regimes analytically accessible.

\item \textbf{Universal periodic crossover beyond criticality.} For the full periodic chain, we derived the massive, critical, and finite-size crossover regimes from the exact product formula. Away from the field-tuned critical line, the SRE is governed by a volume-law density, with finite-size corrections that are exponentially small in the ratio between the system size and the correlation length. Near criticality, the relevant scaling variable is \(x=L|1-h|/\gamma\), and the subleading part is controlled by the universal crossover function \(-\ln(1+e^{-x})\). This function gives an exact lattice interpolation between the critical regime, where the finite part approaches \(-\ln2\), and the massive regime, where the finite part vanishes. Thus, the periodic result does more than reproduce the critical constant: it resolves how the stabilizer entropy crosses over from critical to noncritical scaling at finite size.

\item \textbf{Nonanalytic bulk response.}
The periodic formula also gives a direct analytic signature of criticality in the large-scale entropy density. The density is continuous across the anisotropic transition, but its derivative with respect to the transverse field jumps at \(h=1\). Thus, the SRE density develops a cusp at the critical line. Related singular responses have recently appeared in numerical and quantum Monte Carlo studies of many-body SRE, where derivatives and volume-law terms reveal critical behavior \cite{DingYan2025}. The finite-size crossover term can be viewed as the rounded precursor of this large-scale nonanalyticity.

\item \textbf{Separate XX endpoint scaling.}
We identified the XX endpoint \((h,\gamma)=(1,0)\) as a separate scaling regime. It is not obtained by taking the fixed-\(\gamma\) critical expansion and then sending \(\gamma\to0\). Instead, the correct endpoint behavior requires a different double-scaling window, with \(h-1=O(L^{-2})\) and \(\gamma=O(L^{-2})\). This separates the ordinary anisotropic critical regime from the XX endpoint.

\item \textbf{Subsystem boundary crossover beyond criticality.}
For finite intervals in the infinite chain, the scaling theory acquires genuinely boundary-sensitive data. The large-\(\ell\) expansion contains the same bulk volume-law density as the large-scale periodic chain, because the interior of a long interval is locally translationally invariant. The finite part, however, is controlled by the two interval endpoints and is therefore absent from the periodic geometry. In the massive regime, these endpoint contributions remain finite, with corrections suppressed by the correlation length. Upon approaching the field-tuned critical line, the massive boundary term develops the matching behavior \(\frac14\ln\xi^{-1}(h,\gamma)\), showing that the critical logarithm is not lost away from criticality, but is reorganized into a correlation-length-dependent boundary contribution.

In the crossover regime, the relevant variable is \(x=\ell/\xi\), which compares the interval length with the correlation length. For the ordinary \(XY\) interval at \(\alpha=1/2\), after subtracting the moving bulk, the critical logarithm, and the critical constant, the leading residual crossover is described by a single function \(\Phi_{\rm XY}(x)\). Its large-\(x\) tail, \(\Phi_{\rm XY}(x)\sim\frac18\ln x\), precisely reproduces the massive matching term. Finite intervals may still display two preasymptotic branches associated with the two approaches to the critical line, but their separation is absorbed into the leading finite-size correction \(\Psi_{\rm XY}^{\sigma}(x;\gamma)/\sqrt{\ell}\). Thus, the side of the transition and the anisotropy affect the finite-size approach, while the leading \(XY\) boundary crossover is governed by the scaling variable \(x\).

This subsystem result goes beyond the known critical logarithm \(-\frac14\ln\ell\). It gives an exact lattice description of how interval-boundary data evolve from critical to noncritical scaling, separating the universal bulk density from boundary constants, crossover functions, and finite-size branch corrections. This provides a microscopic counterpart of the boundary-sensitive structures expected from CFT and defect-based approaches to SRE~\cite{HoshinoOshikawaAshida2026,HoshinoAshida2025Defects}.


\item \textbf{Extensions through exact SRE correspondences.}
The scaling theory is not restricted to the elementary XY representative. Exact stabilizer-entropy correspondences transfer the periodic and subsystem results to internal XY reductions, finite-range spin-chain representatives, and Cluster--Ising models. For folded finite-range chains, the correlation matrix decomposes into independent blocks. For Cluster--Ising representatives, the correspondence follows from Kramers--Wannier duality and the matching of the relevant Pauli algebra. This connects the present exact results with recent studies of SRE in spin models and infinite matrix product states \cite{Viscardi2026SpinModels,Hallam2026}.


\end{itemize}

These results provide an exact microscopic lattice realization of the emerging CFT picture of stabilizer entropy \cite{HoshinoOshikawaAshida2026,RamirezRajabpour2026,HoshinoAshida2025Defects}, and is complementary to finite-temperature open-chain results where stabilizer entropy detects hidden boundary- or defect-like conformal data \cite{KhassehRajabpour2026ThermalSRE}. Existing field-theory approaches mainly clarify the universal structure at criticality, while the present work follows the same observable into massive and crossover regimes. In this sense, the exact formulas derived here give a benchmark for determining which features of stabilizer entropy are universal, which are boundary sensitive, and which depend on the microscopic realization.

Several directions remain open. A natural next step is to extend the analysis beyond \(\alpha=1/2\). Since this value lies below the boundary-sector transition expected at \(\alpha=4\), our results suggest that the same qualitative scaling structure may persist throughout that replica sector: a volume law, universal critical terms, massive-to-critical crossover functions, and boundary-sensitive subsystem data. The corresponding constants and scaling functions are expected to depend on the Rényi index, while the transition point should mark a change of boundary sector and require a modified description~\cite{RamirezRajabpour2026,HoshinoOshikawaAshida2026}. Proving this picture microscopically remains an important open problem.

A related open problem is to derive the subsystem crossover function \(\Phi_{\rm XY}(x)\) analytically, beyond the numerical extraction used here. Although the scalar symbol correctly identifies the bulk density, the correlation length, and the local Fisher--Hartwig singularity, the full crossover function is an interval quantity and is therefore controlled by the underlying block-Pfaffian Toeplitz problem. A first-principles derivation would require a careful asymptotic analysis of the near-critical block symbol, including its Wiener--Hopf/Riemann--Hilbert factorization and the associated Pfaffian Fisher--Hartwig scaling limit. This is technically closer to a matrix Toeplitz/Riemann--Hilbert problem than to the periodic product formula, and we leave its full treatment for future work.

Another important direction is to derive the massive and crossover functions directly from a continuum QFT description. This would clarify which parts of the crossover are universal and how they are encoded by the low-energy theory. It would also be valuable to identify more explicitly the boundary conditions, defect data, or boundary-condition-changing structures associated with the shifted interval representatives~\cite{HoshinoAshida2025Defects}. Finally, the formulas obtained here can serve as benchmarks for numerical, tensor-network, quantum Monte Carlo, and experimental approaches to many-body nonstabilizerness~\cite{Tarabunga2023,Tarabunga2024MPS,Frau2025,DingYan2025,Oliviero2022MeasuringMagic,Haug2024EfficientAlgorithms}.

Overall, our results show that stabilizer Rényi entropy is not only a diagnostic of nonstabilizerness, but also a sharply defined scaling observable. In the solvable spin chains studied here, its universal crossover structure can be computed exactly beyond criticality, providing a microscopic bridge between many-body stabilizer entropy and its emerging QFT description.

{\it Acknowledgements:}
M.A.R. acknowledges partial support from CNPq and FAPERJ (grant number E-26/210.062/2023). E.A.R.T. acknowledges support from CNPq (Process No.~141672/2023-4). R.K. gratefully acknowledges the resources on the LiCCA HPC cluster of the University of Augsburg, co-funded by the Deutsche Forschungsgemeinschaft (DFG, German Research Foundation)–Project-ID 499211671. 
\\

\clearpage
\onecolumngrid

\appendix
\section{Computational complexity of the absolute-minor functional}
\label{app:proof-theorem1}

In the main text, the index \(\alpha=1/2\) reduces the Gaussian SRE to the absolute-minor functional \(\operatorname{Det}_1(\mathbf G)\). We now explain why this functional is not expected to be efficiently computable for arbitrary matrices. The statement below applies to general rational matrices and is independent of the Toeplitz, Pfaffian, and finite-periodic structures used in the XY-chain calculations. Thus, the exact formulas derived in this work are not consequences of Gaussianity alone, but of additional algebraic structure.

The reduction uses a standard object from convex geometry. Given vectors \(\mathbf a_1,\ldots,\mathbf a_n\in\mathbb R^r\) spanning an $r$-dimensional space, the associated zonotope is the Minkowski sum of the line segments generated by these vectors, equivalently 
\begin{equation}
Z(\mathbf A)
=
\left\{
\sum_{j=1}^{n}\lambda_j \mathbf a_j
\;:\;
0\le \lambda_j \le 1
\right\},
\qquad
\mathbf A=
\left[
\mathbf a_1,\ldots,\mathbf a_n
\right].
\end{equation}
Its \(r\)-dimensional volume is the sum of the absolute values of the maximal minors of \(\mathbf A\)~\cite{Gover2016}. Since exact zonotope-volume computation is \(\#P\)-hard~\cite{Dyer1998}, this gives a direct route to hardness for absolute-minor sums.

\begin{theorem}
\label{theorem1}
Let \(\mathbf A\in\mathbb Q^{m\times n}\), and define \(\displaystyle S(\mathbf A):=\sum_{k=0}^{\min(m,n)}\sum_{\substack{I\subseteq\{1,\dots,m\},\,J\subseteq\{1,\dots,n\}\\ |I|=|J|=k}}\bigl|\det \mathbf A[I,J]\bigr|\), where \(\mathbf A[I,J]\) is the submatrix with rows \(I\) and columns \(J\). Then the exact computation of \(S(\mathbf A)\) is \(\#P\)-hard under polynomial-time Turing reductions.
\end{theorem}

\begin{proof}
We reduce from exact zonotope volume. Let \(\mathbf A=[\mathbf a_1,\dots,\mathbf a_n]\in\mathbb Q^{r\times n}\) have rank \(r\). The zonotope generated by its columns has volume
\begin{equation}
\displaystyle \operatorname{vol}_r Z(\mathbf A)=\sum_{\substack{J\subseteq\{1,\dots,n\}\\ |J|=r}}\bigl|\det \mathbf A[:,J]\bigr|,
\end{equation}
where \(\mathbf A[:,J]\) denotes the submatrix obtained by keeping all rows and the columns in \(J\). This is the standard maximal-minor formula for zonotope volume. Assume that an oracle computes \(S(\mathbf B)\) exactly for any rational matrix \(\mathbf B\). For a nonnegative integer \(t\), every \(k\times k\) minor of \(t\mathbf A\) is multiplied by \(t^k\). Therefore
\begin{equation}
\displaystyle S(t\mathbf A)=\sum_{k=0}^{r}t^k\sum_{\substack{I\subseteq\{1,\dots,r\},\,J\subseteq\{1,\dots,n\}\\ |I|=|J|=k}}\bigl|\det \mathbf A[I,J]\bigr|,
\end{equation}
so \(S(t\mathbf A)\) is a polynomial in \(t\) of degree at most \(r\). Its leading coefficient is \(\displaystyle \sum_{\substack{J\subseteq\{1,\dots,n\}\\ |J|=r}}\bigl|\det \mathbf A[:,J]\bigr|\), because for \(k=r\) the only possible row set is \(I=\{1,\dots,r\}\).

The leading coefficient can be recovered from the \(r+1\) oracle values \(S(0\mathbf A),S(\mathbf A),\dots,S(r\mathbf A)\) by interpolation, or equivalently by the finite-difference identity
\begin{equation}
\displaystyle \sum_{\substack{J\subseteq\{1,\dots,n\}\\ |J|=r}}\bigl|\det \mathbf A[:,J]\bigr|=\frac{1}{r!}\sum_{j=0}^{r}(-1)^{r-j}\binom{r}{j}S(j\mathbf A).
\end{equation}
The matrices \(j\mathbf A\), \(0\le j\le r\), have rational entries with polynomially bounded bit length. Hence an exact oracle for \(S\) would compute exact zonotope volumes in polynomial time.

Since exact zonotope volume is \(\#P\)-hard, exact computation of \(S(\mathbf A)\) is also \(\#P\)-hard under polynomial-time Turing reductions.
\end{proof}
\section{Hamiltonian--symbol dictionary}
\label{app:hamiltonian-symbol-dictionary}

This appendix collects the Hamiltonian representatives and Laurent symbols used throughout the paper. The purpose is not to rederive the Jordan--Wigner solution, but to fix conventions. After the Jordan--Wigner transformation, the finite-range spin chains considered here become real quadratic fermionic Hamiltonians. In a fixed fermion-parity sector, translation invariance allows one to encode the single-particle Bogoliubov data in a Laurent polynomial \(F(z)=\sum_m A_m z^m\), with \(z=e^{iq}\). The phase \(F(e^{iq})/|F(e^{iq})|\) then determines the Gaussian correlation matrix used in the SRE calculation.

In the spin representation, each Laurent monomial corresponds to a Jordan--Wigner string operator. With the conventions of the main text, the zeroth power gives the transverse-field term, \(O_j^{(0)}=\sigma_j^z\), while positive and negative powers generate, respectively, \(x\)- and \(y\)-string interactions:
\[
O_j^{(m)}
=
\sigma_j^x
\left(\prod_{s=1}^{m-1}\sigma_{j+s}^z\right)
\sigma_{j+m}^x,
\qquad m>0,
\]
and
\[
O_j^{(-m)}
=
\sigma_j^y
\left(\prod_{s=1}^{m-1}\sigma_{j+s}^z\right)
\sigma_{j+m}^y,
\qquad m>0.
\]
Thus an XY-type Laurent symbol \(F(z)=\sum_m A_m z^m\) represents the spin Hamiltonian
\[
H_F=-\sum_j\sum_m A_m\,O_j^{(m)}.
\]
Table~\ref{tab:xy-hamiltonians-symbols} applies this rule to the nearest-neighbor XY chain, its TFI and XX limits, the lifted symbols \(Cz^r f_{\gamma,h}(z^n)\), and the shifted interval representatives used later in the subsystem correspondences. The Cluster--Ising entries are listed separately because their connection with the XY symbol is obtained through Kramers--Wannier duality rather than by a direct \(G\)-matrix block decimation.

\begin{table*}[t]
\centering
\tiny
\setlength{\tabcolsep}{3pt}
\renewcommand{\arraystretch}{1.35}

\begin{tabular*}{\textwidth}{@{\extracolsep{\fill}}lll}
\hline\hline
\rowcolor{blue!10}
\parbox[c]{0.15\textwidth}{\centering\textbf{Model / representative}}
&
\parbox[c]{0.34\textwidth}{\centering\textbf{Laurent symbol}}
&
\parbox[c]{0.43\textwidth}{\centering\textbf{Hamiltonian}}
\\
\hline

\rowcolor{blue!3}
\parbox[c]{0.15\textwidth}{\centering XY chain}
&
\parbox[c]{0.34\textwidth}{\centering
\(\displaystyle f_{\gamma,h}(z)=h+\frac{1+\gamma}{2}z+\frac{1-\gamma}{2}z^{-1}\)}
&
\parbox[c]{0.43\textwidth}{\centering
\(\displaystyle H_{\rm XY}=-\sum_j\left[
\frac{1+\gamma}{2}\sigma_j^x\sigma_{j+1}^x+
\frac{1-\gamma}{2}\sigma_j^y\sigma_{j+1}^y+
h\sigma_j^z
\right]\)}
\\
\hline

\parbox[c]{0.15\textwidth}{\centering TFI line}
&
\parbox[c]{0.34\textwidth}{\centering
\(\displaystyle f_{1,h}(z)=h+z\)}
&
\parbox[c]{0.43\textwidth}{\centering
\(\displaystyle H_{\rm TFI}=-\sum_j\left[
\sigma_j^x\sigma_{j+1}^x+h\sigma_j^z
\right]\)}
\\
\hline

\rowcolor{blue!3}
\parbox[c]{0.15\textwidth}{\centering Opposite TFI orientation}
&
\parbox[c]{0.34\textwidth}{\centering
\(\displaystyle f_{-1,h}(z)=h+z^{-1}\)}
&
\parbox[c]{0.43\textwidth}{\centering
\(\displaystyle H_{\rm TFI}^{\rm opp}=-\sum_j\left[
\sigma_j^y\sigma_{j+1}^y+h\sigma_j^z
\right]\)}
\\
\hline

\parbox[c]{0.15\textwidth}{\centering XX chain}
&
\parbox[c]{0.34\textwidth}{\centering
\(\displaystyle f_{0,h}(z)=h+\frac12(z+z^{-1})\)}
&
\parbox[c]{0.43\textwidth}{\centering
\(\displaystyle H_{\rm XX}=-\sum_j\left[
\frac12\sigma_j^x\sigma_{j+1}^x+
\frac12\sigma_j^y\sigma_{j+1}^y+
h\sigma_j^z
\right]\)}
\\
\hline

\rowcolor{blue!3}
\parbox[c]{0.15\textwidth}{\centering Zero-field XY}
&
\parbox[c]{0.34\textwidth}{\centering
\(\displaystyle f_{\gamma,0}(z)=\frac{1+\gamma}{2}z+\frac{1-\gamma}{2}z^{-1}\)}
&
\parbox[c]{0.43\textwidth}{\centering
\(\displaystyle H_{\rm XY}(\gamma,0)=-\sum_j\left[
\frac{1+\gamma}{2}\sigma_j^x\sigma_{j+1}^x+
\frac{1-\gamma}{2}\sigma_j^y\sigma_{j+1}^y
\right]\)}
\\
\hline

\parbox[c]{0.15\textwidth}{\centering Zero-field XY as TFI lift}
&
\parbox[c]{0.34\textwidth}{\centering
\(\displaystyle f_{\gamma,0}(z)=\frac{1+\gamma}{2}z^{-1}f_{1,\tilde h}(z^2),\quad
\tilde h=\frac{1-\gamma}{1+\gamma}\)}
&
\parbox[c]{0.43\textwidth}{\centering
Same zero-field XY Hamiltonian; after decimation it gives TFI blocks.
}
\\
\hline

\rowcolor{blue!3}
\parbox[c]{0.15\textwidth}{\centering Range-\(n\) XY lift}
&
\parbox[c]{0.34\textwidth}{\centering
\(\displaystyle f_{\gamma,h}(z^n)=h+\frac{1+\gamma}{2}z^n+\frac{1-\gamma}{2}z^{-n}\)}
&
\parbox[c]{0.43\textwidth}{\centering
\(\displaystyle H_n=-\sum_j\left[
\frac{1+\gamma}{2}O_j^{(n)}
+\frac{1-\gamma}{2}O_j^{(-n)}
+hO_j^{(0)}
\right]\)}
\\
\hline

\parbox[c]{0.15\textwidth}{\centering Range-\(2\) XY lift}
&
\parbox[c]{0.34\textwidth}{\centering
\(\displaystyle f_{\gamma,h}(z^2)=h+\frac{1+\gamma}{2}z^2+\frac{1-\gamma}{2}z^{-2}\)}
&
\parbox[c]{0.43\textwidth}{\centering
\(\displaystyle H_2=-\sum_j\left[
\frac{1+\gamma}{2}\sigma_j^x\sigma_{j+1}^z\sigma_{j+2}^x+
\frac{1-\gamma}{2}\sigma_j^y\sigma_{j+1}^z\sigma_{j+2}^y+
h\sigma_j^z
\right]\)}
\\
\hline

\rowcolor{blue!3}
\parbox[c]{0.15\textwidth}{\centering Lifted symbol, \(0\le r<n\)}
&
\parbox[c]{0.34\textwidth}{\centering
\(\displaystyle 
F(z)=Cz^r f_{\gamma,h}(z^n)
=
C\left[
\frac{1-\gamma}{2}z^{-(n-r)}
+
hz^r
+
\frac{1+\gamma}{2}z^{r+n}
\right]\)}
&
\parbox[c]{0.43\textwidth}{\centering
\(\displaystyle
H_F
=
-C\sum_j
\left[
\frac{1-\gamma}{2}O_j^{-(n-r)}
+
hO_j^{(r)}
+
\frac{1+\gamma}{2}O_j^{(r+n)}
\right]\)}
\\
\hline

\parbox[c]{0.15\textwidth}{\centering Lifted symbol, \(r=n\)}
&
\parbox[c]{0.34\textwidth}{\centering
\(\displaystyle 
F(z)=Cz^n f_{\gamma,h}(z^n)
=
C\left[
\frac{1-\gamma}{2}
+
h z^n
+
\frac{1+\gamma}{2}z^{2n}
\right]\)}
&
\parbox[c]{0.43\textwidth}{\centering
\(\displaystyle
H_F
=
-C\sum_j
\left[
\frac{1-\gamma}{2}\sigma^z_j
+
hO_j^{(n)}
+
\frac{1+\gamma}{2}O_j^{(2n)}
\right]\)}
\\
\hline

\rowcolor{blue!3}
\parbox[c]{0.15\textwidth}{\centering Lifted symbol, \(r>n\)}
&
\parbox[c]{0.34\textwidth}{\centering
\(\displaystyle 
F(z)=Cz^r f_{\gamma,h}(z^n)
=
C\left[
\frac{1-\gamma}{2}z^{r-n}
+
h z^r
+
\frac{1+\gamma}{2}z^{r+n}
\right]\)}
&
\parbox[c]{0.43\textwidth}{\centering
\(\displaystyle
H_F
=
-C\sum_j
\left[
\frac{1-\gamma}{2}O_j^{(r-n)}
+
hO_j^{(r)}
+
\frac{1+\gamma}{2}O_j^{(r+n)}
\right]\)}
\\
\hline

\parbox[c]{0.15\textwidth}{\centering Left-shifted block \(XYL\)}
&
\parbox[c]{0.34\textwidth}{\centering
\(\displaystyle zf_{\gamma,h}(z)=\frac{1-\gamma}{2}+hz+\frac{1+\gamma}{2}z^2\)}
&
\parbox[c]{0.43\textwidth}{\centering
\(\displaystyle H_{XYL}=-\sum_j\left[
\frac{1-\gamma}{2}\sigma_j^z
+h\sigma_j^x\sigma_{j+1}^x
+\frac{1+\gamma}{2}\sigma_j^x\sigma_{j+1}^z\sigma_{j+2}^x
\right]\)}
\\
\hline

\rowcolor{blue!3}
\parbox[c]{0.15\textwidth}{\centering Right-shifted block \(XYR\)}
&
\parbox[c]{0.34\textwidth}{\centering
\(\displaystyle zf_{\gamma,h}(z^{-1})=\frac{1+\gamma}{2}+hz+\frac{1-\gamma}{2}z^2\)}
&
\parbox[c]{0.43\textwidth}{\centering
\(\displaystyle H_{XYR}=-\sum_j\left[
\frac{1+\gamma}{2}\sigma_j^z
+h\sigma_j^x\sigma_{j+1}^x
+\frac{1-\gamma}{2}\sigma_j^x\sigma_{j+1}^z\sigma_{j+2}^x
\right]\)}
\\
\hline

\parbox[c]{0.15\textwidth}{\centering TFI self-dual form}
&
\parbox[c]{0.34\textwidth}{\centering
\(\displaystyle f_{1,h}(z)=hz f_{1,h^{-1}}(z^{-1})\)}
&
\parbox[c]{0.43\textwidth}{\centering
Same TFI Hamiltonian, written in the dual orientation with field \(h^{-1}\).
}
\\

\hline\hline
\end{tabular*}

\caption{
Direct Hamiltonian--symbol dictionary for XY-type representatives. Each coefficient
of \(z^m\) determines the corresponding string operator \(O_j^{(m)}\). Positive
powers give \(x\)-string interactions, negative powers give \(y\)-string
interactions, and the zeroth power gives the transverse field. The lifted symbols
\(Cz^r f_{\gamma,h}(z^n)\) generate longer-range string Hamiltonians and decimate
into XY blocks.
}
\label{tab:xy-hamiltonians-symbols}
\end{table*}

\begin{table*}[t]
\centering
\tiny
\setlength{\tabcolsep}{3pt}
\renewcommand{\arraystretch}{1.35}

\begin{tabular*}{\textwidth}{@{\extracolsep{\fill}}lll}
\hline\hline
\rowcolor{blue!10}
\parbox[c]{0.15\textwidth}{\centering\textbf{Model / representative}}
&
\parbox[c]{0.34\textwidth}{\centering\textbf{Laurent symbol}}
&
\parbox[c]{0.43\textwidth}{\centering\textbf{Hamiltonian / parameter map}}
\\
\hline

\rowcolor{blue!3}
\parbox[c]{0.15\textwidth}{\centering Cluster--Ising}
&
\parbox[c]{0.34\textwidth}{\centering
\(\displaystyle F_{\rm CI}(z)=g_x-g_{zz}z+g_{zxz}z^2\)}
&
\parbox[c]{0.43\textwidth}{\centering
\(\displaystyle H_{\rm CI}=g_{zxz}\sum_j\sigma_j^z\sigma_{j+1}^x\sigma_{j+2}^z
-g_{zz}\sum_j\sigma_j^z\sigma_{j+1}^z
-g_x\sum_j\sigma_j^x\)}
\\
\hline

\parbox[c]{0.15\textwidth}{\centering Cluster--Ising as shifted XY}
&
\parbox[c]{0.34\textwidth}{\centering
\(\displaystyle F_{\rm CI}(z)=(g_x+g_{zxz})z f_{\gamma,-h}(z^{-1})\)}
&
\parbox[c]{0.43\textwidth}{\centering
\(\displaystyle \gamma=\frac{g_x-g_{zxz}}{g_x+g_{zxz}},
\quad
h=\frac{g_{zz}}{g_x+g_{zxz}}\)}
\\
\hline

\rowcolor{blue!3}
\parbox[c]{0.15\textwidth}{\centering Dual XY image of CI}
&
\parbox[c]{0.34\textwidth}{\centering
\(\displaystyle f_{\gamma,h}(z),\quad
\gamma=\frac{g_x-g_{zxz}}{g_x+g_{zxz}},
\quad
h=\frac{g_{zz}}{g_x+g_{zxz}}\)}
&
\parbox[c]{0.43\textwidth}{\centering
\(\displaystyle H_{\rm XY}^{\rm dual}
=-\sum_j\left[
g_x\tau_j^x\tau_{j+1}^x+
g_{zxz}\tau_j^y\tau_{j+1}^y+
g_{zz}\tau_j^z
\right]\)}
\\
\hline

\parbox[c]{0.15\textwidth}{\centering MPS skeleton}
&
\parbox[c]{0.34\textwidth}{\centering
\(\displaystyle F_{\rm skel}(z)=\left[(1+g)-(1-g)z\right]^2\)}
&
\parbox[c]{0.43\textwidth}{\centering
\(\displaystyle H_{\rm skel}=(g-1)^2\sum_j\sigma_j^z\sigma_{j+1}^x\sigma_{j+2}^z
+2(g^2-1)\sum_j\sigma_j^z\sigma_{j+1}^z
-(1+g)^2\sum_j\sigma_j^x\)}
\\
\hline

\rowcolor{blue!3}
\parbox[c]{0.15\textwidth}{\centering Cluster--Ising critical line}
&
\parbox[c]{0.34\textwidth}{\centering
\(\displaystyle F_{\rm CI}^{\rm crit}(z)=(2-g_c)-2z+g_cz^2\)}
&
\parbox[c]{0.43\textwidth}{\centering
\(\displaystyle H_{\rm CI}^{\rm crit}
=g_c\sum_j\sigma_j^z\sigma_{j+1}^x\sigma_{j+2}^z
-2\sum_j\sigma_j^z\sigma_{j+1}^z
-(2-g_c)\sum_j\sigma_j^x\)}
\\

\hline\hline
\end{tabular*}

\caption{
Cluster--Ising and duality-related representatives. The Cluster--Ising symbols
are listed separately because they are related to the XY symbols through
Kramers--Wannier duality.
}
\label{tab:cluster-hamiltonians-symbols}
\end{table*}

\section{Cluster--Ising model, MPS skeleton, and dual XY description}
\label{app:cluster_XY_duality}

This appendix fixes the conventions behind the cluster--Ising representatives used in the main text. The essential point is that the same bulk free-fermion data can be described in two equivalent but differently oriented ways. The Cluster--Ising representatives and their dual XY images are summarized in Table~\ref{tab:cluster-hamiltonians-symbols}. The first is the Hamiltonian Kramers--Wannier map to an anisotropic XY chain; the second is the direct cluster--Ising Laurent symbol. They differ by a monomial shift, an inversion \(z\mapsto z^{-1}\), and a field-sign convention. These differences are harmless for the SRE once the Pauli algebra is matched, but they matter for identifying the correct XY representative, the finite-size sector, and the interval block. We use the standard Kramers--Wannier duality conventions~\cite{KramersWannier1941a,KramersWannier1941b,Kogut1979} and the cluster--Ising terminology of Refs.~\cite{Son2011,Smacchia2011}; the MPS skeleton below was recently used as a benchmark for transfer-matrix and spectral approaches to SRE~\cite{Hallam2026}.

We consider
\begin{equation}
\label{eq:H_cluster_ising_app}
H_{\rm CI}
=
g_{zxz}\sum_j\sigma^z_j\sigma^x_{j+1}\sigma^z_{j+2}
-
g_{zz}\sum_j\sigma^z_j\sigma^z_{j+1}
-
g_x\sum_j\sigma^x_j .
\end{equation}
The three couplings interpolate between the cluster stabilizer interaction, the Ising interaction, and the transverse-field paramagnet. A useful exactly solvable trajectory is the MPS skeleton,
\[
H_{\rm skel}(g)
=
(g-1)^2\sum_j\sigma^z_j\sigma^x_{j+1}\sigma^z_{j+2}
+
2(g^2-1)\sum_j\sigma^z_j\sigma^z_{j+1}
-
(1+g)^2\sum_j\sigma^x_j ,
\]
which corresponds, in the convention of Eq.~\eqref{eq:H_cluster_ising_app}, to \(g_{zxz}=(g-1)^2\), \(g_{zz}=2(1-g^2)\), and \(g_x=(1+g)^2\).

\paragraph*{Bulk Kramers--Wannier map.}
In the bulk, introduce dual Pauli operators
\[
\tau_j^z=\sigma_j^z\sigma_{j+1}^z,
\qquad
\tau_j^x=\prod_{k\le j}\sigma_k^x .
\]
Then \(\sigma_j^x=\tau_{j-1}^x\tau_j^x\) and \(\sigma_j^z\sigma_{j+1}^z=\tau_j^z\). With the Pauli-axis convention used here, the cluster interaction maps as \(\sigma_j^z\sigma_{j+1}^x\sigma_{j+2}^z\mapsto-\tau_j^y\tau_{j+1}^y\). Up to boundary terms,
\begin{equation}
\label{eq:H_dual_app}
H_{\rm CI}
\mapsto
H_{\rm dual}
=
-\sum_j
\left[
g_x\tau_j^x\tau_{j+1}^x
+
g_{zxz}\tau_j^y\tau_{j+1}^y
+
g_{zz}\tau_j^z
\right].
\end{equation}
Comparing with \(H_{\rm XY}=-\frac12\sum_j[(1+\gamma)\tau_j^x\tau_{j+1}^x+(1-\gamma)\tau_j^y\tau_{j+1}^y+2h\tau_j^z]\), and removing the overall positive scale \(g_x+g_{zxz}\), gives the XY parameters
\begin{equation}
\label{eq:gamma_h_app}
\gamma
=
\frac{g_x-g_{zxz}}{g_x+g_{zxz}},
\qquad
h
=
\frac{g_{zz}}{g_x+g_{zxz}} .
\end{equation}
Thus the Hamiltonian-level dual representative is the same XY symbol used in the main text,
\[
f_{\gamma,h}(z)
=
h+\frac{1+\gamma}{2}z+\frac{1-\gamma}{2}z^{-1},
\qquad z=e^{iq}.
\]
The cluster--Ising SRE is obtained from this XY Gaussian data after matching the periodic sector or, for subsystems, the interval algebra. The local operator map and the resulting XY parameters are summarized in Fig.~\ref{fig:KW_cluster_XY_scheme}.

\begin{figure}[!ht]
\centering
\includegraphics[width=0.6\textwidth]{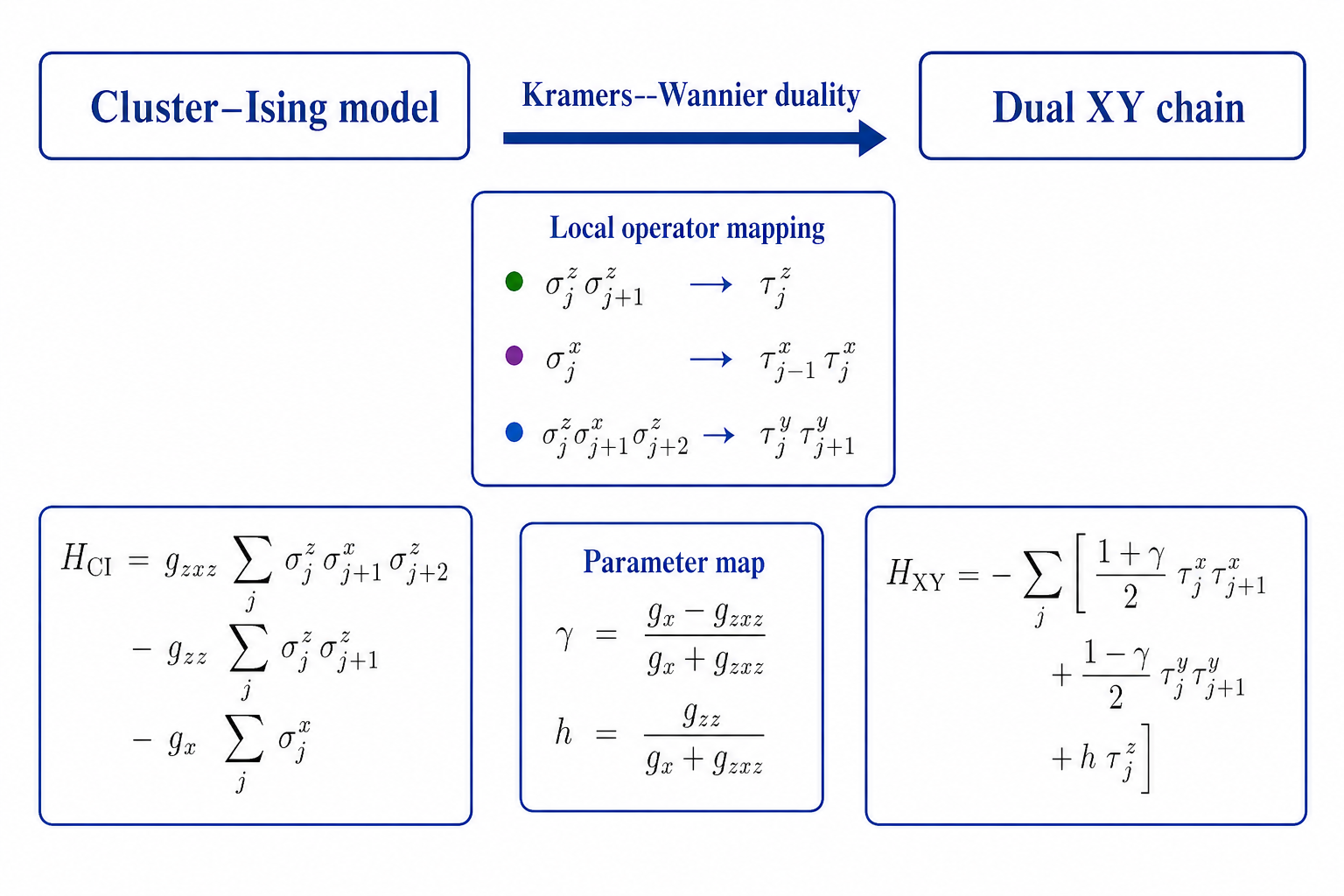}
\caption{Kramers--Wannier map between the cluster--Ising chain and the dual XY chain. In the bulk, \(\sigma_j^z\sigma_{j+1}^z\mapsto\tau_j^z\), \(\sigma_j^x\mapsto\tau_{j-1}^x\tau_j^x\), and \(\sigma_j^z\sigma_{j+1}^x\sigma_{j+2}^z\mapsto-\tau_j^y\tau_{j+1}^y\), up to the Pauli-axis convention. The resulting XY parameters are those of Eq.~\eqref{eq:gamma_h_app}. For periodic boundary conditions, the correspondence is sector dependent.}
\label{fig:KW_cluster_XY_scheme}
\end{figure}

\paragraph*{Direct cluster symbol and orientation.}
The direct Jordan--Wigner representation of Eq.~\eqref{eq:H_cluster_ising_app} has Laurent symbol
\begin{equation}
\label{eq:F_CI_direct_app}
F_{\rm CI}(z)=g_x-g_{zz}z+g_{zxz}z^2 .
\end{equation}
This form is useful because the cluster interaction is the \(z^2\) term. However, it is not the same orientation as the Hamiltonian-level XY representative above. Using Eq.~\eqref{eq:gamma_h_app},
\begin{equation}
\label{eq:F_CI_shifted_XY_app}
F_{\rm CI}(z)
=
(g_x+g_{zxz})\,z\,f_{\gamma,-h}(z^{-1}) .
\end{equation}
Thus the direct cluster symbol is a shifted, inverted XY representative with field \(-h\). This convention changes the location of zeros in the Laurent representative, but it does not change the SRE after the correct Pauli algebra has been matched. This is the only point that must be tracked carefully: the Hamiltonian-level map identifies the physical XY parameters, while the direct symbol determines the shifted interval representative.

\paragraph*{Two useful trajectories.}
For the MPS skeleton one obtains
\[
\gamma=\frac{2g}{1+g^2},
\qquad
h=\frac{1-g^2}{1+g^2},
\qquad
h^2+\gamma^2=1.
\]
Equivalently, \(g=\tan(\theta/2)\) gives \(h=\cos\theta\) and \(\gamma=\sin\theta\). Hence the skeleton maps to the disorder circle of the dual XY parameter space. It should not be confused with the anisotropic XY critical line \(h=1\). In the direct cluster-symbol convention, the same trajectory factorizes as \(F_{\rm skel}(z)=[(1+g)-(1-g)z]^2\).

A distinct trajectory is the cluster--Ising critical line
\[
(g_{zxz},g_{zz},g_x)=(g_c,2,2-g_c),
\qquad
0\le g_c\le2 .
\]
Under Eq.~\eqref{eq:gamma_h_app}, this gives \(h=1\), \(\gamma=1-g_c\). Therefore, the physical cluster critical line maps to the anisotropic XY critical line. If one restricts to \(\gamma\ge0\), the segment \(0\le g_c\le1\) is represented directly; the segment \(1\le g_c\le2\) is obtained by the orientation-reversed convention \(\gamma\mapsto-\gamma\). The midpoint \(g_c=1\) maps to the XX endpoint \((h,\gamma)=(1,0)\), where the fixed-\(\gamma\) critical expansion is nonuniform. In the direct symbol convention, this critical line is \(F_{\rm CI}^{\rm crit}(z)=(2-g_c)-2z+g_cz^2\), so \(F_{\rm CI}^{\rm crit}(1)=0\) for all \(g_c\). This zero is the same criticality, expressed in the shifted representative \(z f_{\gamma,-h}(z^{-1})\).

\begin{table}[!ht]
\caption{Two cluster--Ising trajectories and their Hamiltonian-level XY images. The direct cluster symbol realizes the shifted representative \(z f_{\gamma,-h}(z^{-1})\), so its zero structure must be interpreted together with the orientation convention.}
\label{tab:cluster-trajectories-xy-images}
\begin{ruledtabular}
\begin{tabular}{lll}
\rowcolor{blue!10}
Trajectory & Couplings & XY image \\
\hline
\rowcolor{blue!3}
MPS skeleton &
\((g_{zxz},g_{zz},g_x)=((g-1)^2,\,2(1-g^2),\,(1+g)^2)\) &
\(h=\frac{1-g^2}{1+g^2}\), \(\gamma=\frac{2g}{1+g^2}\), hence \(h^2+\gamma^2=1\) \\
Cluster critical line &
\((g_{zxz},g_{zz},g_x)=(g_c,\,2,\,2-g_c)\) &
\(h=1\), \(\gamma=1-g_c\) \\
\end{tabular}
\end{ruledtabular}
\end{table}

The Table~\ref{tab:cluster-trajectories-xy-images} summarizes the distinction needed in the main text: the MPS skeleton maps to the circle \(h^2+\gamma^2=1\), whereas the physical cluster critical line maps to \(h=1\). They meet only at the singular point \((h,\gamma)=(1,0)\), corresponding to \(g=0\) on the skeleton and \(g_c=1\) on the critical line.

\paragraph*{Periodic sectors, interval algebras, and SRE invariance.}
For a finite periodic chain, the Kramers--Wannier map must be supplemented by global sector data. From \(\tau_j^z=\sigma_j^z\sigma_{j+1}^z\) and periodicity of the original spins, one obtains \(\prod_{j=1}^L\tau_j^z=1\). The dual variables are therefore constrained, while the original global symmetry \(\prod_j\sigma_j^x\) fixes the boundary condition, or twist, of the dual XY chain. Thus the periodic cluster--XY correspondence is a sector-by-sector statement.

For an interval in the infinite chain, there is no NS/R sector, but the interval algebra has endpoints. The monomial shift in Eq.~\eqref{eq:F_CI_shifted_XY_app} cannot be absorbed by cyclicity and instead changes the Toeplitz truncation. This is why the subsystem correspondence uses the cluster-dual shifted block, denoted \(XYR\) in the main text, rather than the ordinary \(XY\) interval block.

Finally, the SRE is invariant under the matched Kramers--Wannier map because the duality is Clifford on the relevant Pauli algebra. If \(\rho_{\rm XY}=U_{\rm KW}\rho_{\rm CI}U_{\rm KW}^\dagger\), then \(U_{\rm KW}^\dagger P U_{\rm KW}=\varphi(P)\) for a bijection \(\varphi\) of Pauli strings. Hence the Pauli probability distributions differ only by a permutation, and for \(\alpha>0\), \(\alpha\neq1\),
\begin{equation}
\label{eq:SRE_KW_app}
M_\alpha(\rho_{\rm XY})=M_\alpha(\rho_{\rm CI}) .
\end{equation}
The same conclusion holds for the Shannon limit \(\alpha\to1\). Combining this invariance with Eq.~\eqref{eq:gamma_h_app} gives the cluster--Ising SRE correspondences used in the main text: periodic chains require sector matching, while subsystems require the shifted interval algebra described in Appendix~\ref{app:subsystem_toeplitz_pfaffian}.
\section{Decimation mechanisms and shifted XY blocks}
\label{app:xy_decimation}

This appendix gives the block-decomposition mechanisms used in the main text. Here \emph{decimation} means a minor-preserving decomposition of the Gaussian correlation matrix: after row and column permutations, and harmless diagonal gauge transformations, the matrix splits into smaller blocks. This is different from restricting the XY Hamiltonian to a special line such as the TFI line \(\gamma=1\) or the XX line \(\gamma=0\). The Gaussian-symbol conventions follow the standard free-fermion solution of the XY chain~\cite{Lieb1961,BarouchMcCoyDresden1970,BarouchMcCoy1971}, while the interval Toeplitz formulation follows the correlation-matrix approach to reduced Gaussian states~\cite{Peschel2003,PeschelEisler2009}.

The key distinction is geometric. On a periodic chain, monomial shifts of the Laurent symbol can often be absorbed by cyclicity, diagonal gauges, and sector matching. On an interval, the endpoints remember the shift. Thus, the same bulk folding \(z\mapsto z^n\) gives equivalent periodic blocks, but in the subsystem geometry it can produce ordinary XY blocks, shifted XY blocks, and rectangular boundary-unbalanced blocks.

We use \(\mathbf A\simeq_{\rm min}\mathbf B\) to denote equivalence under independent row permutations, column permutations, and diagonal unitary gauges. This equivalence preserves all absolute values of minors and therefore all functionals \(\operatorname{Det}_\beta\). For rectangular blocks, we use the auxiliary notation
\begin{equation}
\label{eq:rectangular-Detbeta-app}
\mathbf{Det}_{\beta}(\mathbf A)
=
\sum_{k=0}^{\min(p,q)}
\sum_{\substack{I\subseteq\{1,\dots,p\},\,J\subseteq\{1,\dots,q\}\\ |I|=|J|=k}}
\left|\det \mathbf A[I,J]\right|^\beta ,
\qquad
\mathbf A\in\mathbb C^{p\times q}.
\end{equation}
For \(\beta=2\), Cauchy--Binet gives \(\operatorname{Det}_2(\mathbf A)=\det(\mathbf I_q+\mathbf A^\dagger\mathbf A)=\det(\mathbf I_p+\mathbf A\mathbf A^\dagger)\). The corresponding normalized contribution is
\begin{equation}
\label{eq:rectangular-Malpha-app}
M_\alpha(\mathbf A)=\frac{1}{1-\alpha}
\ln
\frac{\mathbf{Det}_{2\alpha}(\mathbf A)}
{\mathbf{Det}_{2}(\mathbf A)} .
\end{equation}
For square Gaussian subsystem blocks, this is the SRE used in the main text. For rectangular blocks, it should be read as the algebraic contribution produced by the minor factorization. The different levels at which decimation acts in the constructions below are summarized in Table~\ref{tab:decimation-levels-app}.

\begin{table}[!ht]
\caption{Levels of decimation used in this appendix. The symbol remembers the algebraic lift, a finite ring sees the folding modulo \(L\), and an interval also remembers the boundary shift.}
\label{tab:decimation-levels-app}
\begin{ruledtabular}
\begin{tabular}{lll}
\rowcolor{blue!10}
Level & Input data & Decimation statement \\
\hline
\rowcolor{blue!3}
Symbol level & Laurent symbol \(F(z)\) & \(F(z)=Cz^r b(z^n)\) \\
Phase level & \(F(e^{i\theta})/|F(e^{i\theta})|\) & \(e^{ir\theta}H(e^{in\theta})\) \\
\rowcolor{blue!3}
Finite PBC matrix & sampled phase on \(K_{\rm NS/R}\) & \(g=\gcd(L,n)\) equivalent blocks \\
Infinite subsystem & Toeplitz truncation & ordinary, shifted, or rectangular blocks \\
\end{tabular}
\end{ruledtabular}
\end{table}

\paragraph*{Finite periodic criterion.}
Fix a ring of length \(L\) and a momentum grid \(K_\nu\), with \(\nu={\rm NS}\) or \({\rm R}\). Given phase data \(s_k\in U(1)\), the finite Gaussian matrix has the form
\begin{equation}
\label{eq:finite-correlation-matrix-app}
G^{(L)}_{n,m}
=
\frac1L
\sum_{\theta_k\in K_\nu}
s_k e^{i\theta_k(n-m)}
=
c_{n-m}.
\end{equation}
In a free-fermion chain, \(s_k=\omega_F(\theta_k)=F(e^{i\theta_k})/|F(e^{i\theta_k})|\).

\begin{theorem}[Finite-size decimation criterion]
\label{thm:finite-size-decimation-criterion-app}
Let \(g\mid L\), with \(g\ge2\). The following statements are equivalent: \emph{(i)} the Fourier coefficients satisfy \(c_d=0\) unless \(d\equiv-r\pmod g\); \emph{(ii)} after grouping row and column indices by residue classes modulo \(g\), the matrix is minor-equivalent to a direct sum of \(g\) equal blocks; \emph{(iii)} the sampled phase has the folded form \(s_k=e^{ir\theta_k}H(e^{ig\theta_k})\) for some function \(H\) on the folded grid.
\end{theorem}

\begin{proof}
If \(c_d\) is supported on a single residue class \(d\equiv-r\pmod g\), then \(G^{(L)}_{n,m}\) can be nonzero only when \(n-m\equiv-r\pmod g\). Grouping rows and columns by residues therefore produces \(g\) independent residue blocks, equal up to diagonal gauges because the entries depend only on \(n-m\). Conversely, if such a residue decomposition holds, the kernel must vanish outside one residue class. Finally, if \(c_d\) is supported on \(d=-r+tg\), then \(s_k=\sum_d c_d e^{-i\theta_kd}=e^{ir\theta_k}\sum_t c_{-r+tg}e^{-ig\theta_kt}\), which has the folded form. The converse follows by inserting \(s_k=e^{ir\theta_k}H(e^{ig\theta_k})\) into the inverse finite Fourier transform.
\end{proof}

The simplest way to produce the folded form is a symbol lift. If \(F(z)=Cz^r b(z^n)\), then on the unit circle \(\omega_F(\theta)=e^{i(r\theta+\psi)}\omega_b(n\theta)\), with \(e^{i\psi}=C/|C|\). Hence the \(L\)-site periodic matrix decomposes into \(g=\gcd(L,n)\) equivalent blocks. At the level of Laurent exponents, this means that all nonzero powers of \(z\) in \(F\) lie in one residue class modulo \(n\).

\paragraph*{Rigidity inside the nearest-neighbor XY family.}
The nearest-neighbor XY symbol is
\begin{equation}
\label{eq:xy-symbol-decimation-app}
f_{\gamma,h}(z)
=
h+\frac{1+\gamma}{2}z+\frac{1-\gamma}{2}z^{-1}.
\end{equation}
Its exponent support is contained in \(\{-1,0,1\}\). For \(h\neq0\), the three consecutive exponents cannot lie in one nontrivial residue class modulo \(n\ge2\). On the TFI orientations \(\gamma=\pm1\), the support is one of the adjacent pairs \(\{0,1\}\) or \(\{-1,0\}\), which again does not give a nontrivial lift, except at monomial endpoints. The only non-monomial way to remove the middle exponent is the zero-field line \(h=0\). The resulting decimation classes inside the nearest-neighbor XY family are collected in Table~\ref{tab:xy-internal-decimation-app}.

\begin{table}[!ht]
\caption{Decimation classes inside the nearest-neighbor XY family. The only non-monomial reducible line is \(h=0\).}
\label{tab:xy-internal-decimation-app}
\begin{ruledtabular}
\begin{tabular}{llll}
\rowcolor{blue!10}
Subfamily & Support & Decimates? & Reduced class \\
\hline
\rowcolor{blue!3}
Generic XY, \(h\neq0\) & \(\{-1,0,1\}\) & No & irreducible \\
TFI, \(\gamma=1\), \(h\neq0\) & \(\{0,1\}\) & No & irreducible \\
\rowcolor{blue!3}
Opposite TFI, \(\gamma=-1\), \(h\neq0\) & \(\{-1,0\}\) & No & irreducible \\
Zero-field XY, \(h=0\) & \(\{-1,1\}\) & Yes & TFI with \(\widetilde h=\frac{1-\gamma}{1+\gamma}\) \\
\rowcolor{blue!3}
Zero-field XX, \((\gamma,h)=(0,0)\) & \(\{-1,1\}\) & Yes & critical TFI, \(\widetilde h=1\) \\
Monomial endpoints & one exponent & Trivial & gauge/shift only \\
\end{tabular}
\end{ruledtabular}
\end{table}

\begin{theorem}[Internal decimation of the XY family]
\label{thm:xy-internal-decimation-app}
Among non-monomial XY symbols, the only nontrivially decimating submanifold is the zero-field line \(h=0\). For \(h=0\) and \(\gamma\neq-1\),
\begin{equation}
\label{eq:zero-field-xy-tfi-factorization-app}
f_{\gamma,0}(z)
=
\frac{1+\gamma}{2}z^{-1}
f_{1,\widetilde h}(z^2),
\qquad
\widetilde h=\frac{1-\gamma}{1+\gamma}.
\end{equation}
Thus zero-field XY is a parity lift of TFI.
\end{theorem}

\begin{proof}
A nontrivial symbol-level decimation requires all nonzero exponents to lie in one residue class modulo some \(n\ge2\). The support discussion above excludes all non-monomial cases with \(h\neq0\). If \(h=0\), the support is \(\{-1,1\}\), and these exponents are congruent modulo \(2\). The factorization follows directly from \(f_{\gamma,0}(z)=\frac{1+\gamma}{2}z+\frac{1-\gamma}{2}z^{-1}=\frac{1+\gamma}{2}z^{-1}(z^2+\widetilde h)\).
\end{proof}

The map is \((\gamma,0)\mapsto(1,\widetilde h)\). In particular, the zero-field XX point maps to the critical TFI point. The endpoint \(\gamma=-1\) is excluded from Eq.~\eqref{eq:zero-field-xy-tfi-factorization-app} only because the prefactor \(1+\gamma\) vanishes; there the symbol is the monomial \(z^{-1}\).

\paragraph*{Periodic and subsystem consequences of zero-field decimation.}
The factorization in Eq.~\eqref{eq:zero-field-xy-tfi-factorization-app} implies the phase relation \(\omega_{\gamma,0}(\theta)=e^{-i\theta}\omega_{1,\widetilde h}(2\theta)\), up to an irrelevant constant phase. For even \(L\), the periodic matrix therefore decomposes, after compatible sector matching, into two equivalent TFI blocks:
\begin{equation}
\label{eq:xy-tfi-block-decomposition-app}
G_{\rm XY}^{\rm PBC}(L;0,\gamma)
\simeq_{\rm min}
G_{\rm TFI}^{\rm PBC}(L/2;1,\widetilde h)
\oplus
G_{\rm TFI}^{\rm PBC}(L/2;1,\widetilde h).
\end{equation}
Consequently,
\begin{equation}
\label{eq:zero-field-pbc-SRE-app}
M_\alpha^{{\rm XY},{\rm PBC}}(L;0,\gamma)
=
2M_\alpha^{{\rm TFI},{\rm PBC}}(L/2;1,\widetilde h).
\end{equation}
For finite periodic spin chains this identity is sector dependent because the Jordan--Wigner transformation introduces boundary-condition sectors. With the conventions of the main text, the compatible zero-field branch is Neveu--Schwarz for \(L=0\mod 4\) and Ramond for \(L=2\mod 4\), in agreement with the finite-size sector structure of the XY chain~\cite{DePasquale2009}.

For subsystems of the infinite chain there is no global NS/R sector. Instead, the interval endpoints determine the sizes of the two parity blocks. Let \(g_r^{\rm TFI}(h):=g_r(h,1)\), and define the rectangular TFI Toeplitz block
\begin{equation}
\label{eq:rectangular-TFI-block-app}
G_{p,q}^{\rm TFI}(h)
=
\bigl[g_{a-b}^{\rm TFI}(h)\bigr]_{a=0,\dots,p-1}^{b=0,\dots,q-1}.
\end{equation}

\begin{theorem}[Subsystem decimation of zero-field XY]
\label{thm:zero-field-xy-subsystem-decimation-app}
Let \(-1<\gamma\le1\), set \(\widetilde h=(1-\gamma)/(1+\gamma)\), and consider an interval of length \(\ell\) in the infinite zero-field XY chain. If the interval starts on one fixed sublattice, define \(\ell_{\rm o}=\lceil\ell/2\rceil\) and \(\ell_{\rm e}=\lfloor\ell/2\rfloor\). Then, for every \(\beta>0\),
\begin{equation}
\label{eq:zero-field-xy-subsystem-Detbeta-app}
\mathbf{Det}_{\beta}
G_{\ell}^{{\rm XY},(\infty)}(0,\gamma)
=
\mathbf{Det}_{\beta}
G_{\ell_{\rm o},\ell_{\rm e}}^{\rm TFI}(\widetilde h)
\,
\mathbf{Det}_{\beta}
G_{\ell_{\rm o},\ell_{\rm e}}^{\rm TFI}(\widetilde h^{-1}) .
\end{equation}
Consequently,
\begin{equation}
\label{eq:zero-field-xy-subsystem-Malpha-all-ell-app}
M_\alpha^{{\rm XY},(\infty)}(\ell;0,\gamma)
=
M_\alpha\!\left[G_{\ell_{\rm o},\ell_{\rm e}}^{\rm TFI}(\widetilde h)\right]
+
M_\alpha\!\left[G_{\ell_{\rm o},\ell_{\rm e}}^{\rm TFI}(\widetilde h^{-1})\right].
\end{equation}
If the interval starts on the opposite sublattice, the two rectangular blocks are transposed, and the value of \(M_\alpha\) is unchanged.
\end{theorem}

\begin{proof}
The relation \(\omega_{\gamma,0}(\theta)=e^{-i\theta}\omega_{1,\widetilde h}(2\theta)\) implies that the Fourier coefficients of the zero-field XY phase live on one parity class. Grouping odd and even sites gives a bipartite Toeplitz matrix,
\begin{equation}
\label{eq:zero-field-xy-bipartite-toeplitz-app}
\Pi
G_{\ell}^{{\rm XY},(\infty)}(0,\gamma)
\Pi^{T}
=
\begin{pmatrix}
0 & A_{\ell_{\rm o},\ell_{\rm e}}(\widetilde h)\\
B_{\ell_{\rm e},\ell_{\rm o}}(\widetilde h) & 0
\end{pmatrix}.
\end{equation}
After an independent permutation of columns, this matrix is minor-equivalent to \(A_{\ell_{\rm o},\ell_{\rm e}}(\widetilde h)\oplus B_{\ell_{\rm e},\ell_{\rm o}}(\widetilde h)\). Up to transposition and diagonal gauges, the two blocks are rectangular TFI blocks at dual fields: \(A_{\ell_{\rm o},\ell_{\rm e}}(\widetilde h)\simeq_{\rm min}G_{\ell_{\rm o},\ell_{\rm e}}^{\rm TFI}(\widetilde h)\) and \(B_{\ell_{\rm e},\ell_{\rm o}}(\widetilde h)\simeq_{\rm min}[G_{\ell_{\rm o},\ell_{\rm e}}^{\rm TFI}(\widetilde h^{-1})]^T\). Multiplicativity of \(\mathbf{Det}_{\beta}\) under direct sums gives Eq.~\eqref{eq:zero-field-xy-subsystem-Detbeta-app}, and Eq.~\eqref{eq:zero-field-xy-subsystem-Malpha-all-ell-app} follows from Eq.~\eqref{eq:rectangular-Malpha-app}.
\end{proof}

For even \(\ell=2m\), Eq.~\eqref{eq:zero-field-xy-subsystem-Malpha-all-ell-app} reduces to
\begin{equation}
\label{eq:zero-field-even-subsystem-SRE-app}
M_\alpha^{{\rm XY},(\infty)}(2m;0,\gamma)
=
M_\alpha^{{\rm TFI},(\infty)}(m;1,\widetilde h)
+
M_\alpha^{{\rm TFI},(\infty)}(m;1,\widetilde h^{-1}).
\end{equation}
For odd \(\ell=2m+1\),
\begin{equation}
\label{eq:zero-field-odd-subsystem-SRE-app}
M_\alpha^{{\rm XY},(\infty)}(2m+1;0,\gamma)
=
M_\alpha\!\left[G_{m+1,m}^{\rm TFI}(\widetilde h)\right]
+
M_\alpha\!\left[G_{m+1,m}^{\rm TFI}(\widetilde h^{-1})\right].
\end{equation}
Thus odd intervals do not give two ordinary square TFI subsystems. They give two rectangular, boundary-unbalanced TFI blocks.

\paragraph*{Lifted Laurent symbols.}
The general class of symbols that decimate to XY blocks is
\begin{equation}
\label{eq:symbol-decimation-to-xy-app}
F(z)=Cz^r f_{\gamma,h}(z^n),
\qquad
n\ge2,\quad r\in\mathbb Z,\quad C\in\mathbb C^\times .
\end{equation}
Equivalently,
\begin{equation}
\label{eq:xy-lift-expanded-app}
F(z)
=
C\left[
\frac{1-\gamma}{2}z^{r-n}
+
h z^r
+
\frac{1+\gamma}{2}z^{r+n}
\right].
\end{equation}

\begin{theorem}[Symbol-level classification of XY lifts]
\label{thm:all-symbols-decimating-to-xy-app}
A three-term Laurent symbol decimates to an XY block with parameters \((h,\gamma)\) if and only if its nonzero exponents are contained in an arithmetic progression \(r-n,r,r+n\), and its coefficients are proportional to \(\bigl((1-\gamma)/2,h,(1+\gamma)/2\bigr)\). Equivalently, it has the form in Eq.~\eqref{eq:symbol-decimation-to-xy-app}.
\end{theorem}

\begin{proof}
Expanding \(Cz^r f_{\gamma,h}(z^n)\) gives Eq.~\eqref{eq:xy-lift-expanded-app}. Conversely, if the exponent pattern and coefficient ratios have this form, factoring out \(Cz^r\) gives \(f_{\gamma,h}(z^n)\).
\end{proof}

At the phase level the class can be larger, because distinct Laurent polynomials may have the same phase on the unit circle. The sufficient condition needed for correlation matrices is
\[
\omega_F(\theta)
=
e^{i(r\theta+\psi)}
\omega_{\gamma,h}(n\theta),
\qquad
e^{i\psi}=C/|C|.
\]
Then, on an \(L\)-site ring, the correlation matrix decomposes into \(g=\gcd(L,n)\) equivalent blocks of size \(L/g\). After sector matching,
\begin{equation}
\label{eq:general-lift-PBC-SRE-app}
M_\alpha^{F,{\rm PBC}}(L)
=
g\,M_\alpha^{{\rm XY},{\rm PBC}}
\left(\frac Lg;h,\gamma\right).
\end{equation}

\paragraph*{Subsystems of lifted symbols.}
For intervals, the same lift remembers the shift and the endpoint imbalance. Let \(F(z)=Cz^r f_{\gamma,h}(z^n)\), and write \(r=n\kappa+\rho\), with \(0\le\rho<n\). Let \(\ell=nm+s\), with \(0\le s<n\), and assume that the interval starts in residue class \(0\). Define
\begin{equation}
\label{eq:residue-sector-sizes-app}
N_a=
\begin{cases}
m+1, & 0\le a<s,\\
m, & s\le a<n,
\end{cases}
\qquad
a=0,\dots,n-1.
\end{equation}
For each row residue \(a\), let \(b_a\) be the unique residue satisfying \(b_a\equiv a-\rho\pmod n\), and set
\begin{equation}
\label{eq:residue-shift-nu-app}
\nu_a=
\kappa+
\begin{cases}
1, & 0\le a<\rho,\\
0, & \rho\le a<n .
\end{cases}
\end{equation}
We use the rectangular shifted XY block
\begin{equation}
\label{eq:rectangular-shifted-block-app}
\mathcal G_{p,q}^{[\nu]}(h,\gamma)
=
\bigl[g_{a-b-\nu}(h,\gamma)\bigr]_{a=0,\dots,p-1}^{b=0,\dots,q-1}.
\end{equation}

\begin{theorem}[Subsystem decomposition of lifted symbols]
\label{thm:general-lift-subsystem-decomposition-app}
With the notation above,
\begin{equation}
\label{eq:general-lift-subsystem-blocks-all-ell-app}
G_{\ell}^{F,(\infty)}
\simeq_{\rm min}
\bigoplus_{a=0}^{n-1}
\mathcal G_{N_a,N_{b_a}}^{[\nu_a]}(h,\gamma).
\end{equation}
Therefore, for every \(\beta>0\),
\begin{equation}
\label{eq:general-lift-subsystem-Detbeta-all-ell-app}
\mathbf{Det}_{\beta}
G_{\ell}^{F,(\infty)}
=
\prod_{a=0}^{n-1}
\mathbf{Det}_{\beta}
\mathcal G_{N_a,N_{b_a}}^{[\nu_a]}(h,\gamma),
\end{equation}
and
\begin{equation}
\label{eq:general-lift-subsystem-Malpha-all-ell-app}
M_\alpha^{F,(\infty)}(\ell)
=
\sum_{a=0}^{n-1}
M_\alpha\!\left[
\mathcal G_{N_a,N_{b_a}}^{[\nu_a]}(h,\gamma)
\right].
\end{equation}
\end{theorem}

\begin{proof}
The lifted phase satisfies \(\omega_F(\theta)=e^{i(r\theta+\psi)}\omega_{\gamma,h}(n\theta)\). Hence its Fourier coefficients vanish unless \(d\equiv r\pmod n\), and for \(d=r+nt\) they are equal, up to the global phase \(e^{i\psi}\), to \(g_t(h,\gamma)\). Group the interval sites as \(j=a+np\), with residue \(a=0,\dots,n-1\). A row in residue \(a\) can have nonzero entries only in the column residue \(b_a\equiv a-\rho\pmod n\). If \(a\ge\rho\), then \(j-k-r=n(p-q-\kappa)\); if \(a<\rho\), then \(j-k-r=n(p-q-\kappa-1)\). Thus the corresponding block is \(\mathcal G_{N_a,N_{b_a}}^{[\nu_a]}(h,\gamma)\). The determinant and entropy identities follow from minor-equivalence and multiplicativity under direct sums.
\end{proof}

If \(\ell=nm\), all residue sectors have size \(m\), and Eq.~\eqref{eq:general-lift-subsystem-Malpha-all-ell-app} becomes
\begin{equation}
\label{eq:general-lift-subsystem-SRE-multiple-n-app}
M_\alpha^{F,(\infty)}(nm)
=
(n-\rho)M_\alpha\!\left[\mathcal G_{m,m}^{[\kappa]}(h,\gamma)\right]
+
\rho M_\alpha\!\left[\mathcal G_{m,m}^{[\kappa+1]}(h,\gamma)\right].
\end{equation}
In the main text, the most important case is \(0\le r<n\), so \(\kappa=0\) and \(\rho=r\). Then
\begin{equation}
\label{eq:general-lift-subsystem-main-text-form-app}
M_\alpha^{F,(\infty)}(nm)
=
(n-r)M_\alpha^{{\rm XY},(\infty)}(m;h,\gamma)
+
rM_\alpha^{{\rm XYL},(\infty)}(m;h,\gamma).
\end{equation}
The unshifted lift \(r=0\) gives \(n\) ordinary XY interval blocks. If \(\ell\) is not a multiple of \(n\), Eq.~\eqref{eq:general-lift-subsystem-Malpha-all-ell-app} is the correct replacement: the residue sectors have unequal sizes and the blocks may be rectangular.

For example, for \(F(z)=Cz f_{\gamma,h}(z^2)\) and \(\ell=2m+1\), one has \(n=2\), \(r=1\), \(\kappa=0\), \(\rho=1\), \(N_0=m+1\), and \(N_1=m\). Hence
\begin{equation}
\label{eq:n2-r1-odd-shifted-formula-app}
M_\alpha^{F,(\infty)}(2m+1)
=
M_\alpha\!\left[
\mathcal G_{m+1,m}^{[1]}(h,\gamma)
\right]
+
M_\alpha\!\left[
\mathcal G_{m,m+1}^{[0]}(h,\gamma)
\right].
\end{equation}
Thus one block is rectangular and shifted, while the other is rectangular and ordinary XY.

\paragraph*{Shifted interval representatives.}
The main text uses three square interval blocks, with \(X\in\{{\rm XY},{\rm XYL},{\rm XYR}\}\):
\begin{equation}
\label{eq:main-shifted-blocks-app}
\begin{aligned}
\bigl(G_\ell^{{\rm XY},(\infty)}(h,\gamma)\bigr)_{ab}
&=g_{a-b}(h,\gamma),\\
\bigl(G_\ell^{{\rm XYL},(\infty)}(h,\gamma)\bigr)_{ab}
&=g_{a-b-1}(h,\gamma),\\
\bigl(G_\ell^{{\rm XYR},(\infty)}(h,\gamma)\bigr)_{ab}
&=g_{1-a+b}(h,\gamma),
\end{aligned}
\qquad
a,b=0,\dots,\ell-1.
\end{equation}
The \(XYL\) block is the direct shifted representative \(zf_{\gamma,h}(z)\), while \(XYR\) is the Kramers--Wannier, or cluster-dual, shifted representative \(zf_{\gamma,h}(z^{-1})\). Equivalently,
\[
G_\ell^{{\rm XY},(\infty)}=\mathcal G_{\ell,\ell}^{[0]}(h,\gamma),
\qquad
G_\ell^{{\rm XYL},(\infty)}=\mathcal G_{\ell,\ell}^{[1]}(h,\gamma),
\qquad
G_\ell^{{\rm XYR},(\infty)}=\mathcal G_{\ell,\ell}^{[1]}(h,-\gamma).
\]

For a standalone shifted symbol, there is no residue decomposition unless the symbol is further lifted. Thus \(zf_{\gamma,h}(z)\) gives directly \(G_\ell^{{\rm XYL},(\infty)}\), while \(zf_{\gamma,h}(z^{-1})\) gives \(G_\ell^{{\rm XYR},(\infty)}\). Rectangular shifted blocks appear only when such a shift occurs inside a larger folded symbol whose residue sectors have unequal sizes.

This is the reason subsystem decimation is richer than periodic decimation. A ring can hide a monomial shift by wrapping the lattice around the boundary. An interval cannot: its endpoints determine which residue class is longer and which shifted representative appears.

\paragraph*{TFI self-duality and the \(XYR\) block.}
At \(\gamma=1\), the XY symbol is \(f_{1,h}(z)=h+z\). For \(h>0\),
\begin{equation}
\label{eq:app-ising-laurent-duality}
f_{1,h^{-1}}(z)
=
h^{-1}z f_{1,h}(z^{-1}).
\end{equation}
The positive factor \(h^{-1}\) does not affect the phase, while the factor \(z\) shifts the Fourier modes. Therefore
\begin{equation}
\label{eq:app-ising-coefficient-duality}
g_r(h^{-1},1)=g_{1-r}(h,1),
\end{equation}
and
\begin{equation}
\label{eq:app-ising-G-dual}
G_\ell^{{\rm TFI},(\infty)}(h^{-1})
=
\bigl[g_{a-b}(h^{-1},1)\bigr]_{a,b=0}^{\ell-1}
=
\bigl[g_{1-a+b}(h,1)\bigr]_{a,b=0}^{\ell-1}
=
G_\ell^{{\rm XYR},(\infty)}(h,1).
\end{equation}
Thus TFI self-duality maps a finite interval to a shifted Toeplitz truncation. At \(\gamma=1\), this shifted block is again an ordinary TFI block at the dual field. Away from \(\gamma=1\), however, the same shifted object is
\begin{equation}
\label{eq:app-general-shifted-symbol}
zf_{\gamma,h}(z^{-1})
=
\frac{1+\gamma}{2}
+
hz
+
\frac{1-\gamma}{2}z^2,
\end{equation}
and its interval matrix is \(G_\ell^{{\rm XYR},(\infty)}(h,\gamma)\). Therefore, for \(\gamma\neq1\), the cluster-dual interval object should not be replaced by an ordinary XY subsystem at a dual field. It is the shifted XYR block. The subsystem blocks used in the main text are summarized in Table~\ref{tab:shifted-xy-block-dictionary-app}.

\begin{table}[!ht]
\caption{Subsystem blocks used in the main text. The labels \(XYL\) and \(XYR\) are interval-block labels, not new Hamiltonian families.}
\label{tab:shifted-xy-block-dictionary-app}
\begin{ruledtabular}
\begin{tabular}{llll}
\rowcolor{blue!10}
Label & Symbol representative & Matrix entries & Interpretation \\
\hline
\rowcolor{blue!3}
\(XY\) & \(f_{\gamma,h}(z)\) &
\([g_{a-b}(h,\gamma)]\) & ordinary XY interval \\
\(XYL\) & \(zf_{\gamma,h}(z)\) &
\([g_{a-b-1}(h,\gamma)]\) & direct shifted interval \\
\rowcolor{blue!3}
\(XYR\) & \(zf_{\gamma,h}(z^{-1})\) &
\([g_{1-a+b}(h,\gamma)]\) & KW/cluster-dual shifted interval \\
\end{tabular}
\end{ruledtabular}
\end{table}

\section{Numerical checks and use of the decimation correspondences}
\label{app:numerical-decimation-checks}

In this appendix, we give explicit finite-size checks of the correspondence formulas used in Sec.~\ref{sec:generating-functions-correspondences}. The purpose is twofold. First, the examples verify the equalities directly from the definition of the absolute-minor functional. Second, they illustrate how the formulas should be used in practice, especially the distinction between periodic chains and finite intervals.

For a finite \(p\times q\) matrix \(A\), we use the rectangular extension
\[
\mathbf{Det}_{\beta}(A)
=
\sum_{k=0}^{\min(p,q)}
\sum_{\substack{I\subseteq[p],\,J\subseteq[q]\\ |I|=|J|=k}}
|\det A[I,J]|^\beta .
\]
At \(\alpha=\frac12\), we write
\[
M_{\frac12}[A]
=
2\log \mathbf{Det}_{1}(A)
-
2\log \mathbf{Det}_{2}(A),
\qquad
\mathbf{Det}_{2}(A)
=
\det(\mathbf 1+AA^\dagger).
\]
For square Gaussian correlation matrices this coincides with the stabilizer entropy \(M_{\frac12}\) used in the main text. For rectangular shifted blocks it is only a convenient block contribution; additivity follows from the direct-sum decomposition.

All numbers below were obtained by brute-force summation over all minors for \(\operatorname{Det}_1\). Thus these checks do not use the product formulas of Sec.~\ref{sec:pbc-exact-results} or the Pfaffian formulas of Sec.~\ref{sec:subsystem-exact-results}. The entries of the periodic matrices were computed from the finite Fourier sums, while the entries of the infinite-chain Toeplitz blocks were obtained from the Fourier coefficients \(g_r(h,\gamma)\) of the XY phase symbol. The reported error is the absolute difference between the two sides of the claimed correspondence.

\paragraph*{Zero-field XY as a parity lift of TFI.}
The zero-field XY symbol satisfies \(f_{\gamma,0}(z)=\frac{1+\gamma}{2}z^{-1}f_{1,\widetilde h}(z^2)\), with \(\widetilde h=(1-\gamma)/(1+\gamma)\). Therefore, for a compatible periodic branch,
\[
M_{\frac12}^{\mathrm{XY,PBC}}(L;0,\gamma)
=
2M_{\frac12}^{\mathrm{TFI,PBC}}\!\left(\frac L2;\widetilde h\right).
\]
Here \(M_{\frac12}^{\mathrm{XY,PBC}}(L;h,\gamma)\) uses the parameter order \((h,\gamma)\), while \(M_{\frac12}^{\mathrm{TFI,PBC}}(L;h)\) denotes the TFI line \(\gamma=1\) at field \(h\). For an infinite interval of even length \(\ell=2m\), the interval endpoints retain the monomial shift, and the two parity blocks are dual TFI intervals:
\[
M_{\frac12}^{\mathrm{XY},(\infty)}(2m;0,\gamma)
=
M_{\frac12}^{\mathrm{TFI},(\infty)}(m;\widetilde h)
+
M_{\frac12}^{\mathrm{TFI},(\infty)}\!\left(m;\widetilde h^{-1}\right).
\]
The direct checks are
\[
\begin{array}{c|c|c|c}
\text{case} & \text{parameters} & \text{left-hand side} & \text{right-hand side} \\ \hline
\mathrm{PBC} & L=8,\ \gamma=0.6,\ \widetilde h=0.25
& 0.967447765980 & 0.967447765980 \\
(\infty)\text{ interval} & \ell=4,\ \gamma=0.3,\ \widetilde h=0.5384615385
& 0.831674280106 & 0.831674280106 \\
(\infty)\text{ interval} & \ell=6,\ \gamma=0.6,\ \widetilde h=0.25
& 0.610215639044 & 0.610215639044
\end{array}
\]
The largest discrepancy in these examples is below \(2\times10^{-14}\).

\paragraph*{Lifted Laurent symbols.}
Consider a lifted XY symbol \(F(z)=Cz^r f_{\gamma,h}(z^n)\), with \(n\ge2\) and \(r\in\mathbb Z\). For a periodic chain, the folding gives \(g=\gcd(L,n)\) identical blocks after sector matching:
\[
M_{\frac12}^{F,\mathrm{PBC}}(L)
=
g\,M_{\frac12}^{\mathrm{XY,PBC}}\!\left(\frac Lg;h,\gamma\right).
\]
The monomial shift \(z^r\) does not appear explicitly in the final periodic formula because it can be absorbed by cyclicity, gauge transformations, and sector relabeling. Direct checks are
\[
\begin{array}{c|c|c|c|c}
L & (n,r) & (h,\gamma) & M_{\frac12}^{F,\mathrm{PBC}}(L) &
gM_{\frac12}^{\mathrm{XY,PBC}}(L/g;h,\gamma) \\ \hline
8 & (2,1) & (1.3,0.4) & 1.638727942876 & 1.638727942876 \\
9 & (3,1) & (1.2,0.5) & 1.085383187597 & 1.085383187597 \\
8 & (4,3) & (1.4,0.7) & 1.261637551217 & 1.261637551217
\end{array}
\]
The largest discrepancy in these examples is below \(3\times10^{-13}\).

For an infinite interval, the monomial shift cannot be removed. If \(\ell=nm\) and \(r=n\kappa+\rho\), with \(0\le\rho<n\), then
\[
M_{\frac12}^{F,(\infty)}(nm)
=
(n-\rho)M_{\frac12}\!\left[G_{m,m}^{[\kappa]}(h,\gamma)\right]
+
\rho\,M_{\frac12}\!\left[G_{m,m}^{[\kappa+1]}(h,\gamma)\right],
\]
where
\[
G_{p,q}^{[\nu]}(h,\gamma)
=
[g_{a-b-\nu}(h,\gamma)]_{a=0,\dots,p-1}^{b=0,\dots,q-1}.
\]
For \(\ell=6\), \(n=3\), \(h=1.2\), and \(\gamma=0.5\), we obtain
\[
\begin{array}{c|c|c|c|c}
r & (\kappa,\rho) & M_{\frac12}^{F,(\infty)}(6) &
\text{block formula} & \text{error} \\ \hline
0 & (0,0) & 0.910865975275 & 0.910865975275 & 4.4\times10^{-16} \\
1 & (0,1) & 1.114152960598 & 1.114152960598 & 2.9\times10^{-15} \\
2 & (0,2) & 1.317439945922 & 1.317439945922 & 1.8\times10^{-15} \\
4 & (1,1) & 1.699817924175 & 1.699817924175 & 6.7\times10^{-16} \\
7 & (2,1) & 1.872757938774 & 1.872757938774 & 0
\end{array}
\]
The last line is conceptually useful: \(\kappa=2\) is a valid shifted block for a general lifted Laurent symbol, but it is not one of the standard shifted interval representatives used for the cluster correspondence. The labels \(XYR\) and \(XYL\) denote one-site shifted interval embeddings, whereas \(\kappa\) and \(\rho\) describe the arithmetic decomposition \(r=n\kappa+\rho\) of a general lift.

If \(\ell\) is not divisible by \(n\), the same construction produces rectangular blocks. For example, for \(F(z)=zf_{\gamma,h}(z^2)\), \(\ell=5\), \(h=1.3\), and \(\gamma=0.4\), the interval decomposes into two rectangular blocks of sizes \(3\times2\) and \(2\times3\). We find
\[
M_{\frac12}^{F,(\infty)}(5)
=
0.766181082424
=
M_{\frac12}\!\left[G_{3,2}^{[1]}(h,\gamma)\right]
+
M_{\frac12}\!\left[G_{2,3}^{[0]}(h,\gamma)\right],
\]
with agreement at the displayed precision.

\paragraph*{Cluster orientations and shifted XY blocks.}
The direct and Kramers--Wannier cluster orientations are represented, at the subsystem level, by the shifted XY blocks used in the main text:
\[
G_\ell^{{\rm XYL},(\infty)}(h,\gamma)
=
G_{\ell,\ell}^{[1]}(h,\gamma)
=
[g_{a-b-1}(h,\gamma)],
\]
and
\[
G_\ell^{{\rm XYR},(\infty)}(h,\gamma)
=
[g_{1-a+b}(h,\gamma)]
=
G_{\ell,\ell}^{[1]}(h,-\gamma).
\]
Thus \(XYL\) and \(XYR\) should be treated as interval-orientation labels, not as values of \(\kappa\). For periodic chains, both shifts are absorbed. For example, with \(L=6\), \(h=1.2\), and \(\gamma=0.5\),
\[
M_{\frac12}^{\mathrm{XYR,PBC}}(6;h,\gamma)
=
1.715717811791
=
M_{\frac12}^{\mathrm{XY,PBC}}(6;h,\gamma),
\]
and the same value is obtained for the direct shifted representative \(XYL\). For intervals the distinction is physical. With \(\ell=4\), \(h=1.2\), and \(\gamma=0.5\),
\[
M_{\frac12}\!\left[G_\ell^{{\rm XYR},(\infty)}(h,\gamma)\right]=1.152762586559,
\qquad
M_{\frac12}\!\left[G_\ell^{{\rm XYL},(\infty)}(h,\gamma)\right]=1.071055368083,
\]
whereas
\[
M_{\frac12}^{\mathrm{XY},(\infty)}(4;h,\gamma)=0.877147345537.
\]
Therefore, the interval correspondence must use the shifted block specified by the duality, rather than the ordinary XY Toeplitz truncation.

\paragraph*{TFI self-duality.}
The TFI relation \(f_{1,h}(z)=hzf_{1,h^{-1}}(z^{-1})\) gives the periodic identity
\[
M_{\frac12}^{\mathrm{TFI,PBC}}(L;h)
=
M_{\frac12}^{\mathrm{TFI,PBC}}\!\left(L;h^{-1}\right).
\]
For \(L=6\) and \(h=1.7\), both sides equal \(1.642534175554\). For finite intervals, the same duality must be read as a shifted-block identity:
\[
M_{\frac12}^{\mathrm{TFI},(\infty)}\!\left(\ell;h^{-1}\right)
=
M_{\frac12}\!\left[G_\ell^{{\rm XYR},(\infty)}(h,1)\right].
\]
For \(\ell=4\) and \(h=1.7\),
\[
M_{\frac12}^{\mathrm{TFI},(\infty)}(4;h^{-1})
=
1.084500717533
=
M_{\frac12}\!\left[G_4^{{\rm XYR},(\infty)}(h,1)\right].
\]
By contrast, the ordinary interval at \(h\) gives
\[
M_{\frac12}^{\mathrm{TFI},(\infty)}(4;h)
=
0.879668352707,
\]
which is different. This illustrates the main boundary effect: on a ring, the duality shift is invisible after sector matching, while on an interval, it changes the Toeplitz coefficients retained by the subsystem.

These finite-size checks confirm the two practical rules used throughout the paper. For periodic chains, use the folded block formula and absorb monomial shifts through sector matching. For finite intervals, keep the shift data explicitly through \(G_{p,q}^{[\nu]}\), \(G_\ell^{{\rm XYL},(\infty)}\), or \(G_\ell^{{\rm XYR},(\infty)}\). The labels \(XYL\) and \(XYR\) describe interval orientations, whereas \(\kappa\) and \(\rho\) describe the arithmetic decomposition \(r=n\kappa+\rho\) of a general lifted Laurent symbol.

\section{Stability chamber for the finite-periodic product}
\label{app:pbc-stability}

This appendix explains the finite-size branch condition behind the periodic product formula used in the main text. The point is not a new Hamiltonian constraint, but a sign constraint for the absolute-minor polynomial. At \(\alpha=\frac12\), the SRE depends on the sum of absolute values of all minors of the finite-periodic Gaussian matrix \(G^{\rm PBC}(h,\gamma)\). Since minors can change sign only by crossing their zero loci, this polynomial is analytic only chamber by chamber in the \((h,\gamma)\)-plane. The product formula used in the main text describes the chamber connected to the stable large-\(L\) region.

We work in the NS sector with \(L=2M\), as in the main text. The relevant generating function is
\begin{equation}
\label{eq:app-pbc-abs-minor-poly}
P_L^{\rm PBC}(u;h,\gamma)
=
\sum_{p=0}^{L}u^p
\sum_{\substack{S,T\subseteq[L]\\ |S|=|T|=p}}
\left|
\det G^{\rm PBC}[S,T]
\right| .
\end{equation}
Inside a connected chamber, every minor has a fixed sign, so the absolute values in Eq.~\eqref{eq:app-pbc-abs-minor-poly} select a fixed algebraic branch. Crossing a wall \(\det G^{\rm PBC}[S,T]=0\) can change the sign pattern and therefore the branch of \(P_L^{\rm PBC}\).

\paragraph*{The principal product branch.}
For finite \(L\), the branch used in the main text is detected by the patterned minor \(D_M(h,\gamma):=\det G^{\rm PBC}[I_M,J_M]\), with \(I_M=\{1,\dots,M\}\) and \(J_M=\{1,M+2,\dots,2M\}\). The zero set of \(D_M\) can contain several components. We denote by \(\gamma=\gamma_*(h;L)\) the component first reached when increasing \(\gamma\) from the small-\(\gamma\) side outside the circle \(h^2+\gamma^2=1\). The corresponding finite-size product chamber is
\begin{equation}
\label{eq:app-pbc-stability-chamber}
\mathcal R_L^{\rm PBC}
=
\left\{
(h,\gamma)\in\mathbb R_{\ge0}^{2}:
h^2+\gamma^2>1,\;
0\le \gamma<\gamma_*(h;L)
\right\}.
\end{equation}
The circle \(h^2+\gamma^2=1\) is the physical lower wall of the XY symbol, while the selected component of \(D_M(h,\gamma)=0\) is a finite-size sign wall of the absolute-minor polynomial. It is not an additional Hamiltonian critical line.

On this branch, the finite-periodic product reads
\begin{equation}
\label{eq:app-pbc-product-branch}
P_L^{\rm PBC}(u;h,\gamma)
=
\prod_{j=1}^{M}
\left[
1+u^2+2u\Lambda_j(h,\gamma)
\right],
\end{equation}
where \(\alpha_j=(2j-1)\pi/L\), \(\Delta_j(h,\gamma)=\sqrt{(h+\cos\alpha_j)^2+\gamma^2\sin^2\alpha_j}\), and
\begin{equation}
\label{eq:app-Lambda-j}
\Lambda_j(h,\gamma)
=
\frac{h+\gamma+(1-\gamma)\cos\alpha_j}
{\Delta_j(h,\gamma)} .
\end{equation}
Equivalently, if \(\cosh\varepsilon_j=\Lambda_j(h,\gamma)\), then \(P_L^{\rm PBC}(u;h,\gamma)=\prod_{j=1}^{M}(1+ue^{\varepsilon_j})(1+ue^{-\varepsilon_j})\). This resembles a free product over independent momentum-pair rapidities, but the factorization is not a consequence of momentum-space diagonalization alone. The original object is a site-space sum over absolute values of all minors, and the product requires a coherent sign branch inside \(\mathcal R_L^{\rm PBC}\).

\paragraph*{The exactly solvable \(L=4\) chamber.}
The first nontrivial case is \(L=4\), where the chamber walls can be derived explicitly. The NS momenta are \(q_0=\pi/4\), \(q_1=3\pi/4\), \(q_2=5\pi/4\), and \(q_3=7\pi/4\), and the matrix entries are \(G_{nm}^{\rm PBC}=\frac14\sum_{k=0}^{3}\omega(q_k)e^{iq_k(n-m)}\), with \(\omega(q)=(h+\cos q+i\gamma\sin q)/\sqrt{(h+\cos q)^2+\gamma^2\sin^2 q}\). The matrix is real, Toeplitz, and orthogonal, so its minors reduce to a small number of symmetry orbits.

Introduce \(R(h,\gamma)=\sqrt{(\gamma^2+2h^2-1)^2+4\gamma^2}\). Up to translations, reflections, and overall signs, the relevant representatives are
\begin{equation}
\label{eq:app-L4-representative-minors}
\begin{aligned}
&G_{14}=0 \quad\Longleftrightarrow\quad \gamma^2+2h^2=1,\\
&\det G^{\rm PBC}[\{1,2\},\{1,2\}]
=
\frac{\gamma^2+6h^2+R(h,\gamma)-3}{4R(h,\gamma)},\\
&\det G^{\rm PBC}[\{1,2\},\{3,4\}]
=
\frac{-3\gamma^2-2h^2+R(h,\gamma)+1}{4R(h,\gamma)},\\
&\det G^{\rm PBC}[\{1,2\},\{1,4\}]
=
\frac{\gamma^2-4\gamma-2h^2+R(h,\gamma)+1}{4R(h,\gamma)}.
\end{aligned}
\end{equation}
There is one additional orbit of \(2\times2\) minors equal to \(\pm\gamma h/R(h,\gamma)\), which never vanishes for \(h>0\) and \(\gamma>0\), and therefore does not generate a wall.

\begin{figure}[!ht]
\centering
\includegraphics[width=0.45\textwidth]{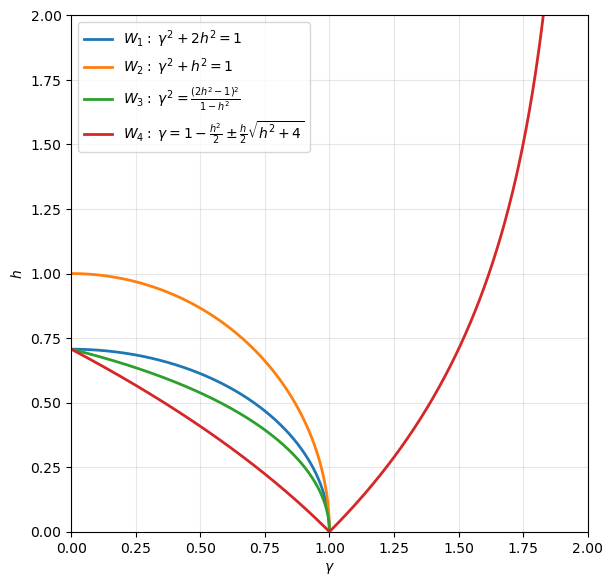}
\caption{
Complete \(L=4\) wall set in the positive \((\gamma,h)\)-plane, with equal axis scale. The orange curve \(W_2:\gamma^2+h^2=1\) is the circle wall. The red curve \(W_4\) is the selected patterned-minor branch, \(W_4=\gamma_*(h;4)\), and gives the exit wall of the principal \(L=4\) product chamber. The blue and green curves are additional finite-size sign walls associated with other minor orbits.
}
\label{fig:app-L4-wall-branches}
\end{figure}

The first wall follows from the \(1\times1\) representative minor and is \(W_1:\gamma^2+2h^2=1\). The far-separated minor gives \(R(h,\gamma)=3\gamma^2+2h^2-1\); after squaring and using the definition of \(R\), this becomes \(8\gamma^2(\gamma^2+h^2-1)=0\), hence \(W_2:\gamma^2+h^2=1\). The principal \(2\times2\) minor gives \(R(h,\gamma)=3-\gamma^2-6h^2\); after squaring, one obtains \(\gamma^2(h^2-1)=-(2h^2-1)^2\), and the unsquared equation selects \(0<h\le1/\sqrt2\). Thus \(W_3:\gamma^2=(2h^2-1)^2/(1-h^2)\), with \(0<h\le1/\sqrt2\). Finally, the mixed minor gives \(R(h,\gamma)=-\gamma^2+4\gamma+2h^2-1\); after squaring, one finds \(\gamma^2+\gamma h^2-2\gamma-2h^2+1=0\), whose relevant positive branch is \(W_4:\gamma=1-\frac{h^2}{2}+\frac{h}{2}\sqrt{h^2+4}\). The complete positive-quadrant wall set is therefore
\begin{equation}
\label{eq:app-L4-wall-set}
\begin{aligned}
W_1&:\ \gamma^2+2h^2=1,\\
W_2&:\ \gamma^2+h^2=1,\\
W_3&:\ \gamma^2=\frac{(2h^2-1)^2}{1-h^2},\qquad 0<h\le\frac1{\sqrt2},\\
W_4&:\ \gamma=1-\frac{h^2}{2}+\frac{h}{2}\sqrt{h^2+4}.
\end{aligned}
\end{equation}
For \(L=4\), the selected patterned-minor wall is \(W_4\). Hence \(\gamma_*(h;4)=1-\frac{h^2}{2}+\frac{h}{2}\sqrt{h^2+4}\). Since Fig.~\ref{fig:app-L4-wall-branches} uses \(\gamma\) as the horizontal axis, the same branch may also be written as \(h_*(\gamma)=(\gamma-1)/\sqrt{2-\gamma}\), with \(1\le\gamma<2\).

The principal \(L=4\) product chamber is the connected region outside the circle wall \(W_2\) and before crossing \(W_4\), namely \(\mathcal R_4^{\rm PBC}=\{(h,\gamma)\in\mathbb R_{\ge0}^{2}:h^2+\gamma^2>1,\ 0\le\gamma<\gamma_*(h;4)\}\). The other two walls, \(W_1\) and \(W_3\), are genuine zero-minor loci, but they lie inside the circle wall and therefore do not bound the principal product chamber.

The line \(h=1\) gives a useful one-dimensional check. At \(h=1\), \(W_1\) has no positive solution, \(W_2\) reduces to \(\gamma=0\), and \(W_3\) is absent because it requires \(h\le1/\sqrt2\). The only nontrivial positive wall is \(W_4\), which gives \(\gamma_*(1;4)=1-\frac12+\frac12\sqrt5=(1+\sqrt5)/2\). Thus, for \(L=4\), the product branch extends beyond \(\gamma=1\) along \(h=1\) and terminates at the golden-ratio wall. This finite-size extension is the simplest explicit example of why the finite-\(L\) branch can be larger than the stable large-\(L\) region.

\paragraph*{Numerical chamber diagnostics.}
The \(L=4\) solution gives a concrete picture of the chamber structure. For larger \(L\), the same structure is visible numerically. A first diagnostic is obtained by plotting the zero loci of all distinct minors. Figure~\ref{fig:app-pbc-zero-loci} shows this for \(L=6,8\). These curves partition the \((\gamma,h)\)-plane into chambers. The product formula does not describe all chambers simultaneously; it selects the chamber connected to the stable large-\(L\) region. The remaining zero curves are genuine sign walls, but they correspond to other algebraic branches of the same absolute-minor polynomial.

\begin{figure}[!t]
\centering
\begin{minipage}{0.45\textwidth}
    \centering
    \includegraphics[width=\linewidth]{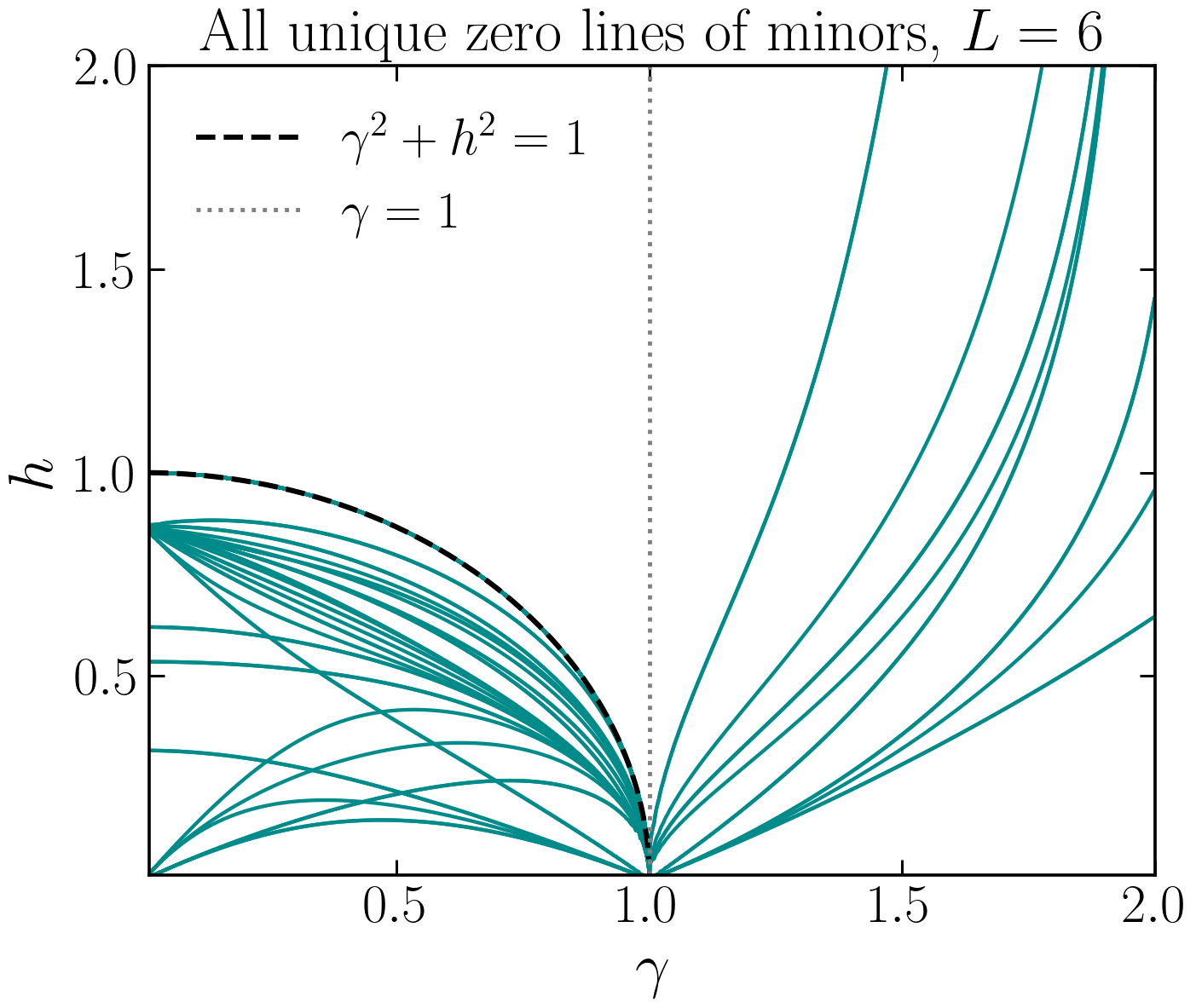}
    \textbf{(a)}
\end{minipage}
\hfill
\begin{minipage}{0.45\textwidth}
    \centering
    \includegraphics[width=\linewidth]{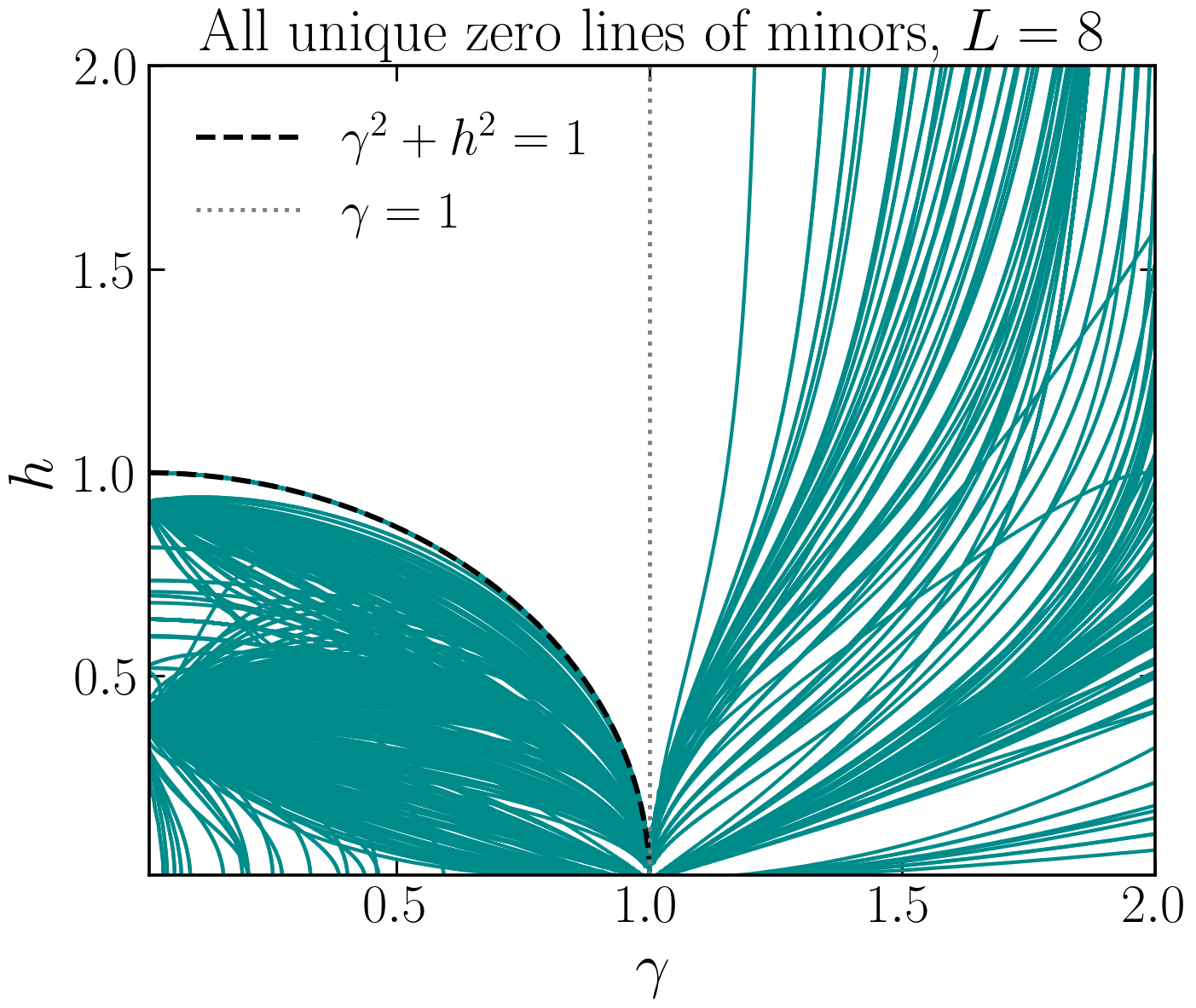}
    \textbf{(b)}
\end{minipage}
\caption{
Zero loci of distinct minors of \(G^{\rm PBC}\) in the \((\gamma,h)\)-plane for \(L=6,8\). The dashed curve is the circle \(h^2+\gamma^2=1\), and the dotted vertical line is \(\gamma=1\). The collection of curves shows that the absolute-minor polynomial is chamber dependent. The branch selected by the patterned minor \(D_M(h,\gamma)\) gives the observed exit wall of the product chamber used in the main text.
}
\label{fig:app-pbc-zero-loci}
\end{figure}

A second diagnostic directly compares the exact absolute-minor enumeration with the product formula. We plot \(\Delta_L(u;h,\gamma)=|P_{L,\rm exact}^{\rm PBC}(u;h,\gamma)-P_{L,\rm prod}^{\rm PBC}(u;h,\gamma)|\). As shown in Fig.~\ref{fig:app-pbc-product-validation}, the difference is at numerical precision inside the principal chamber and becomes finite after crossing the selected patterned-minor wall.

\paragraph*{Large-\(L\) stability region.}
The finite-size data show two related features. First, for each fixed \(L\), the product branch can extend beyond the asymptotic stable region \(h^2+\gamma^2>1\), \(0\le\gamma\le1\). The \(L=4\) golden-ratio wall along \(h=1\) is the simplest explicit example. Second, as \(L\) increases, the selected patterned-minor branch moves toward the line \(\gamma=1\). This supports the large-\(L\) stability region
\begin{equation}
\label{eq:app-pbc-large-L-stability-region}
\mathcal R_\infty^{\rm PBC}
=
\left\{
(h,\gamma):
h\ge0,\;
0\le\gamma\le1,\;
h^2+\gamma^2>1
\right\}.
\end{equation}
Within this region, the coherent sign branch remains stable in the scaling limit and gives the product formula used for the large-\(L\) analysis in the main text.

\section{Equivalent forms of the finite-periodic product}
\label{app:pbc-equivalent-forms}

This appendix collects equivalent formulations of the finite-periodic product formula. These forms are not needed to compute \(M_{\frac12}^{\rm PBC}\), but they clarify the algebraic structure behind the absolute-minor polynomial and its factorization on the principal product branch.

Let \(L=2M\), with positive NS momenta \(\alpha_j=(2j-1)\pi/L\), \(j=1,\dots,M\). We use the notation of Appendix~\ref{app:pbc-stability}: \(\Lambda_j(h,\gamma)=[h+\gamma+(1-\gamma)\cos\alpha_j]/\Delta_j(h,\gamma)\), with \(\Delta_j(h,\gamma)=\sqrt{(h+\cos\alpha_j)^2+\gamma^2\sin^2\alpha_j}\). The product formula is
\begin{equation}
\label{eq:app-PBC-product-basic}
P_L^{\rm PBC}(u;h,\gamma)
=
\prod_{j=1}^{M}
\left[
1+u^2+2u\Lambda_j(h,\gamma)
\right].
\end{equation}

\begin{figure}[!ht]
\centering
\begin{minipage}{0.45\textwidth}
    \centering
    \includegraphics[width=\linewidth]{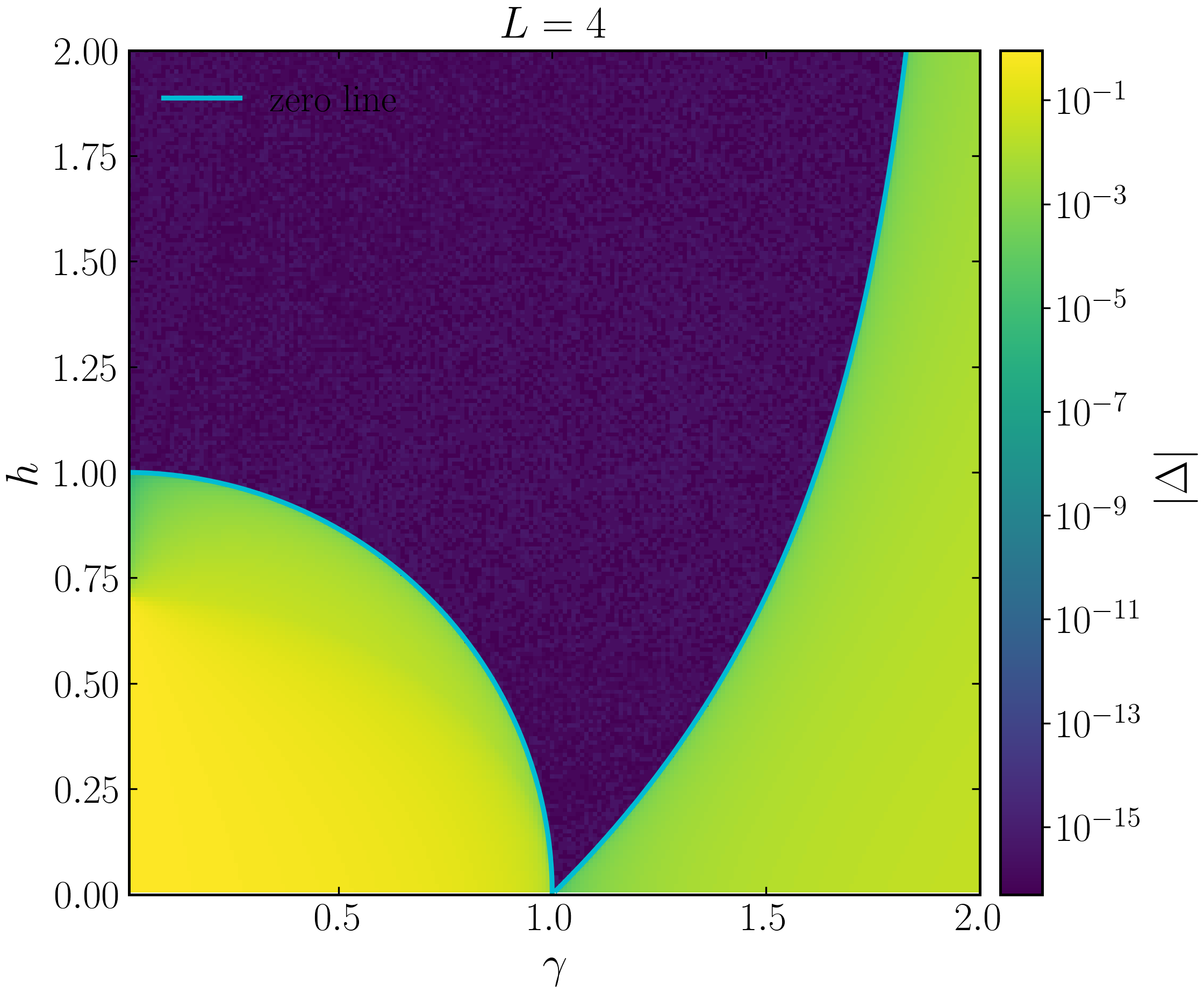}
    \textbf{(a)}
\end{minipage}
\hfill
\begin{minipage}{0.45\textwidth}
    \centering
    \includegraphics[width=\linewidth]{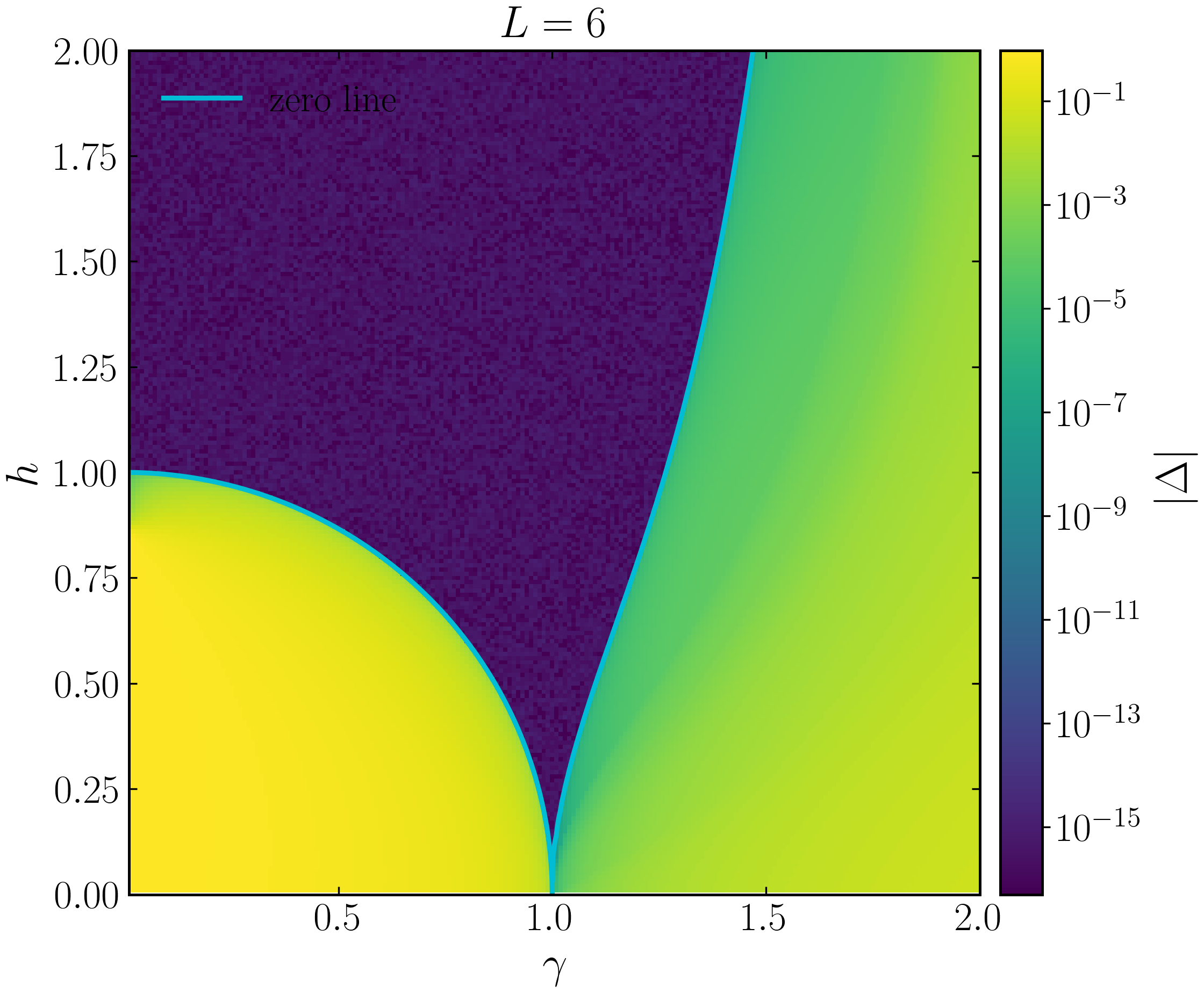}
    \textbf{(b)}
\end{minipage}
\caption{
Numerical validation of the finite-periodic product formula for \(L=4,6\). The color scale shows \(\Delta_L(u;h,\gamma)\), the absolute difference between the direct absolute-minor computation and the product formula. Dark regions indicate agreement up to numerical precision. The overlaid curve is the selected zero branch of the patterned minor \(D_M(h,\gamma)\), which gives the exit wall of the observed product chamber.
}
\label{fig:app-pbc-product-validation}
\end{figure}

\paragraph*{Rapidity form.}
On the product branch, introduce rapidities \(\varepsilon_j\ge0\) by \(\cosh\varepsilon_j=\Lambda_j(h,\gamma)\). Each quadratic factor then splits as \(1+u^2+2u\cosh\varepsilon_j=(1+ue^{\varepsilon_j})(1+ue^{-\varepsilon_j})\), and therefore
\begin{equation}
\label{eq:app-PBC-rapidity-product}
P_L^{\rm PBC}(u;h,\gamma)
=
\prod_{j=1}^{M}
(1+ue^{\varepsilon_j})(1+ue^{-\varepsilon_j}).
\end{equation}
Equivalently, if \(P_L^{\rm PBC}(u;h,\gamma)=\sum_{p=0}^{L}A_p^{\rm PBC}(h,\gamma)u^p\), then \(A_p^{\rm PBC}(h,\gamma)=e_p(e^{\varepsilon_1},e^{-\varepsilon_1},\dots,e^{\varepsilon_M},e^{-\varepsilon_M})\), where \(e_p\) is the elementary symmetric polynomial of degree \(p\).

\paragraph*{Finite-mode determinant form.}
Define the \(2\times2\) hyperbolic block \(\mathbf T_j(h,\gamma)=\begin{pmatrix}\cosh\varepsilon_j & \sinh\varepsilon_j\\ \sinh\varepsilon_j & \cosh\varepsilon_j\end{pmatrix}\), so that \(\det\mathbf T_j=1\). Since \(\det(\mathbf I_2+u\mathbf T_j)=1+u^2+2u\cosh\varepsilon_j\), the product can be written as \(P_L^{\rm PBC}(u;h,\gamma)=\prod_{j=1}^{M}\det(\mathbf I_2+u\mathbf T_j)\). Equivalently, with \(\mathbf T_{\rm blk}:=\bigoplus_{j=1}^{M}\mathbf T_j\),
\begin{equation}
\label{eq:app-Tblk-determinant}
P_L^{\rm PBC}(u;h,\gamma)
=
\det
\left[
\mathbf I_L+u\mathbf T_{\rm blk}(h,\gamma)
\right].
\end{equation}

\paragraph*{Rotation blocks and companion matrix.}
The signed one-particle matrix \(G^{\rm PBC}(h,\gamma)\) is real orthogonally equivalent to a direct sum of planar rotations. Namely, for some real orthogonal matrix \(Q\), one has \(Q^T G^{\rm PBC}Q=\bigoplus_{j=1}^{M}R(\varphi_j)\), where \(R(\varphi_j)=\begin{pmatrix}\cos\varphi_j&-\sin\varphi_j\\ \sin\varphi_j&\cos\varphi_j\end{pmatrix}\), with \(\cos\varphi_j=(h+\cos\alpha_j)/\Delta_j(h,\gamma)\) and \(\sin\varphi_j=\gamma\sin\alpha_j/\Delta_j(h,\gamma)\). Thus \(G^{\rm PBC}\) is controlled by rotation blocks, while the absolute-minor product is represented by the hyperbolic blocks \(\mathbf T_j\).

Using the same \(Q\), define the site-space companion matrix \(\widetilde{\mathbf T}(h,\gamma):=Q\mathbf T_{\rm blk}(h,\gamma)Q^T\). Then Eq.~\eqref{eq:app-Tblk-determinant} becomes
\begin{equation}
\label{eq:app-site-space-determinant}
P_L^{\rm PBC}(u;h,\gamma)
=
\det
\left[
\mathbf I_L+u\widetilde{\mathbf T}(h,\gamma)
\right].
\end{equation}
Comparing coefficients gives \(A_p^{\rm PBC}(h,\gamma)=\sum_{S\subseteq[L],\,|S|=p}\det\widetilde{\mathbf T}(h,\gamma)[S,S]\). Hence, on the product branch,
\begin{equation}
\label{eq:app-absolute-minor-principal-minor}
\sum_{\substack{S,T\subseteq[L]\\ |S|=|T|=p}}
\left|
\det G^{\rm PBC}[S,T]
\right|
=
\sum_{\substack{S\subseteq[L]\\ |S|=p}}
\det \widetilde{\mathbf T}[S,S].
\end{equation}
Thus the sum of absolute values of all \(p\times p\) minors of \(G^{\rm PBC}\) is represented as a sum of principal minors of the symmetric companion matrix \(\widetilde{\mathbf T}\).

\paragraph*{Exterior-power form.}
For ordered sets \(S,T\subseteq[L]\), with \(|S|=|T|=p\), one has \(\det G^{\rm PBC}[S,T]=\langle e_S,\wedge^pG^{\rm PBC}e_T\rangle\). Therefore \(A_p^{\rm PBC}(h,\gamma)=\|\wedge^pG^{\rm PBC}(h,\gamma)\|_{\ell^1\text{-entrywise}}\). On the other hand, \(\sum_{S\subseteq[L],\,|S|=p}\det\widetilde{\mathbf T}[S,S]=\operatorname{Tr}(\wedge^p\widetilde{\mathbf T})\). Hence Eq.~\eqref{eq:app-absolute-minor-principal-minor} can be written as
\begin{equation}
\label{eq:app-exterior-power-form}
\left\|
\wedge^pG^{\rm PBC}(h,\gamma)
\right\|_{\ell^1\text{-entrywise}}
=
\operatorname{Tr}
\left(
\wedge^p\widetilde{\mathbf T}(h,\gamma)
\right),
\qquad
p=0,\dots,L.
\end{equation}
This is the exterior-power form of the finite-periodic product formula.

\paragraph*{Critical-line specialization.}
On the anisotropic critical line \(h=1\), the product factors simplify by half-angle identities. Defining \(t_j=\tan(\alpha_j/2)=\tan[(2j-1)\pi/(2L)]\), one obtains
\begin{equation}
\label{eq:app-critical-line-product}
P_L^{\rm PBC}(u;1,\gamma)
=
\prod_{j=1}^{M}
\left[
1+u^2+
2u
\frac{1+\gamma t_j^2}
{\sqrt{1+\gamma^2 t_j^2}}
\right].
\end{equation}
At the critical TFI point, \((h,\gamma)=(1,1)\), this reduces to
\begin{equation}
\label{eq:app-critical-TFI-product}
P_L^{\rm PBC}(u;1,1)
=
\prod_{j=1}^{M}
\left[
1+u^2+2u\sec\frac{\alpha_j}{2}
\right].
\end{equation}

\section{Asymptotic regimes of the full periodic product and application to \(M_{\frac12}\)}\label{app:pbc-regimes-Mhalf}

In this appendix, we derive the large-\(L\) regimes of the full periodic product branch and explain how they translate into the stabilizer R\'enyi entropy at \(\alpha=1/2\). We work in the principal stability chamber defined in Sec.~\ref{sec:pbc-exact-results}, where the absolute-minor generating polynomial factorizes. For \(L=2M\), the product formula is
\begin{equation}
\label{eq:app-PBC-product-start}
P_L^{\mathrm{PBC}}(u;h,\gamma)=\prod_{j=1}^{M}\left[1+u^2+2u\Lambda_{h,\gamma}(\alpha_j)\right],
\end{equation}
where \(\alpha_j=\frac{(2j-1)\pi}{L}\), and \(\Lambda_{h,\gamma}(q)=\frac{h+\gamma+(1-\gamma)\cos q}{\sqrt{(h+\cos q)^2+\gamma^2\sin^2 q}}\). It is useful to introduce \(S_L(u;h,\gamma):=\ln P_L^{\mathrm{PBC}}(u;h,\gamma)\), \(F_u(q;h,\gamma):=\ln[1+u^2+2u\Lambda_{h,\gamma}(q)]\), and \(f_P(u;h,\gamma):=(2\pi)^{-1}\int_0^\pi F_u(q;h,\gamma)\,dq\). Then Eq.~\eqref{eq:app-PBC-product-start} becomes the midpoint sum \(S_L(u;h,\gamma)=\sum_{j=1}^{M}F_u(\alpha_j;h,\gamma)\). The only point where the summand may become singular in the stable region is \(q=\pi\) as \(h\to1\). The natural scaling variable is therefore
\begin{equation}\label{eq:app-scaling-variable}
\begin{split}
x=\frac{L|1-h|}{\gamma}\sim Lm(h,\gamma),
\\
m(h,\gamma)=\left|\ln\frac{h+\sqrt{h^2+\gamma^2-1}}{1+\gamma}\right|,    
\end{split}
\end{equation}
where \(m(h,\gamma)\) is the inverse correlation length and \(m(h,\gamma)\sim |1-h|/\gamma\) near \(h=1\).

\paragraph*{Regime I: massive region.}
In the massive regime \(x\to\infty\), in particular for fixed \(h\neq1\), the function \(F_u(q;h,\gamma)\) is smooth on the midpoint mesh. Hence, the midpoint sum is given by its Riemann integral up to exponentially small corrections. More precisely,
\begin{equation}\label{eq:app-regime-I-logP}
\begin{split}
S_L(u;h,\gamma)=L f_P(u;h,\gamma)&-\frac12\ln\!\left(1+e^{-m(h,\gamma)L}\right)+O(e^{-(m+\eta)L}) ,
\end{split}
\end{equation}
for some \(\eta>0\). Thus, the constant term is zero in the massive full-system regime. Expanding the logarithm gives the nearest-singularity tower
\begin{equation}
\label{eq:app-regime-I-logP-tower}
S_L=L f_P-\frac12 e^{-mL}+\frac14 e^{-2mL}-\frac16 e^{-3mL}
+\frac18 e^{-4mL}
+\cdots ,
\end{equation}
where \(m=m(h,\gamma)\). This expansion describes the tower generated by the closest complex singularities; farther singularities may enter through the \(O(e^{-(m+\eta)L})\) remainder.

\paragraph*{Regime II: critical line.}
At \(h=1\), the singularity sits exactly at the endpoint \(q=\pi\). Writing \(q=\pi-t\), one has \(\Lambda_{1,\gamma}(\pi-t)=2/t+O(t)\), and therefore \(F_u(\pi-t;1,\gamma)=\ln(4u/t)+O(t)\). We separate this singular part by writing \(F_u(\pi-t;1,\gamma)=\ln(4u/t)+H_u(t;\gamma)\), where \(H_u\) is regular at \(t=0\). The singular sum over \(t_k=(2k-1)\pi/L\) gives the universal constant \(-\frac12\ln2\), while the regular part is treated by the midpoint Euler--Maclaurin formula. If \(H_u(t;\gamma)=c_1t+c_2t^2+c_3t^3+c_4t^4+c_5t^5+\cdots\) near \(t=0\), then
\begin{equation}\label{eq:app-regime-II-logP}
\begin{split}
S_L(u;1,\gamma)=&L f_P(u;1,\gamma)-\frac12\ln2+\frac{A_1(u)}{L}+\frac{A_3(u,\gamma)}{L^3}+\frac{A_5(u,\gamma)}{L^5}+O(L^{-7}) ,
\end{split}
\end{equation}
with \(A_1=\pi c_1/12\), \(A_3=-7\pi^3c_3/120\), and \(A_5=31\pi^5c_5/252\). The first two local coefficients are \(c_1=(1+u^2)/(4u)\) and \(c_3=[\gamma^2(u^6+7u^4+7u^2+1)
-12\gamma(u^4+u^2)+6(u^4+u^2)]/192\gamma^2u^3\). The coefficient \(c_5\) is not needed for the general discussion below, but for the application to \(M_{\frac12}\) we will use its value at \(u=1\).

\paragraph*{Regime III: crossover.}
The crossover regime is obtained by taking \(L\to\infty\) and \(h\to1\) with \(x=L|1-h|/\gamma=O(1)\). Equivalently, set \(h_L=1+\sigma\gamma x/L\), with \(\sigma=\pm1\). Near \(q=\pi-t\), the denominator in \(\Lambda_{h,\gamma}\) has the local form \(\sqrt{(h-1)^2+\gamma^2t^2}\), so the discrete endpoint region depends on \(x\). Subtracting the corresponding local model and comparing the singular midpoint sum with its integral gives
\begin{equation}\label{eq:app-regime-III-logP}
\begin{split}
S_L(u;h_L,\gamma)=&L f_P(u;h_L,\gamma)-\frac12\ln(1+e^{-x})+\frac{A_1^{(\sigma)}(x;u,\gamma)}{L}+O(L^{-2}) .
\end{split}
\end{equation}
The leading crossover constant follows from the identity \(\cosh(x/2)=\prod_{k\ge1}\left[1+x^2/(\pi^2(2k-1)^2)\right]\). The first correction can be written as
\begin{equation}
\label{eq:app-crossover-A1}
A_1^{(\sigma)}(x;u,\gamma)
=
-\frac{1+u^2}{4u}\,
\frac{x}{\pi}
\sum_{n=1}^{\infty}
\frac{(-1)^n}{n}K_1(nx)
-\sigma\frac{x^2}{4\gamma(e^x+1)} ,
\end{equation}
where \(K_1\) is the modified Bessel function. This term has the correct matching limits: \(A_1^{(\sigma)}(x;u,\gamma)\to \pi(1+u^2)/(48u)\) as \(x\to0\), reproducing the critical \(1/L\) coefficient, while \(A_1^{(\sigma)}(x;u,\gamma)\to0\) as \(x\to\infty\), consistently with the massive expansion.

\paragraph*{Application to \(M_{\frac12}\).}
For the full periodic pure Gaussian state, \(\mathbf G^{\top}\mathbf G=\mathbf I\), hence \(\mathbf{Det}_2(\mathbf G)=2^L\). Moreover, \(P_L^{\mathrm{PBC}}(1;h,\gamma)=\mathbf{Det}_1(\mathbf G^{\mathrm{PBC}})\). Therefore
\begin{equation}
\label{eq:app-Mhalf-from-P}
M_{\frac12}^{\mathrm{PBC}}(L;h,\gamma)
=
2S_L(1;h,\gamma)-2L\ln2 .
\end{equation}
Thus the results above immediately give the asymptotic regimes of the stabilizer R\'enyi entropy. Define the extensive density \(M_{\frac12}(h,\gamma):=2f_P(1;h,\gamma)-2\ln2\). In the massive regime,
\begin{equation}
\label{eq:app-Mhalf-regime-I}
M_{\frac12}^{\mathrm{PBC}}
=
L M_{\frac12}(h,\gamma)
-\ln\!\left(1+e^{-m(h,\gamma)L}\right)
+O(e^{-(m+\eta)L}) ,
\end{equation}
or, equivalently, \(M_{\frac12}^{\mathrm{PBC}}=LM_{\frac12}-e^{-mL}+\frac12e^{-2mL}-\frac13e^{-3mL}+\cdots\). On the critical line \(h=1\), Eq.~\eqref{eq:app-regime-II-logP} gives
\begin{equation}
\label{eq:app-Mhalf-regime-II}
\begin{aligned}
M_{\frac12}^{\mathrm{PBC}}(L;1,\gamma)=L M_{\frac12}(1,\gamma)-\ln2+\frac{\pi}{12L}-\frac{7\pi^3}{2880\gamma^2}
(4\gamma^2-6\gamma+3)\frac{1}{L^3}+\frac{31\pi^5c_5}{252L^5}+O(L^{-7}) .
\end{aligned}
\end{equation}
Here we used the \(u=1\) coefficients \(c_1=1/2\), \(c_3=(4\gamma^2-6\gamma+3)/(48\gamma^2)\), and \(c_5=(16\gamma^4-48\gamma^3+48\gamma^2-12\gamma-3)/(768\gamma^4)\). Finally, in the crossover regime \(h_L=1+\sigma\gamma x/L\),
\begin{equation}\label{eq:app-Mhalf-regime-III}
\begin{split}
M_{\frac12}^{\mathrm{PBC}}(L;h_L,\gamma)=&L M_{\frac12}(h_L,\gamma)
-\ln(1+e^{-x})+\frac{B_1^{(\sigma)}(x;\gamma)}{L}+O(L^{-2}) ,
\end{split}
\end{equation}
where \(B_1^{(\sigma)}(x;\gamma)=-\frac{x}{\pi}\sum_{n=1}^{\infty}\frac{(-1)^n}{n}K_1(nx)-\sigma\frac{x^2}{2\gamma(e^x+1)}\). The crossover constant has the expected limits: \(-\ln(1+e^{-x})\to-\ln2\) for \(x\to0\), reproducing the critical constant, and \(-\ln(1+e^{-x})\to0\) for \(x\to\infty\), reproducing the massive regime. This confirms that the full periodic chain has no logarithmic finite-size correction at \(\alpha=1/2\); the only universal subleading term in the leading scaling is the chamber-dependent constant, with the crossover controlled by \(x=L|1-h|/\gamma\).

\section{Large-scale cusp of the periodic SRE density}
\label{app:pbc-density-cusp}

In this Appendix, we give the formal statement and proof of the cusp result used in Sec.~\ref{sec:pbc-exact-results}. The result concerns the large-scale periodic density obtained from the finite-size product formula,
\begin{equation}
\label{eq:app_pbc_density_cusp_density}
m_{\frac12}^{\mathrm{PBC}}(h,\gamma)
=
\frac{1}{\pi}
\int_0^\pi
\ln\left[
\frac{1+\Lambda_{h,\gamma}(q)}{2}
\right]dq,
\end{equation}
where \(0<\gamma\le1\) is fixed and
\begin{equation}
\label{eq:app_pbc_density_cusp_lambda}
\Lambda_{h,\gamma}(q)
=
\frac{h+\gamma+(1-\gamma)\cos q}
{\sqrt{(h+\cos q)^2+\gamma^2\sin^2 q}} .
\end{equation}
For finite \(L\), the corresponding density \(m_L(h,\gamma)=L^{-1}M_{\frac12}^{\mathrm{PBC}}(L;h,\gamma)\) is smooth inside a fixed sign chamber. The nonanalyticity below appears only after the scaling limit has been taken.

\begin{theorem}[Large-scale cusp of the periodic SRE density]
\label{thm:pbc_density_cusp}
Fix \(0<\gamma\le1\). The large-scale density \(m_{\frac12}^{\mathrm{PBC}}(h,\gamma)\) defined in Eq.~\eqref{eq:app_pbc_density_cusp_density} is continuous at \(h=1\), but it is not differentiable there. More precisely, the one-sided derivatives exist and satisfy \(\partial_h m_{\frac12}^{\mathrm{PBC}}(1^+,\gamma)-\partial_h m_{\frac12}^{\mathrm{PBC}}(1^-,\gamma)=-1/\gamma\). Equivalently, the leading large-scale density develops a cusp at the anisotropic critical line.
\end{theorem}

\begin{proof}
The only possible singular point is \(h=1\), \(q=\pi\). Away from \(q=\pi\), the denominator in Eq.~\eqref{eq:app_pbc_density_cusp_lambda} is nonzero in a neighborhood of \(h=1\), and the integrand is a smooth function of \(h\). Hence any jump in \(\partial_h m_{\frac12}^{\mathrm{PBC}}\) must come entirely from the endpoint \(q=\pi\).

We isolate the endpoint by writing \(h=1+\delta\) and \(q=\pi-t\), with \(t\ge0\) small. Then \(\cos(\pi-t)=-1+t^2/2+O(t^4)\) and \(\sin(\pi-t)=t+O(t^3)\). Therefore
\begin{equation}
\label{eq:app_pbc_density_cusp_denominator}
(h+\cos q)^2+\gamma^2\sin^2q
=
\delta^2+\gamma^2t^2+\delta t^2+O(t^4),
\end{equation}
whereas the numerator satisfies
\begin{equation}
\label{eq:app_pbc_density_cusp_numerator}
h+\gamma+(1-\gamma)\cos q
=
2\gamma+\delta+\frac{1-\gamma}{2}t^2+O(t^4).
\end{equation}
Since \(0<\gamma\le1\), the numerator is nonzero at \(\delta=0\), \(t=0\). Thus the numerator does not generate the nonanalyticity; the singular behavior comes from the square root in the denominator.

Let \(D_\delta(t):=(h+\cos q)^2+\gamma^2\sin^2q\) and \(N_\delta(t):=h+\gamma+(1-\gamma)\cos q\). Near the endpoint,
\begin{equation}
\label{eq:app_pbc_density_cusp_log_split}
\ln\left[
\frac{1+\Lambda_{h,\gamma}(q)}{2}
\right]
=
\ln\left[N_\delta(t)+\sqrt{D_\delta(t)}\right]
-\ln2
-\frac12\ln D_\delta(t).
\end{equation}
The first term in Eq.~\eqref{eq:app_pbc_density_cusp_log_split} has identical one-sided \(h\)-derivatives at \(\delta=0\). Indeed, \(N_\delta(t)\) stays bounded away from zero near \(t=0\), and replacing \(\sqrt{D_\delta(t)}\) by \(\sqrt{\delta^2+\gamma^2t^2}\) changes the local integral only by terms differentiable in \(\delta\) at the origin. Thus the jump in the derivative is entirely controlled by the last term in Eq.~\eqref{eq:app_pbc_density_cusp_log_split}.

The singular model is therefore
\begin{equation}
\label{eq:app_pbc_density_cusp_singular_model}
m_{\mathrm{sing}}(\delta,\gamma)
=
-\frac{1}{2\pi}
\int_0^\epsilon
\ln\left(\delta^2+\gamma^2t^2\right)dt,
\end{equation}
with \(\epsilon>0\) fixed and small. Replacing the exact \(D_\delta(t)\) by \(\delta^2+\gamma^2t^2\) affects only terms whose left and right derivatives at \(\delta=0\) coincide, so Eq.~\eqref{eq:app_pbc_density_cusp_singular_model} contains the full nonanalytic contribution.

For \(\delta\neq0\), differentiating Eq.~\eqref{eq:app_pbc_density_cusp_singular_model} gives
\begin{equation}
\label{eq:app_pbc_density_cusp_derivative}
\partial_\delta m_{\mathrm{sing}}(\delta,\gamma)
=
-\frac{1}{\pi}
\int_0^\epsilon
\frac{\delta}{\delta^2+\gamma^2t^2}\,dt
=
-\frac{\operatorname{sgn}(\delta)}{\pi\gamma}
\arctan\left(\frac{\gamma\epsilon}{|\delta|}\right).
\end{equation}
Taking the two one-sided limits, one obtains
\begin{equation}
\label{eq:app_pbc_density_cusp_one_sided}
\lim_{\delta\to0^+}\partial_\delta m_{\mathrm{sing}}(\delta,\gamma)
=
-\frac{1}{2\gamma},
\qquad
\lim_{\delta\to0^-}\partial_\delta m_{\mathrm{sing}}(\delta,\gamma)
=
+\frac{1}{2\gamma}.
\end{equation}
All remaining terms in the full density contribute the same finite value to both one-sided derivatives. Therefore
\begin{equation}
\label{eq:app_pbc_density_cusp_jump}
\partial_h m_{\frac12}^{\mathrm{PBC}}(1^+,\gamma)
-
\partial_h m_{\frac12}^{\mathrm{PBC}}(1^-,\gamma)
=
-\frac{1}{\gamma}.
\end{equation}
Since this jump is nonzero for every fixed \(0<\gamma\le1\), the large-scale density is not differentiable at \(h=1\).

Finally, the same local calculation shows the form of the nonanalytic contribution itself. Integrating the one-sided singular derivative gives
\begin{equation}
\label{eq:app_pbc_density_cusp_local_nonanalytic}
m_{\frac12}^{\mathrm{PBC}}(h,\gamma)
=
m_{\mathrm{reg}}(h,\gamma)
-\frac{|h-1|}{2\gamma}
+o(|h-1|),
\end{equation}
where \(m_{\mathrm{reg}}(h,\gamma)\) is differentiable at \(h=1\). Thus, the entropy density is continuous, but its slope jumps across the anisotropic critical line.
\end{proof}

\section{Subsystem Toeplitz truncations, Pfaffian formulas, and shifted interval chambers}
\label{app:subsystem_toeplitz_pfaffian}

In this appendix, we collect the Pfaffian generating-function formulas used for infinite-chain subsystems. The construction is algebraic: for a finite Toeplitz-type block, the generating polynomial of the absolute values of all square minors can be represented as a Pfaffian on a selected sign-stable branch. This is the object entering the numerator of the stabilizer Rényi entropy at \(\alpha=1/2\). The subsystem problem is not obtained by replacing the periodic length \(L\) by the interval length \(\ell\). In the periodic setup one first fixes a finite ring and a fermionic sector. Here one first takes the infinite-volume ground state and then restricts its Toeplitz kernel to an interval. Thus there is no NS/R sector and no discrete momentum mesh. What remains is a finite Toeplitz truncation together with the sign-chamber structure of the absolute-minor sum.

The formulas below should be viewed as finite-\(\ell\) identities on a specified analytic branch. The branch qualification is essential because the absolute values make the minor polynomial only piecewise analytic in parameter space. The Pfaffian expressions provide signed generating functions globally, but they coincide with the absolute-minor sums only inside the branch where the signs of the relevant minors agree with the chosen reference chamber. The construction follows standard free-fermion and Pfaffian-minor techniques; see Refs.~\cite{BravyiKitaev2005,PeschelEisler2009,Franchini2017}.

\paragraph*{Common notation.}
We use zero-based interval labels \(a,b=0,\ldots,\ell-1\). For \(h\ge0\) and \(\gamma\ge0\), the infinite-volume XY phase and its Fourier coefficients are
\begin{equation}
\label{eq:app_omega_symbol}
\omega_{\gamma,h}(e^{iq})
=
\frac{h+\cos q+i\gamma\sin q}
{\sqrt{(h+\cos q)^2+\gamma^2\sin^2 q}},
\qquad
g_r(h,\gamma)
=
\frac{1}{2\pi}
\int_0^{2\pi}
\omega_{\gamma,h}(e^{iq})e^{-irq}\,dq .
\end{equation}
Thus \(\omega_{\gamma,h}(e^{iq})=\sum_{r\in\mathbb Z}g_r(h,\gamma)e^{irq}\). We often suppress the parameter dependence and write simply \(g_r\). The three square interval blocks used in the main text are
\begin{equation}
\label{eq:app_three_square_blocks}
\begin{aligned}
\bigl(G_\ell^{{\rm XY},(\infty)}\bigr)_{ab}
&=g_{a-b},\\
\bigl(G_\ell^{{\rm XYR},(\infty)}\bigr)_{ab}
&=g_{1-a+b},\\
\bigl(G_\ell^{{\rm XYL},(\infty)}\bigr)_{ab}
&=g_{a-b-1}.
\end{aligned}
\end{equation}
Equivalently, \(G_\ell^{X,(\infty)}=[\kappa^X_{a-b}]_{a,b=0}^{\ell-1}\), with \(\kappa_r^{\rm XY}=g_r\), \(\kappa_r^{\rm XYR}=g_{1-r}\), and \(\kappa_r^{\rm XYL}=g_{r-1}\). The corresponding scalar representatives are \(\omega_{\rm XY}(z)=\omega_{\gamma,h}(z)\), \(\omega_{\rm XYR}(z)=z\,\omega_{\gamma,h}(z^{-1})\), and \(\omega_{\rm XYL}(z)=z\,\omega_{\gamma,h}(z)\). Thus \(XYR\) is the Kramers--Wannier, or cluster-dual, shifted interval representative, while \(XYL\) is the direct cluster-oriented shifted interval representative. These labels refer to interval blocks, not to new bulk universality classes.

For each \(X\in\{{\rm XY},{\rm XYR},{\rm XYL}\}\), define
\begin{equation}
\label{eq:app_P_X_def}
P_\ell^{X,(\infty)}(u;h,\gamma)
=
\sum_{p=0}^{\ell}u^p
\sum_{\substack{I,J\subseteq\{0,\ldots,\ell-1\}\\ |I|=|J|=p}}
\left|
\det G_\ell^{X,(\infty)}(h,\gamma)[I,J]
\right| .
\end{equation}
At \(u=1\), this gives \(\operatorname{Det}_1(G_\ell^{X,(\infty)})\). The formulas below give Pfaffian representations of \(P_\ell^{X,(\infty)}\) on the appropriate sign-stable branch. Outside that branch, the same Pfaffians should be read as signed generating functions, not necessarily as sums of absolute values.

\subsection{Ordinary \(XY\) interval}

\paragraph*{Toeplitz block and minor polynomial.}
The ordinary infinite-chain subsystem is the principal Toeplitz truncation \(\mathbf{G}_\ell^{{\rm XY},(\infty)}=[g_{a-b}]_{a,b=0}^{\ell-1}\). Its absolute-minor generating function is \(P_\ell^{{\rm XY},(\infty)}(u;h,\gamma)\) as defined in Eq.~\eqref{eq:app_P_X_def}. This is the subsystem analogue of the determinant sum entering the stabilizer entropy at \(\alpha=1/2\). The boundary effects in this geometry come from the finite Toeplitz truncation itself, not from closing the system into a finite ring.

The phase \(\omega_{\gamma,h}\) is the phase of the Laurent polynomial \(f_{\gamma,h}(z)=h+\frac{1+\gamma}{2}z+\frac{1-\gamma}{2}z^{-1}\). After multiplication by \(z\), the associated quadratic equation is \((1+\gamma)z^2+2hz+(1-\gamma)=0\). Its discriminant is proportional to \(h^2+\gamma^2-1\). Hence the circle \(h^2+\gamma^2=1\) separates the real-root and complex-root regimes of the underlying XY symbol. In the finite-\(\ell\) minor problem, this circle serves as a natural lower boundary for the sign-stable massive chamber.

\paragraph*{Doubled skew matrix and block-Toeplitz symbol.}
Introduce two copies of the interval labels, \(r_0,\ldots,r_{\ell-1}\) and \(c_0,\ldots,c_{\ell-1}\), and define
\begin{equation}
\label{eq:app_Q_XY_def}
\mathcal Q_\ell^{\rm XY}
=
\begin{pmatrix}
\mathbf 0 & \mathbf{G}_\ell^{{\rm XY},(\infty)}\\
-\bigl(\mathbf{G}_\ell^{{\rm XY},(\infty)}\bigr)^T & \mathbf 0
\end{pmatrix}.
\end{equation}
For subsets \(I,J\subseteq\{0,\ldots,\ell-1\}\) with \(|I|=|J|\), the determinant \(\det \mathbf{G}_\ell^{{\rm XY},(\infty)}[I,J]\) is, up to the standard ordering sign, the Pfaffian of the principal submatrix of \(\mathcal Q_\ell^{\rm XY}\) supported on \(r_I\cup c_J\). This identity is the basic mechanism that turns determinant sums into a Pfaffian problem.

Now reorder the labels into the interlaced order \(r_0,c_0,r_1,c_1,\ldots,r_{\ell-1},c_{\ell-1}\). Let \(\Pi_\ell^{\rm int}\) be the corresponding permutation matrix and let \(\mathbf D_\ell^{\rm int}=\operatorname{diag}(\epsilon_0,\epsilon_0,\epsilon_1,\epsilon_1,\ldots,\epsilon_{\ell-1},\epsilon_{\ell-1})\), with \(\epsilon_a=(-1)^a\). Then
\begin{equation}
\label{eq:app_interlaced_block_toeplitz_XY}
\mathbf D_\ell^{\rm int} \mathbf{\Pi}_\ell^{\rm int}
\mathcal Q_\ell^{\rm XY}
(\mathbf{\Pi}_\ell^{\rm int})^\top
\mathbf D_\ell^{\rm int}=
\bigl(\mathbf A_d^{\rm XY}\bigr)_{0\le a,b\le\ell-1},
\qquad
d=a-b,
\end{equation}
where
\begin{equation}
\label{eq:app_block_entries_XY}
\mathbf A_d^{\rm XY}
=
\begin{pmatrix}
0 & (-1)^d g_d\\
-(-1)^d g_{-d} & 0
\end{pmatrix}.
\end{equation}
With the convention \(\mathbf{\Phi}_A^{\rm XY}(e^{iq})=\sum_d \mathbf A_d^{\rm XY}e^{-idq}\), the corresponding block-Toeplitz symbol is
\begin{equation}
\label{eq:app_block_symbol_XY}
\mathbf{\Phi}_A^{\rm XY}(e^{iq};h,\gamma)
=
\begin{pmatrix}
0 & \omega_{\gamma,h}(-e^{-iq})\\
-\omega_{\gamma,h}(-e^{iq}) & 0
\end{pmatrix}.
\end{equation}
This is the exact undeformed block symbol of the doubled subsystem matrix. It records both the Toeplitz structure and the row-column orientation used in the Pfaffian representation.

\paragraph*{Pfaffian generating function and block symbol.}
The \(u\)-deformation that packages the absolute-minor generating function is built from the universal skew matrix \(\mathbf J_\ell(u)\) and the cross block \( \mathbf T_\ell^{\rm XY}(u)\). Their entries are
\[
(\mathbf J_\ell(u))_{ab}=u^{-1}\operatorname{sgn}(a-b),
\qquad
(\mathbf T_\ell^{\rm XY}(u))_{ab}=(-1)^{a-b}g_{a-b}+u^{-1}\varepsilon_{ab},
\]
where \(\operatorname{sgn}(0)=0\), \(\varepsilon_{ab}=1\) for \(a\le b\), and \(\varepsilon_{ab}=-1\) for \(a>b\). 
In the natural doubled order \((r_0,\ldots,r_{\ell-1},c_0,\ldots,c_{\ell-1})\), define
\begin{equation}
\label{eq:app_PhiK_symbol_XY}
\mathbf{K}_\ell^{\rm XY}(u;h,\gamma)
=\begin{pmatrix}
\mathbf J_\ell(u) & \mathbf T_\ell^{\rm XY}(u;h,\gamma)\\
-\bigl(\mathbf T_\ell^{\rm XY}(u;h,\gamma)\bigr)^\top & \mathbf J_\ell(u)
\end{pmatrix}.
\end{equation}
If \(\Pi_\ell\) denotes the permutation to the interlaced order
\((r_0,c_0,r_1,c_1,\ldots,r_{\ell-1},c_{\ell-1})\), we set
\begin{equation}
\label{eq:app_K_interlaced_XY}
\overline{\mathbf{K}}_\ell^{\rm XY}(u;h,\gamma)
=\mathbf \Pi_\ell \mathbf{K}_\ell^{\rm XY}(u;h,\gamma)\mathbf\Pi_\ell^T .
\end{equation}
On the sign-stable branch connected to the large-\(h\) region,
\begin{equation}
\label{eq:app_Pf_interlaced_XY}
P_\ell^{{\rm XY},(\infty)}(u;h,\gamma)
=
\left|u^\ell\operatorname{Pf}\mathbf{K}_\ell^{\rm XY}(u;h,\gamma)\right|
=
\left|u^\ell\operatorname{Pf}\overline{ \mathbf{K}}_\ell^{\rm XY}(u;h,\gamma)\right|.
\end{equation}
Equivalently, after fixing a Pfaffian convention inside the chamber, the absolute value may be replaced by a constant sign factor independent of \(u,h,\gamma\).
The prefactor \(u^\ell\) compensates the powers of \(u^{-1}\) carried by the auxiliary kernel, so that a term with \(p\) selected Toeplitz cross-pairings contributes with weight \(u^p\). The selected cross-pairings reproduce determinants of submatrices of \(G_\ell^{{\rm XY},(\infty)}\), while the remaining labels are paired by the universal auxiliary contractions.

\begin{figure}[!ht]
\centering
\begin{minipage}{0.45\textwidth}
    \centering
    \includegraphics[width=\linewidth]{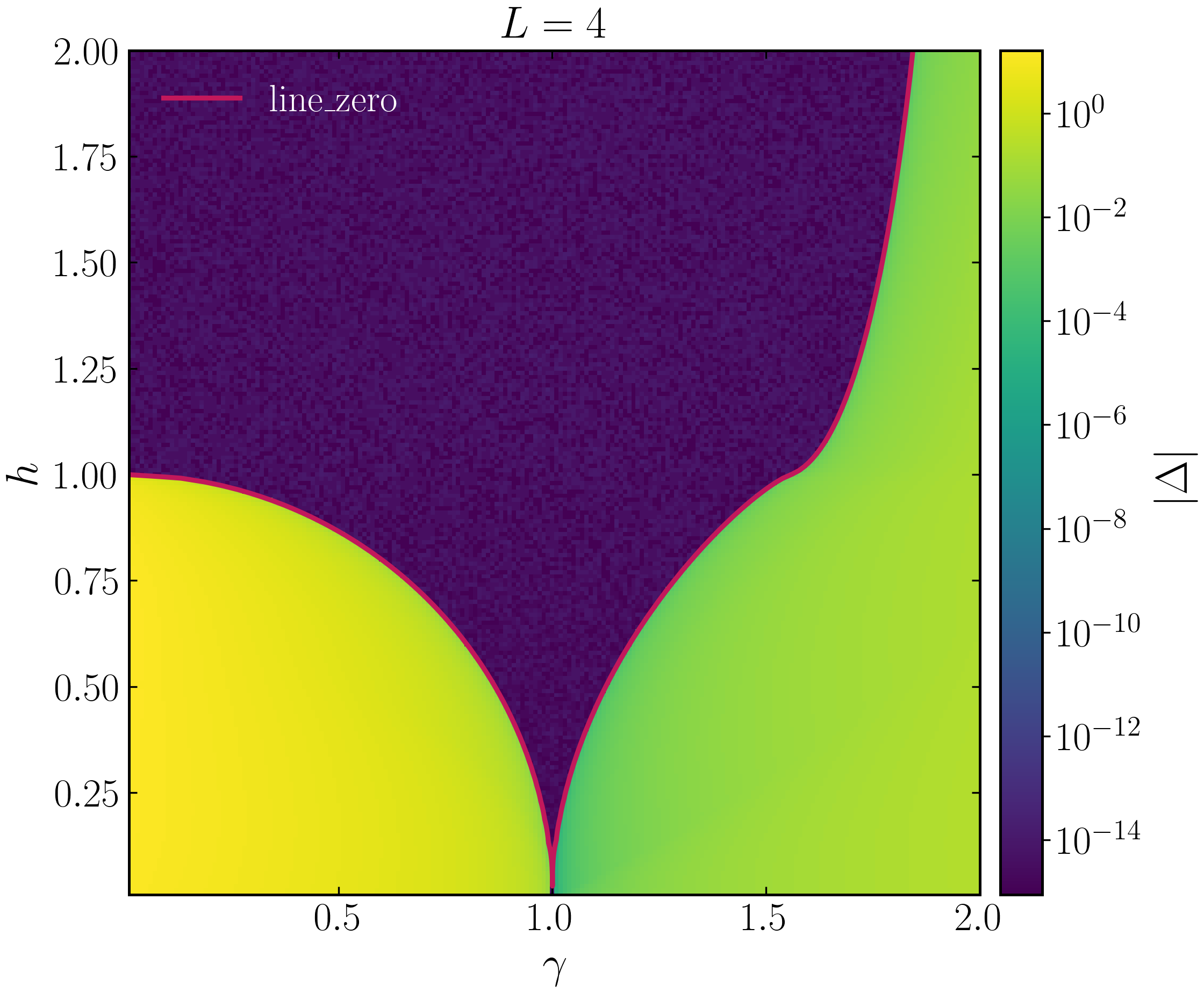}
    \textbf{(a)}
\end{minipage}
\hfill
\begin{minipage}{0.45\textwidth}
    \centering
    \includegraphics[width=\linewidth]{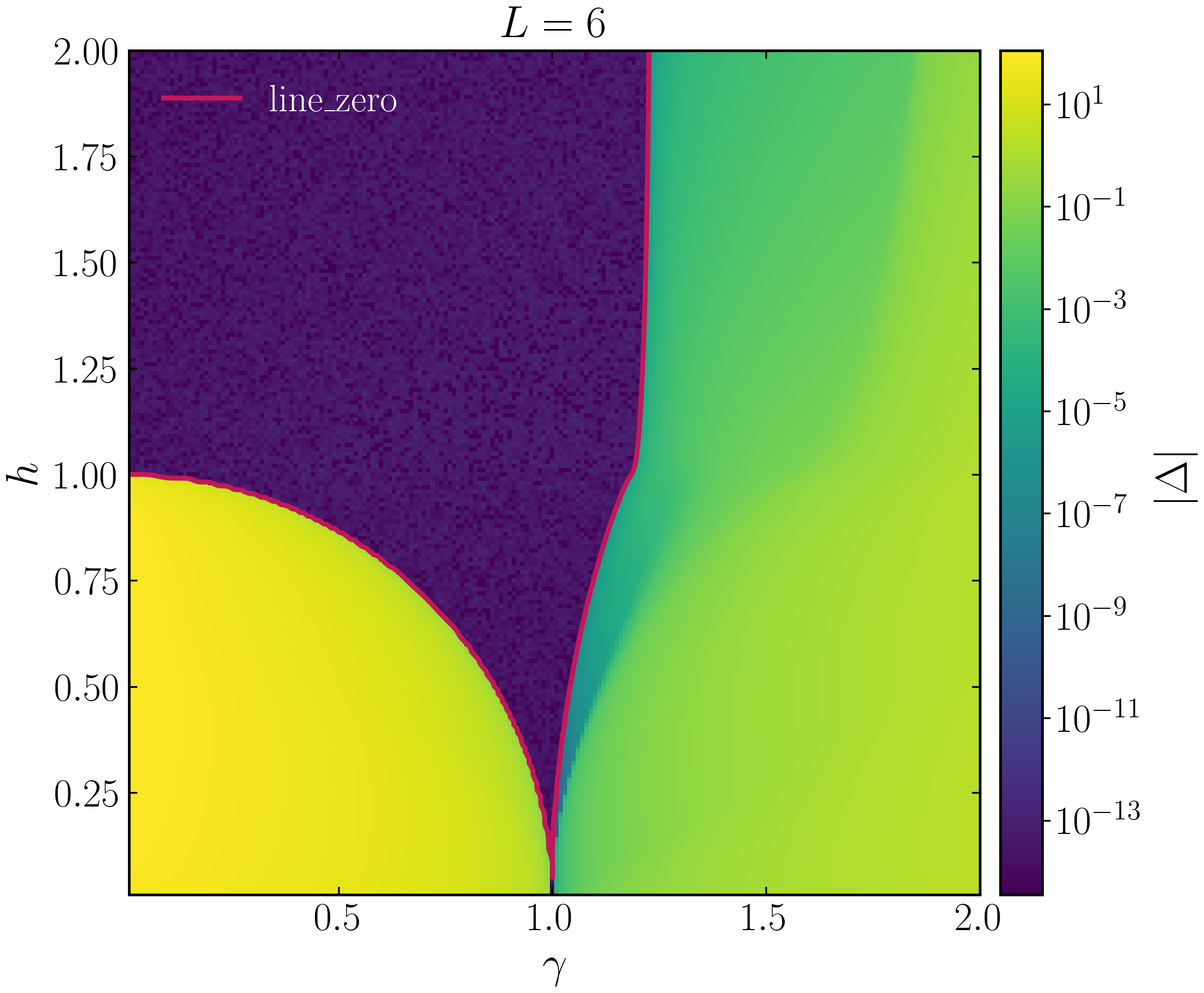}
    \textbf{(b)}
\end{minipage}
\caption{
Numerical validation of the subsystem Pfaffian representation for the ordinary \(XY\) interval with \(\ell=4,6\). The color scale shows \(\Delta_\ell(u;h,\gamma)=\left|P_\ell^{{\rm XY},(\infty)}(u;h,\gamma)-u^\ell\operatorname{Pf}K_\ell^{\rm XY}(u;h,\gamma)\right|\). Dark regions indicate agreement up to numerical precision. The overlaid curve marks the circle wall \(h^2+\gamma^2=1\), which separates the non-oscillatory and oscillatory root regimes of the XY symbol.
}
\label{fig:app-subsystem-pfaffian-validation}
\end{figure}

In the interlaced order, \(\overline {\mathbf K}_\ell^{\rm XY}\) is an \(\ell\times\ell\) block Toeplitz matrix with \(2\times2\) coefficients depending only on \(r=a-b\):
\begin{equation}
\label{eq:app_Kr_XY}
\mathbf K_r^{\rm XY}(u)=\begin{pmatrix}
u^{-1}\operatorname{sgn}(r)
&
(-1)^r g_r + u^{-1}\varepsilon_r
\\
-(-1)^{-r}g_{-r}-u^{-1}\varepsilon_{-r}
&
u^{-1}\operatorname{sgn}(r)
\end{pmatrix},
\qquad
\varepsilon_r=
\begin{cases}
1, & r\le 0,\\
-1, & r>0.
\end{cases}
\end{equation}
Its block symbol
\begin{equation}
\label{eq:app_A_symbol_XY}
\mathbf{\Phi}_{K,u}^{\rm XY}(e^{iq};h,\gamma)=
\sum_{r\in\mathbb Z}\mathbf{K}_r^{\rm XY}(u)e^{-irq}
\end{equation}
is
\begin{equation}
\label{eq:app_A_symbol_XY_explicit}
\mathbf{\Phi}_{K,u}^{\rm XY}(e^{iq};h,\gamma)=
\begin{pmatrix}
-\dfrac{i}{u}\cot\dfrac{q}{2}&\omega_{\gamma,h}(-e^{-iq})
+\dfrac{1}{u}\left(1+i\cot\dfrac{q}{2}\right)\\[1.1em]
-\omega_{\gamma,h}(-e^{iq})-\dfrac{1}{u}\left(1-i\cot\dfrac{q}{2}\right)&-\dfrac{i}{u}\cot\dfrac{q}{2}
\end{pmatrix},
\end{equation}
where \(\cot(q/2)\) is understood in the principal-value sense. This auxiliary block symbol should not be confused with the physical scalar phase \(\omega_{\gamma,h}\): the scalar phase defines the Toeplitz kernel, while \(A_{K,u}^{\rm XY}\) encodes the determinant sum through the Pfaffian deformation. Obtaining the block symbols for cases with the shift key is parallel to what has been done here.

The corresponding scalar determinant symbol is not an independent assumption but the determinant reduction of the block symbol:
\begin{equation}
\label{eq:app_scalar_symbol_def_XY}
a_{u;h,\gamma}(e^{iq})
:=
u^2\det \mathbf{\Phi}_{K,u}^{\rm XY}(e^{iq};h,\gamma).
\end{equation}
Writing
\[
W_-(q)=\omega_{\gamma,h}(-e^{-iq}),
\qquad
W_+(q)=\omega_{\gamma,h}(-e^{iq})
\]
one obtains
\begin{equation}
\label{eq:app_scalar_symbol_general_XY}
a_{u;h,\gamma}(e^{iq})
=
1+u^2 W_-(q)W_+(q)
+
u\left[
W_-(q)\left(1-i\cot\frac q2\right)
+
W_+(q)\left(1+i\cot\frac q2\right)
\right].
\end{equation}
For the XY phase, \(W_-(q)W_+(q)=1\). Moreover, after the harmless momentum relabelling \(q\mapsto q+\pi\), this reduces to
\begin{equation}
\label{eq:app_scalar_symbol_XY_final}
a_{u;h,\gamma}(e^{iq})=1+u^2+2u\Lambda_{h,\gamma}(q),
\end{equation}
with
\begin{equation}
\label{eq:app_Lambda_XY}
\Lambda_{h,\gamma}(q)=\frac{h+\gamma+(1-\gamma)\cos q}
{\sqrt{(h+\cos q)^2+\gamma^2\sin^2 q}}.
\end{equation}
Thus \(\mathbf{\Phi}_{K,u}^{\rm XY}\) is the genuine block symbol of the Pfaffian problem, whereas \(a_{u;h,\gamma}\) controls only the determinant-level bulk and Fisher--Hartwig singular structure. The finite crossover remainder still depends on the local block structure of \(\mathbf{\Phi}_{K,u}^{\rm XY}\).

\begin{figure}[!ht]
\centering
\begin{minipage}{0.45\textwidth}
    \centering
    \includegraphics[width=\linewidth]{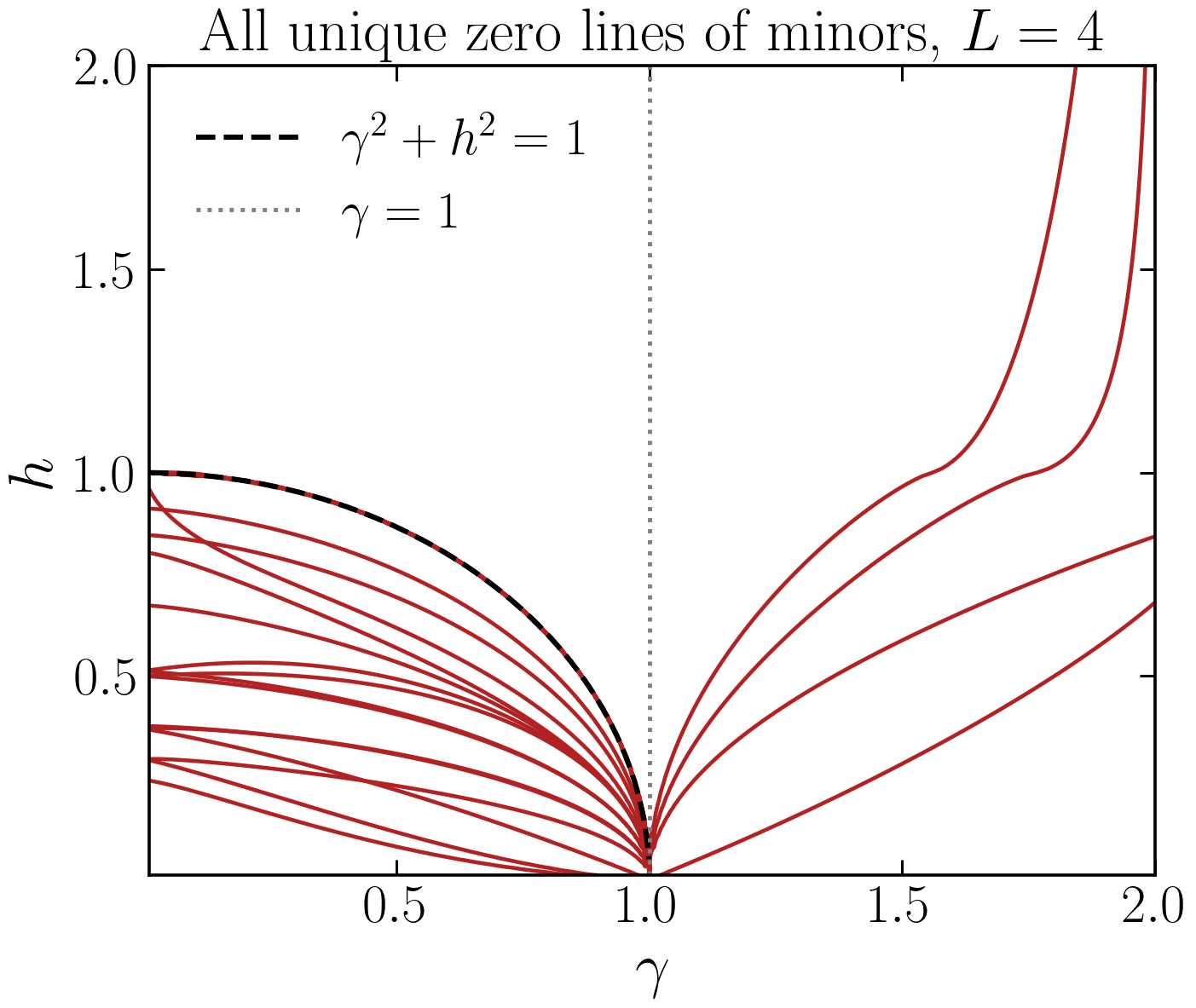}
    \textbf{(a)}
\end{minipage}
\hfill
\begin{minipage}{0.45\textwidth}
    \centering
    \includegraphics[width=\linewidth]{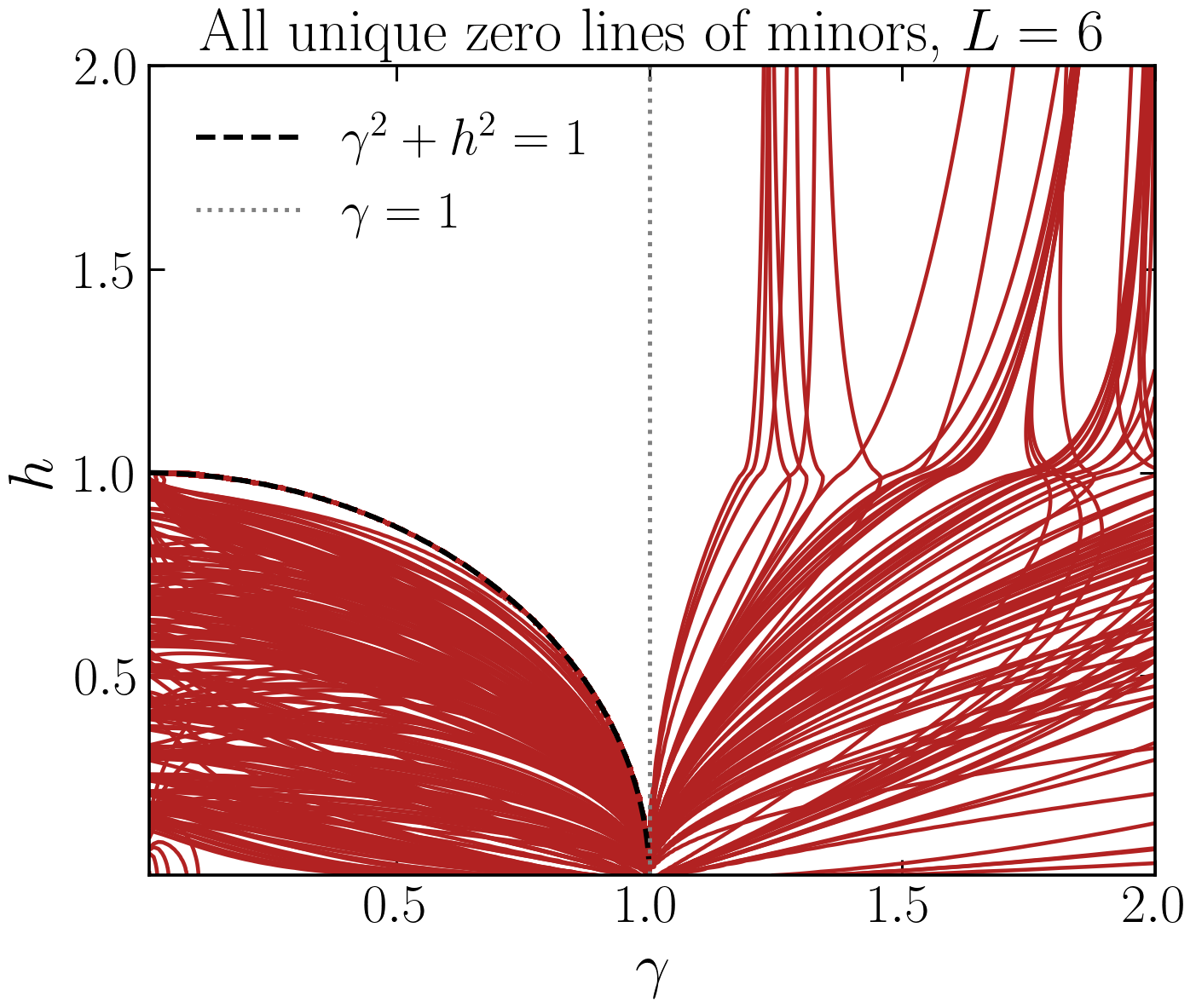}
    \textbf{(b)}
\end{minipage}
\caption{
Minor-zero loci for ordinary \(XY\) subsystem Toeplitz truncations with \(\ell=4,6\). The solid curves are zero lines of distinct minors of \(G_\ell^{{\rm XY},(\infty)}(h,\gamma)\). The dashed curve is the circle wall \(h^2+\gamma^2=1\), and the dotted vertical line marks \(\gamma=1\). The relevant stability chamber is the connected region outside the circle and on the small-\(\gamma\) side of the distinguished zero branch of the patterned minor \(\mathcal D_m^{\rm XY}(h,\gamma)\).
}
\label{fig:app-subsystem-minor-zero-loci}
\end{figure}

The numerical agreement between the direct absolute-minor summation and the Pfaffian expression in Eq.~\eqref{eq:app_Pf_interlaced_XY} is shown in Fig.~\ref{fig:app-subsystem-pfaffian-validation} for \(\ell=4,6\). The figure confirms that the interlaced Pfaffian representation reproduces the absolute-minor polynomial on the selected branch, with loss of agreement occurring when the parameters cross into a different sign chamber.

\paragraph*{Minor-zero chamber for \(XY\).}
The stability condition for Eq.~\eqref{eq:app_Pf_interlaced_XY} is not a sector-selection condition. It is a statement about the analytic branch of the absolute-minor sum. Away from singularities of the Toeplitz symbol, every minor \(\det G_\ell^{{\rm XY},(\infty)}[I,J]\) is analytic in \((h,\gamma)\). Its sign can change only when that minor vanishes. Hence \(P_\ell^{{\rm XY},(\infty)}\) is piecewise analytic, and each connected component cut by minor-zero loci defines one sign branch.

For even \(\ell=2m\), the finite-size sign boundary is detected by the patterned minor
\begin{equation}
\label{eq:app_patterned_minor_XY}
\mathcal D_m^{\rm XY}(h,\gamma)
=
\det G_{2m}^{{\rm XY},(\infty)}(h,\gamma)[I_m,J_m],
\end{equation}
where \(I_m=\{0,1,\ldots,m-1\}\) and \(J_m=\{m-1,m,\ldots,2m-2\}\). Let \(\gamma_*^{\rm XY}(h;\ell)\) denote the selected zero branch of this minor, reached from the small-\(\gamma\) side outside the circle. The ordinary subsystem chamber is
\begin{equation}
\label{eq:app_RXY_stability}
\mathcal R_\ell^{\rm XY}
=
\left\{
(h,\gamma)\in\mathbb R_{\ge0}^2:
h^2+\gamma^2>1,\quad
0\le\gamma<\gamma_*^{\rm XY}(h;\ell)
\right\}.
\end{equation}
For odd \(\ell\), the same branch is defined by direct minor tracking, or equivalently by continuity from neighboring even truncations. The minor-zero loci that determine this branch structure are illustrated in Fig.~\ref{fig:app-subsystem-minor-zero-loci}.

\subsection{Kramers--Wannier shifted interval \(XYR\)}

\paragraph*{Shifted symbol, Toeplitz block, and minor polynomial.}
The Kramers--Wannier, or cluster-dual, shifted representative is \(F_{\rm XYR}(z)=z f_{\gamma,h}(z^{-1})\), with phase \(\omega_{\rm XYR}(z)=z\,\omega_{\gamma,h}(z^{-1})\). Its Fourier coefficients are \(\kappa_r^{\rm XYR}=g_{1-r}\), so
\begin{equation}
\label{eq:app_G_XYR_def}
G_\ell^{{\rm XYR},(\infty)}(h,\gamma)
=
\bigl[g_{1-a+b}(h,\gamma)\bigr]_{a,b=0}^{\ell-1}.
\end{equation}
The corresponding absolute-minor polynomial is \(P_\ell^{{\rm XYR},(\infty)}(u;h,\gamma)\), as defined in Eq.~\eqref{eq:app_P_X_def}. This is the square interval block naturally produced by the Kramers--Wannier or cluster-dual shifted representative.

The \(XYR\) block is treated by the same interlaced Pfaffian construction as the ordinary \(XY\) block. Algebraically, the passage from \(XY\) to \(XYR\) amounts to the replacement \(g_{a-b}\mapsto g_{1-a+b}\), or equivalently \(\omega_{\gamma,h}(z)\mapsto z\,\omega_{\gamma,h}(z^{-1})\).

\paragraph*{Doubled skew matrix and block-Toeplitz symbol.}
Define
\begin{equation}
\label{eq:app_Q_XYR_def}
\mathcal Q_\ell^{\rm XYR}
=
\begin{pmatrix}
\mathbf 0 & G_\ell^{{\rm XYR},(\infty)}\\
-\bigl(G_\ell^{{\rm XYR},(\infty)}\bigr)^T & \mathbf 0
\end{pmatrix}.
\end{equation}
Using the same interlaced order and diagonal gauge as for \(XY\), one obtains
\begin{equation}
\label{eq:app_interlaced_block_toeplitz_XYR}
D_\ell^{\rm int}
\Pi_\ell^{\rm int}
\mathcal Q_\ell^{\rm XYR}
(\Pi_\ell^{\rm int})^T
D_\ell^{\rm int}
=
\bigl(A_d^{\rm XYR}\bigr)_{0\le a,b\le\ell-1},
\qquad
d=a-b,
\end{equation}
with
\begin{equation}
\label{eq:app_block_entries_XYR}
A_d^{\rm XYR}
=
\begin{pmatrix}
0 & (-1)^d g_{1-d}\\
-(-1)^d g_{1+d} & 0
\end{pmatrix}.
\end{equation}
Equivalently,
\begin{equation}
\label{eq:app_block_symbol_XYR}
\Phi_A^{\rm XYR}(e^{iq};h,\gamma)
=
\begin{pmatrix}
0 & -e^{-iq}\omega_{\gamma,h}(-e^{iq})\\
e^{iq}\omega_{\gamma,h}(-e^{-iq}) & 0
\end{pmatrix}.
\end{equation}
This is the same doubled block-Toeplitz construction as in the ordinary case, with the scalar phase replaced by \(\omega_{\rm XYR}(z)=z\omega_{\gamma,h}(z^{-1})\).

\paragraph*{Interlaced Pfaffian generating function.}
Since \(XYR\) uses the same interlaced sign convention as \(XY\), its Pfaffian deformation is obtained by the coefficient replacement \(g_{a-b}\mapsto g_{1-a+b}\). Define \((T_\ell^{\rm XYR}(u))_{ab}=(-1)^{a-b}g_{1-a+b}+u^{-1}\varepsilon_{ab}\) and
\begin{equation}
\label{eq:app_K_interlaced_XYR}
K_\ell^{\rm XYR}(u;h,\gamma)
=
\begin{pmatrix}
J_\ell(u) & T_\ell^{\rm XYR}(u;h,\gamma)\\
-\bigl(T_\ell^{\rm XYR}(u;h,\gamma)\bigr)^T & J_\ell(u)
\end{pmatrix}.
\end{equation}
Then, on the corresponding sign-stable branch,
\begin{equation}
\label{eq:app_Pf_interlaced_XYR}
P_\ell^{{\rm XYR},(\infty)}(u;h,\gamma)
=
u^\ell\operatorname{Pf}K_\ell^{\rm XYR}(u;h,\gamma).
\end{equation}
The distribution-valued block symbol for \(K_\ell^{\rm XYR}\) is obtained from Eq.~\eqref{eq:app_PhiK_symbol_XY} by replacing \(\omega_{\gamma,h}\) with \(\omega_{\rm XYR}\). Equivalently, the top-right entry contains \(\omega_{\rm XYR}(-e^{-iq})=-e^{-iq}\omega_{\gamma,h}(-e^{iq})\), and the bottom-left entry contains \(-\omega_{\rm XYR}(-e^{iq})=e^{iq}\omega_{\gamma,h}(-e^{-iq})\).

\paragraph*{Relation to TFI self-duality.}
On the Ising line, for \(h>0\), the shifted block closes inside the TFI family. Indeed, \(f_{1,h^{-1}}(z)=h^{-1}z f_{1,h}(z^{-1})\), and the positive prefactor \(h^{-1}\) does not affect the phase. Therefore
\begin{equation}
\label{eq:app_TFI_self_dual_coefficients}
g_r(h^{-1},1)=g_{1-r}(h,1),
\end{equation}
so that
\begin{equation}
\label{eq:app_XYR_TFI_dual_block}
G_\ell^{{\rm XYR},(\infty)}(h,1)
=
G_\ell^{{\rm TFI},(\infty)}(h^{-1}).
\end{equation}
Consequently \(P_\ell^{{\rm XYR},(\infty)}(u;h,1)=P_\ell^{{\rm TFI},(\infty)}(u;h^{-1})\). Away from \(\gamma=1\), however, \(XYR\) is a shifted interval block of the anisotropic XY kernel and should not be identified with an ordinary XY subsystem at a dual field. This self-duality relation is the subsystem counterpart of the standard Kramers--Wannier duality of the Ising chain~\cite{KramersWannier1941a,Pfeuty1970,Franchini2017}.

\paragraph*{Minor-zero chamber for \(XYR\).}
The \(XYR\) chamber is defined by the same analytic principle as the ordinary chamber, but using minors of \(G_\ell^{{\rm XYR},(\infty)}\). The lower wall is again the circle \(h^2+\gamma^2=1\), because the representative \(z f_{\gamma,h}(z^{-1})\) has the same root-regime boundary as \(f_{\gamma,h}\), up to inversion and a monomial factor.

For even \(\ell=2m\), define the shifted patterned minor
\begin{equation}
\label{eq:app_patterned_minor_XYR}
\mathcal D_m^{\rm XYR}(h,\gamma)
=
\det G_{2m}^{{\rm XYR},(\infty)}(h,\gamma)[I_m,J_m],
\end{equation}
with the same \(I_m=\{0,1,\ldots,m-1\}\) and \(J_m=\{m-1,m,\ldots,2m-2\}\). Let \(\gamma_*^{\rm XYR}(h;\ell)\) denote the selected zero branch reached from the small-\(\gamma\) side outside the circle. Then
\begin{equation}
\label{eq:app_RXYR_stability}
\mathcal R_\ell^{\rm XYR}
=
\left\{
(h,\gamma)\in\mathbb R_{\ge0}^2:
h^2+\gamma^2>1,\quad
0\le\gamma<\gamma_*^{\rm XYR}(h;\ell)
\right\}.
\end{equation}
On the Ising line this reduces to the ordinary TFI chamber at the dual field through Eq.~\eqref{eq:app_TFI_self_dual_coefficients}; away from \(\gamma=1\), it is the sign chamber of the shifted Toeplitz block itself.

\subsection{Direct shifted interval \(XYL\)}

\paragraph*{Shifted symbol, Toeplitz block, and minor polynomial.}
The direct shifted interval representative is \(F_{\rm XYL}(z)=z f_{\gamma,h}(z)\), with phase \(\omega_{\rm XYL}(z)=z\,\omega_{\gamma,h}(z)\). Its Fourier coefficients are \(\kappa_r^{\rm XYL}=g_{r-1}\), and therefore
\begin{equation}
\label{eq:app_G_XYL_def}
G_\ell^{{\rm XYL},(\infty)}(h,\gamma)
=
\bigl[g_{a-b-1}(h,\gamma)\bigr]_{a,b=0}^{\ell-1}.
\end{equation}
The corresponding absolute-minor polynomial is \(P_\ell^{{\rm XYL},(\infty)}(u;h,\gamma)\), as defined in Eq.~\eqref{eq:app_P_X_def}. This block appears in direct Laurent lifts and in shifted subsystem decimations.

The direct shift is algebraically close to the ordinary \(XY\) block, but its natural Pfaffian sign convention is different. The coefficient \(g_{a-b-1}\) aligns the bulk pairs as \((r_{a+1},c_a)\), leaving two boundary labels. This is why the direct orientation is most naturally represented by a snake ordering and a bordered Pfaffian with an ordinary \(XY\) core.

\paragraph*{Interlaced block symbol.}
If one only wants the doubled Toeplitz structure, \(XYL\) can also be interlaced. Starting from
\[
\mathcal Q_\ell^{\rm XYL}=\begin{pmatrix}
\mathbf 0& G_\ell^{{\rm XYL},(\infty)}\\
-\bigl(G_\ell^{{\rm XYL},(\infty)}\bigr)^T&\mathbf 0
\end{pmatrix},
\]
the same interlacing and diagonal gauge produce block entries
\begin{equation}
\label{eq:app_block_entries_XYL_interlaced}
A_d^{\rm XYL}
=
\begin{pmatrix}
0 & (-1)^d g_{d-1}\\
-(-1)^d g_{-d-1} & 0
\end{pmatrix}.
\end{equation}
Equivalently,
\begin{equation}
\label{eq:app_block_symbol_XYL_interlaced}
\Phi_A^{\rm XYL}(e^{iq};h,\gamma)
=
\begin{pmatrix}
0 & -e^{-iq}\omega_{\gamma,h}(-e^{-iq})\\
e^{iq}\omega_{\gamma,h}(-e^{iq}) & 0
\end{pmatrix}.
\end{equation}
This is the doubled block-Toeplitz symbol of the direct shifted matrix. However, the interlaced sign convention is not the coherent sign convention that turns the Pfaffian into the absolute-minor generating function for this orientation. For \(XYL\), the sign-stable branch is represented by the snake ordering described next.

\paragraph*{Snake Pfaffian generating function.}
Use the snake order \(r_0,r_1,c_0,r_2,c_1,\ldots,r_{\ell-1},c_{\ell-2},c_{\ell-1}\). Let \(\Pi_\ell^{\rm XYL}\) be the corresponding permutation matrix. In one-line notation, with one-based indices relative to the original order \((r_0,\ldots,r_{\ell-1},c_0,\ldots,c_{\ell-1})\), this permutation is \(\Pi_\ell^{\rm XYL}=(1,2,\ell+1,3,\ell+2,4,\ell+3,\ldots,\ell,2\ell-1,2\ell)\). In the snake basis, define \(D_\ell^{\rm XYL}=\operatorname{diag}(\delta_1,\ldots,\delta_{2\ell})\), with \(\delta_\mu=(-1)^{\lfloor\mu/2\rfloor}\). Set
\begin{equation}
\label{eq:app_Qhat_XYL}
\widehat{\mathcal Q}_\ell^{\rm XYL}
=
D_\ell^{\rm XYL}
\Pi_\ell^{\rm XYL}
\mathcal Q_\ell^{\rm XYL}
(\Pi_\ell^{\rm XYL})^T
D_\ell^{\rm XYL}.
\end{equation}
Introduce the universal auxiliary skew matrix
\begin{equation}
\label{eq:app_R_auxiliary}
(R_{2\ell})_{\mu\nu}
=
\begin{cases}
(-1)^{\mu+\nu+1}, & \mu<\nu,\\
0, & \mu=\nu,\\
-(-1)^{\mu+\nu+1}, & \mu>\nu .
\end{cases}
\end{equation}
Then, on the snake sign-stable branch,
\begin{equation}
\label{eq:app_Pf_XYL}
P_\ell^{{\rm XYL},(\infty)}(u;h,\gamma)
=
\operatorname{Pf}\!\left(
R_{2\ell}
+
u\,\widehat{\mathcal Q}_\ell^{\rm XYL}
\right)
=
u^\ell\operatorname{Pf}\!\left(
u^{-1}R_{2\ell}
+
\widehat{\mathcal Q}_\ell^{\rm XYL}
\right).
\end{equation}
This is the analogue of Eq.~\eqref{eq:app_Pf_interlaced_XY} for the direct orientation. The Pfaffian mechanism is the same, but the sign convention is different: \(XY\) and \(XYR\) use the interlaced \(J,T\) formula, whereas \(XYL\) uses the snake \(R+\widehat{\mathcal Q}\) formula.

\paragraph*{Bordered Pfaffian and reduction to the \(XY\) core.}
The snake-gauged matrix \(\widehat{\mathcal Q}_\ell^{\rm XYL}\) has a useful block interpretation. Apart from the two boundary labels \(r_0\) and \(c_{\ell-1}\), the snake ordering groups the bulk labels into pairs \((r_{a+1},c_a)\), \(a=0,\ldots,\ell-2\). The corresponding \((\ell-1)\times(\ell-1)\) bulk core is a \(2\times2\) block Toeplitz matrix with blocks
\begin{equation}
\label{eq:app_XYL_snake_core_blocks}
C_d^{\rm XYL}
=
\begin{pmatrix}
0 & (-1)^d g_d\\
-(-1)^d g_{-d} & 0
\end{pmatrix},
\qquad d=a-b.
\end{equation}
Thus the bulk core of the \(XYL\) snake construction is the ordinary \(XY\) doubled Toeplitz core of size \(\ell-1\).

More explicitly, put the two boundary labels \(r_0\) and \(c_{\ell-1}\) first and group the bulk into pairs \((r_{a+1},c_a)\), \(a=0,\ldots,\ell-2\). Set \(n_{\rm c}=\ell-1\). The snake Pfaffian can then be written as the bordered Pfaffian
\begin{equation}
\label{eq:app_XYL_bordered_pfaffian}
P_\ell^{{\rm XYL},(\infty)}(u;h,\gamma)
=
\operatorname{Pf}
\begin{pmatrix}
0 & \beta_\ell^{\rm XYL} & (p_\ell^{\rm XYL})^T\\
-\beta_\ell^{\rm XYL} & 0 & (q_\ell^{\rm XYL})^T\\
-p_\ell^{\rm XYL} & -q_\ell^{\rm XYL} & H_{n_{\rm c}}^{\rm XY}(u)
\end{pmatrix},
\qquad
H_{n_{\rm c}}^{\rm XY}(u)=u\,K_{n_{\rm c}}^{\rm XY}(u).
\end{equation}
The boundary data are
\begin{equation}
\label{eq:app_XYL_boundary_coefficients}
\begin{aligned}
\beta_\ell^{\rm XYL}
&=
1+u(-1)^\ell g_{-\ell},\\
(p_\ell^{\rm XYL})_a
&=
\begin{pmatrix}
1\\
-1+u(-1)^{a+1}g_{-a-1}
\end{pmatrix},\\
(q_\ell^{\rm XYL})_a
&=
\begin{pmatrix}
1+u(-1)^{\ell+a}g_{a+1-\ell}\\
-1
\end{pmatrix},
\end{aligned}
\end{equation}
where \(a=0,\ldots,n_{\rm c}-1\). The \(2\times2\) block of the ordinary \(XY\) core is
\begin{equation}
\label{eq:app_XYL_H_blocks}
(H_{n_{\rm c}}^{\rm XY})_{ab}
=
\begin{pmatrix}
\operatorname{sgn}(a-b)
&
\varepsilon_{ab}+u(-1)^{a-b}g_{a-b}
\\
-\varepsilon_{ba}-u(-1)^{a-b}g_{b-a}
&
\operatorname{sgn}(a-b)
\end{pmatrix}.
\end{equation}
When \(H_{\ell-1}^{\rm XY}(u)\) is invertible, the Pfaffian Schur identity gives
\begin{equation}
\label{eq:app_XYL_log_factorization}
\ln P_\ell^{{\rm XYL},(\infty)}(u;h,\gamma)
=
\ln P_{\ell-1}^{{\rm XY},(\infty)}(u;h,\gamma)
+
\ln \mathcal B_\ell^{\rm XYL}(u;h,\gamma),
\end{equation}
where the scalar boundary factor is
\begin{equation}
\label{eq:app_XYL_boundary_factor}
\mathcal B_\ell^{\rm XYL}(u;h,\gamma)
=
\beta_\ell^{\rm XYL}
+
(p_\ell^{\rm XYL})^T
\left(H_{\ell-1}^{\rm XY}(u)\right)^{-1}
q_\ell^{\rm XYL}.
\end{equation}
Equivalently, one solves \(H_{\ell-1}^{\rm XY}(u)y=q_\ell^{\rm XYL}\) and evaluates \(\mathcal B_\ell^{\rm XYL}=\beta_\ell^{\rm XYL}+(p_\ell^{\rm XYL})^T y\).

This factorization shows that the direct shifted interval is governed by an ordinary \(XY\) core supplemented by two boundary labels. The distinction between \(XY\) and \(XYL\) is therefore encoded in the boundary factor \(\mathcal B_\ell^{\rm XYL}\), while the core Pfaffian is the same object already defined in the ordinary interval subsection.

\begin{figure}[!ht]
\centering
\includegraphics[width=\linewidth]{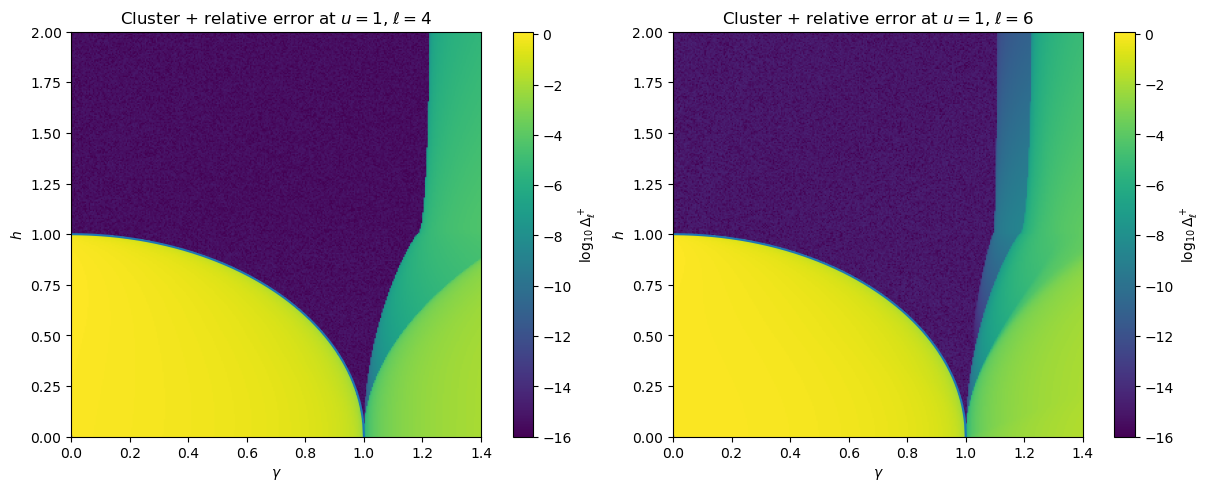}
\caption{
Numerical validation of the \(XYL\) Pfaffian formula at \(u=1\) for subsystem sizes \(\ell=4\) and \(\ell=6\). The color scale shows \(\log_{10}\Delta_\ell^{\rm XYL}\), where
\(\Delta_\ell^{\rm XYL}=
\frac{\left|P_\ell^{{\rm XYL},\mathrm{brute}}(1)-P_\ell^{{\rm XYL},\mathrm{Pf}}(1)\right|}
{\max\{1,\left|P_\ell^{{\rm XYL},\mathrm{brute}}(1)\right|\}}\).
Here \(P_\ell^{{\rm XYL},\mathrm{brute}}\) is obtained by explicit summation over all absolute minors of the shifted Toeplitz block \(G_\ell^{{\rm XYL},(\infty)}=[g_{a-b-1}]\), while \(P_\ell^{{\rm XYL},\mathrm{Pf}}\) is computed from the snake/bordered Pfaffian representation. The blue curve marks the circle wall \(h^2+\gamma^2=1\). The Pfaffian formula agrees with the brute-force result throughout the massive branch outside the circle, whereas the loss of agreement inside the circle signals crossing to a different minor-sign branch.
}
\label{fig:app-subsystem-pfaffian-validationcluster}
\end{figure}

\paragraph*{Stability chamber for \(XYL\).}
The \(XYL\) stability analysis uses the same analytic principle, but in the snake convention. Away from the singular locus of the scalar phase, every minor of \(G_\ell^{{\rm XYL},(\infty)}\) is analytic in \((h,\gamma)\), and sign changes occur only on minor-zero loci. The lower wall is unchanged because \(z f_{\gamma,h}(z)\) has the same nonzero roots as \(f_{\gamma,h}(z)\). Hence \(h^2+\gamma^2=1\) again separates the non-oscillatory massive branch from the oscillatory regime.

The finite-size bulk sign boundary is inherited from the ordinary \(XY\) core. Indeed, \(G_\ell^{{\rm XYL},(\infty)}[\{1,\ldots,\ell-1\},\{0,\ldots,\ell-2\}]=G_{\ell-1}^{{\rm XY},(\infty)}\). Thus the bulk minor-zero problem for \(XYL\) is the ordinary \(XY\) minor-zero problem at the core size \(n_{\rm c}=\ell-1\). Let \(m_{\rm c}=\lceil n_{\rm c}/2\rceil\), and define
\begin{equation}
\label{eq:app_patterned_minor_XYL}
\mathcal D_{\ell}^{\rm XYL}(h,\gamma)
=
\det G_\ell^{{\rm XYL},(\infty)}(h,\gamma)
[I_{m_{\rm c}}^{\rm XYL},J_{m_{\rm c}}^{\rm XYL}],
\end{equation}
where \(I_{m_{\rm c}}^{\rm XYL}=\{1,\ldots,m_{\rm c}\}\) and \(J_{m_{\rm c}}^{\rm XYL}=\{m_{\rm c}-1,\ldots,2m_{\rm c}-2\}\). Since \((G_\ell^{{\rm XYL},(\infty)})_{a+1,b}=g_{a-b}\), this is precisely the ordinary \(XY\) patterned minor of the \((\ell-1)\)-site core. Consequently, \(\gamma_*^{\rm XYL}(h;\ell)=\gamma_*^{\rm XY}(h;\ell-1)\), up to the additional boundary condition described below.

The boundary labels removed in the core reduction are encoded by \(\mathcal B_\ell^{\rm XYL}\). Hence the full \(XYL\) chamber is the connected component, attached to the large-\(h\) point, where the ordinary \(XY\) core stays on its stable branch and the boundary factor does not cross zero. For fixed \(u>0\),
\begin{equation}
\label{eq:app_RXYL_stability}
\mathcal R_\ell^{\rm XYL}(u)
=
\left\{
(h,\gamma)\in\mathbb R_{\ge0}^2:
h^2+\gamma^2>1,\quad
0\le\gamma<\gamma_*^{\rm XY}(h;\ell-1),\quad
\operatorname{sgn}\mathcal B_\ell^{\rm XYL}(u;h,\gamma)
=
\operatorname{sgn}\mathcal B_\ell^{\rm XYL}(u;h,\gamma)\big|_{\rm large\text{-}h}
\right\}.
\end{equation}
For the \(\alpha=\frac12\) SRE one sets \(u=1\). If the full generating function is needed, the last condition should hold for the positive \(u\)-range under consideration. Numerically, \(\mathcal B_\ell^{\rm XYL}(u;h,\gamma)\) has no additional zero inside the massive \(XY\) core chamber. In that case, the practical chamber reduces to the ordinary core chamber on the snake branch connected to large \(h\).

The numerical validation of the snake/bordered Pfaffian formula is shown in Fig.~\ref{fig:app-subsystem-pfaffian-validationcluster}. The comparison with brute-force absolute-minor summation confirms that the formula reproduces the direct shifted interval polynomial on the selected branch.

\subsection{Practical dictionary}

The three square-interval representatives used in the main text are summarized in Table~\ref{tab:app_generating_pfaffian_dictionary}.

\begin{table}[!ht]
\caption{Pfaffian generating functions for the three square subsystem blocks. The labels \(XYR\) and \(XYL\) refer to shifted interval representatives of the same infinite-volume XY kernel.}
\label{tab:app_generating_pfaffian_dictionary}
\begin{ruledtabular}
\begin{tabular}{llll}
\rowcolor{blue!10}
Block & Scalar representative & Matrix entries & Pfaffian convention \\
\hline
\rowcolor{blue!3}
\(XY\) & \(f_{\gamma,h}(z)\) &
\(G_\ell^{{\rm XY},(\infty)}=[g_{a-b}]\) &
Interlaced \(J,T\) \\
\(XYR\) & \(z f_{\gamma,h}(z^{-1})\) &
\(G_\ell^{{\rm XYR},(\infty)}=[g_{1-a+b}]\) &
Interlaced \(J,T\) \\
\rowcolor{blue!3}
\(XYL\) & \(z f_{\gamma,h}(z)\) &
\(G_\ell^{{\rm XYL},(\infty)}=[g_{a-b-1}]\) &
Snake \(R+\widehat{\mathcal Q}\), or bordered \(XY\) core \\
\end{tabular}
\end{ruledtabular}
\end{table}

For rectangular blocks, which appear when a decimated interval does not split evenly into residue sectors, the square Pfaffian formulas above are replaced by the rectangular determinant-sum functional described in Appendix~\ref{app:xy_decimation}. In all cases, the Pfaffian formulas represent absolute-minor generating functions only on the selected sign-stable branch. Outside that branch, they remain well-defined signed Pfaffians but no longer equal the absolute-value sum.
\section{Asymptotics of subsystem generating functions}
\label{sec:subsystem_asymptotics}

In this appendix, we study the large-\(\ell\) asymptotics of the absolute-minor generating functions associated with contiguous intervals of the infinite XY chain and with the two shifted interval representatives used in the subsystem correspondences. We use the notation \(\chi\in\{{\rm XY},{\rm XYR},{\rm XYL}\}\), where \({\rm XYR}\) denotes the Kramers--Wannier/cluster-dual shifted interval and \({\rm XYL}\) denotes the direct cluster-oriented shifted interval. The three interval blocks are built from the same XY Fourier coefficients \(g_r(h,\gamma)\), but with different shifts:
\begin{equation}
\label{eq:sub_three_blocks}
\mathbf G_\ell^{{\rm XY},(\infty)}=[g_{a-b}],
\qquad
\mathbf G_\ell^{{\rm XYR},(\infty)}=[g_{1-a+b}],
\qquad
\mathbf G_\ell^{{\rm XYL},(\infty)}=[g_{a-b-1}],
\qquad
a,b=0,\ldots,\ell-1 .
\end{equation}
Equivalently, \(\mathbf G_\ell^{\chi,(\infty)}=[\kappa^\chi_{a-b}]\), with \(\kappa_r^{\rm XY}=g_r\), \(\kappa_r^{\rm XYR}=g_{1-r}\), and \(\kappa_r^{\rm XYL}=g_{r-1}\). The ordinary \(XY\) and \({\rm XYR}\) blocks use the same interlaced Pfaffian convention. The \({\rm XYL}\) block uses the snake convention, but its snake bulk is the ordinary \(XY\) Pfaffian kernel of size \(\ell-1\) plus two boundary insertions.

The purpose of this appendix is to formulate the scaling structure of these three Pfaffian generating functions. The ordinary \(XY\) block is used as the reference problem. The shifted blocks share the same infinite-volume XY kernel and therefore the same bulk root structure, correlation length, and extensive density. Their \(O(1)\) constants, finite-\(\ell\) sign chambers, and crossover functions are, however, sensitive to the shifted interval representative. We therefore separate the common bulk data from the block-dependent boundary data. This distinction is essential because the exact subsystem object is a Pfaffian with a \(2\times2\) block-Toeplitz origin, while the scalar symbol introduced below is only an effective bulk symbol. A complete derivation of the shifted constants and shifted crossover functions would require the corresponding block-Pfaffian scaling problem.

The Toeplitz and Fisher--Hartwig background used below follows the standard literature on Toeplitz determinants and free-fermion correlation matrices~\cite{Szego1959,Widom1976,BasorEhrhardt2001,BottcherSilbermann2006,DeiftItsKrasovsky2011,ItsJinKorepin2005,JinKorepin2004,PeschelEisler2009,Franchini2017}. Closely related XY-chain crossover analyses for
formation probabilities and full counting statistics were developed in
Refs.~\cite{AresViti2020EFP,AresRajabpourViti2021FCS}.

\paragraph*{Common chamber and bulk data.}

The infinite-volume XY phase symbol is \(\omega_{\gamma,h}(e^{iq})=(h+\cos q+i\gamma\sin q)/[(h+\cos q)^2+\gamma^2\sin^2 q]^{1/2}\), with \(g_r=(2\pi)^{-1}\int_0^{2\pi}\omega_{\gamma,h}(e^{iq})e^{-irq}\,dq\). For each \(\chi\), define
\begin{equation}
\label{eq:sub_P_chi_def}
P_\ell^{\chi,(\infty)}(u;h,\gamma)
=
\sum_{p=0}^{\ell}u^p
\sum_{\substack{I,J\subseteq\{0,\ldots,\ell-1\}\\ |I|=|J|=p}}
\left|
\det \mathbf G_\ell^{\chi,(\infty)}(h,\gamma)[I,J]
\right| .
\end{equation}
On the corresponding sign-stable branch, \(P_\ell^{{\rm XY},(\infty)}\) and \(P_\ell^{{\rm XYR},(\infty)}\) are represented by the interlaced Pfaffian formula \(P_\ell^{\chi,(\infty)}=u^\ell\operatorname{Pf}\mathbf K_\ell^\chi\), with \(\chi\in\{{\rm XY},{\rm XYR}\}\), and \(\mathbf K_\ell^{\rm XYR}\) is obtained from \(\mathbf K_\ell^{\rm XY}\) by the replacement \(g_{a-b}\mapsto g_{1-a+b}\). For the \({\rm XYL}\) block, the snake Pfaffian formula gives \(P_\ell^{{\rm XYL},(\infty)}(u;h,\gamma)=\operatorname{Pf}(\mathbf R_{2\ell}+u\widehat{\mathcal Q}_\ell^{\rm XYL})\).

The exact Pfaffian representation is the algebraic starting point. Nevertheless, for the leading extensive term one may use a scalar effective symbol. The reason is not that the block symbol literally reduces to a scalar symbol. Rather, the scalar symbol captures the common one-particle spectral branch of the XY kernel. This can be justified in four complementary ways. First, in the periodic problem the same scalar factor appears exactly in the mode-by-mode product formula; hence it gives the correct bulk branch before boundary truncation is introduced. Second, the shifts defining \(XYR\) and \(XYL\) are monomial and inversion modifications of the same XY phase, so they do not change the nonzero root structure of the bulk Laurent polynomial. Third, the local singularity at \(h=1\), \(q=\pi\), which fixes the Fisher--Hartwig logarithmic exponent, is the same for the three representatives. Fourth, the finite-\(\ell\) Pfaffian chamber analysis shows that the three sign-stable branches are organized by the same circle wall and approach the same large-\(\ell\) anisotropic chamber. Thus the scalar symbol is used only for the common bulk density and scaling variable, while all boundary-sensitive information is kept in block-dependent constants and crossover functions.

At finite \(\ell\), the minor-zero boundary depends on the chosen block. In the scaling limit, however, the common chamber is naturally organized by \(h^2+\gamma^2>1\), \(0\le\gamma<1\), \(h\ge0\), with the TFI line \(\gamma=1\) treated as a limiting case. The lower wall \(h^2+\gamma^2=1\) is inherited from the root structure of the XY Laurent polynomial, while the upper boundary of the finite-\(\ell\) sign-stable chamber is expected to approach the Ising line. Within this chamber, the three blocks share the same bulk scale, but not necessarily the same \(O(1)\) data.

The scalar determinant-level symbol controlling the common extensive contribution is
\begin{equation}
\label{eq:sub_symbol_bulk}
a_{u;h,\gamma}(e^{iq})
=
1+u^2+2u\Lambda_{h,\gamma}(q),
\qquad
\Lambda_{h,\gamma}(q)
=
\frac{h+\gamma+(1-\gamma)\cos q}
{\sqrt{(h+\cos q)^2+\gamma^2\sin^2q}} .
\end{equation}
We write \(f(u;h,\gamma)=(2\pi)^{-1}\int_0^{2\pi}\ln a_{u;h,\gamma}(e^{iq})\,dq\). Since the generating function is represented at Pfaffian level, the common extensive contribution is \(\ell f/2\). The inverse correlation length is
\begin{equation}
\label{eq:sub_inverse_correlation_length_def}
\xi^{-1}(h,\gamma)
=
|\ln\rho_*(h,\gamma)|,
\qquad
\rho_*(h,\gamma)
=
\frac{h+\sqrt{h^2+\gamma^2-1}}{1+\gamma}.
\end{equation}
For fixed \(0<\gamma\le1\), \(\xi^{-1}(h,\gamma)\sim |1-h|/\gamma\) as \(h\to1\). Thus the scaling variable \(x=\ell/\xi(h,\gamma)\) compares the interval length with the correlation length.

Equation~\eqref{eq:sub_symbol_bulk} should therefore be read as a bulk symbol, not as the full subsystem block symbol. The full block symbols remain necessary to determine the constants \(C_{\rm I}^{P,\chi}\), \(C_{\rm II}^{P,\chi}\), and the crossover functions \(\Phi_\chi^{(\sigma)}\). In particular, two blocks with the same \(f(u;h,\gamma)\) can differ by boundary factors, partial indices, or sign-branch data. The notation below is chosen to keep these distinctions explicit.

\subsection{Ordinary \(XY\) interval}

The ordinary block \(P_\ell^{{\rm XY},(\infty)}\) is the reference problem. Its Pfaffian representation was constructed in Appendix~\ref{app:subsystem_toeplitz_pfaffian}; here we use it to define the three scaling regimes.

\paragraph*{Off-critical regime.}
Let \(x=\ell/\xi(h,\gamma)\). In the off-critical regime, \(h\neq1\), \(0<\gamma\le1\), and \(x\to\infty\). The scalar bulk symbol is smooth, and the interval is much larger than the correlation length. The expected strong Szegő--Widom form is
\begin{equation}
\label{eq:sub_regimeI_XY}
\ln P_\ell^{{\rm XY},(\infty)}(u;h,\gamma)
=
\frac{\ell}{2}f(u;h,\gamma)
+
C_{\rm I}^{P,{\rm XY}}(u;h,\gamma)
+
O(e^{-\ell/\xi(h,\gamma)}) .
\end{equation}
For the ordinary block, the scalar branch gives the explicit Szegő-type candidate
\begin{equation}
\label{eq:sub_CI_XY}
C_{\rm I}^{P,{\rm XY}}(u;h,\gamma)
=
\frac12
\sum_{n=1}^{\infty}
n\,W_nW_{-n},
\qquad
W_n=
\frac{1}{2\pi}
\int_0^{2\pi}
e^{-inq}\ln a_{u;h,\gamma}(e^{iq})\,dq .
\end{equation}
The factor \(1/2\) reflects that \(P_\ell\) is represented at Pfaffian level, equivalently as a half-log of the corresponding determinant problem. Equation~\eqref{eq:sub_CI_XY} should be understood as the ordinary-\(XY\) scalar Szegő candidate on the selected branch. It is not a substitute for the full block-Pfaffian analysis of the shifted representatives. For the shifted blocks, this scalar formula fixes only the shared bulk part; the constants may acquire additional boundary contributions. Although \(C_{\rm I}^{P,{\rm XY}}\) is an \(O(1)\) constant for fixed \(h\neq1\), it becomes singular as the critical line is approached. Matching with the crossover regime requires
\begin{equation}
\label{eq:sub_CI_nearcritical_XY}
C_{\rm I}^{P,{\rm XY}}(u;h,\gamma)
=
\frac18\ln \xi^{-1}(h,\gamma)
+
C_{{\rm I},0}^{P,{\rm XY},\sigma}(u,\gamma)
+
o(1),
\qquad
h\to1^\sigma .
\end{equation}
The label \(\sigma=\operatorname{sgn}(h-1)\) is retained because the two approaches to the critical line may differ by branch-dependent \(O(1)\) matching data. The logarithm is not an additional finite-size correction in the off-critical regime; it is the near-critical singularity of the \(O(1)\) Szegő--Widom constant.

\paragraph*{Critical regime.}
At \(h=1\), the smooth-symbol hypothesis fails at \(q=\pi\). Near \(q=\pi+k\), one has \(a_{u;1,\gamma}(e^{i(\pi+k)})\sim4u/|k|\). This local singularity is independent of the microscopic anisotropy after the natural correlation-length scaling. It is also insensitive to the monomial shifts used to define \(XYR\) and \(XYL\), which is why the logarithmic coefficient is common. The corresponding Pfaffian Fisher--Hartwig contribution gives
\begin{equation}
\label{eq:sub_regimeII_XY}
\ln P_\ell^{{\rm XY},(\infty)}(u;1,\gamma)
=
\frac{\ell}{2}f_c(u;\gamma)
-
\frac18\ln\ell
+
C_{\rm II}^{P,{\rm XY}}(u;\gamma)
+
o(1),
\end{equation}
where \(f_c(u;\gamma)=f(u;1,\gamma)\). For the ordinary XY block, factor \(a_{u;1,\gamma}(e^{iq})=|2-2\cos(q-\pi)|^{-1/2}b_{u,\gamma}(e^{iq})\), with
\begin{equation}
\label{eq:sub_b_smooth_def}
b_{u,\gamma}(e^{iq})
=
2|\cos(q/2)|(1+u^2)
+
4u
\frac{
\cos^2(q/2)+\gamma\sin^2(q/2)
}{
[\cos^2(q/2)+\gamma^2\sin^2(q/2)]^{1/2}
}.
\end{equation}
If \(V_n=(2\pi)^{-1}\int_0^{2\pi}e^{-inq}\ln b_{u,\gamma}(e^{iq})\,dq\), the scalar/Fisher--Hartwig normalization gives
\begin{equation}
\label{eq:sub_CII_XY}
\begin{aligned}
C_{\rm II}^{P,{\rm XY}}(u;\gamma)=
\frac12\sum_{n=1}^{\infty}nV_nV_{-n}
+
\frac14\bigl[\ln(4u)-V_0\bigr]+
\frac12\ln G\!\left(\frac12\right)
-
\frac14\ln2
-
\frac1{16}\ln\gamma .
\end{aligned}
\end{equation}
Here \(G\) denotes the Barnes \(G\)-function. The expression \eqref{eq:sub_CII_XY} should be understood as the ordinary-\(XY\) scalar/Fisher--Hartwig normalization on the selected branch. It is not used to claim that the shifted constants are identical. Instead, the shifted representatives inherit the same bulk and the same logarithmic exponent, while their constants are shifted by boundary data.

\paragraph*{Near-critical crossover.}
Let \(x=\ell/\xi(h,\gamma)\) and take \(\ell\to\infty\), \(h\to1\), with \(x=O(1)\). We write \(h=h_\ell^\sigma\), where \(\ell/\xi(h_\ell^\sigma,\gamma)=x\) and \(\sigma=\operatorname{sgn}(h_\ell^\sigma-1)\). The condition \(x=O(1)\) defines the crossover scaling limit and is independent of any further refinement of finite-size corrections. The local phase is rounded at the scale \(q=\pi+s/\ell\):
\begin{equation}
\label{eq:sub_local_phase}
\omega_{\gamma,h_\ell^\sigma}(e^{i(\pi+s/\ell)})
\longrightarrow
\frac{\sigma x-is}{(x^2+s^2)^{1/2}} .
\end{equation}
After the correlation-length scaling \(x=\ell/\xi(h,\gamma)\), the anisotropy \(\gamma\) drops out of the local phase. Its remaining effect is carried by the bulk density and by finite constants. This local limit explains why the scalar symbol organizes the scaling variable and the leading singular form: the mass dependence enters through the universal combination \((\sigma x-is)/(x^2+s^2)^{1/2}\). However, the local limit alone does not determine the full crossover function, because the latter also depends on how the local mode is embedded in the block-Pfaffian boundary problem.

For a general generating parameter \(u\), the crossover branches are kept as independent scaling functions. The crossover formula uses the moving bulk:
\begin{equation}
\label{eq:sub_regimeIII_XY}
\ln P_\ell^{{\rm XY},(\infty)}(u;h_\ell^\sigma,\gamma)
=
\frac{\ell}{2}f(u;h_\ell^\sigma,\gamma)
-
\frac18\ln\ell
+
C_{\rm II}^{P,{\rm XY}}(u;\gamma)
+
\Phi_{\rm XY}^{(\sigma)}(x;u,\gamma)
+
o(1).
\end{equation}
The moving bulk already contains the regular off-critical correction, so no separate linear term is inserted into \(\Phi_{\rm XY}^{(\sigma)}\). Equivalently, the crossover function is defined by
\begin{equation}
\label{eq:sub_Phi_def_XY}
\Phi_{\rm XY}^{(\sigma)}(x;u,\gamma)
=
\lim_{\ell\to\infty}
\left[
\ln P_\ell^{{\rm XY},(\infty)}(u;h_\ell^\sigma,\gamma)
-
\frac{\ell}{2}f(u;h_\ell^\sigma,\gamma)
+
\frac18\ln\ell
-
C_{\rm II}^{P,{\rm XY}}(u;\gamma)
\right],
\end{equation}
provided the limit is taken at fixed \(x=\ell/\xi(h_\ell^\sigma,\gamma)\) and fixed branch \(\sigma\). The endpoint conditions are
\begin{equation}
\label{eq:sub_Phi_endpoints_XY}
\Phi_{\rm XY}^{(\sigma)}(0;u,\gamma)=0,
\qquad
\Phi_{\rm XY}^{(\sigma)}(x;u,\gamma)
=
\frac18\ln x
+
\Phi_{{\rm XY},\infty}^{(\sigma)}(u,\gamma)
+
o(1)
\quad
(x\to\infty).
\end{equation}
The logarithm \((1/8)\ln x\) is required by matching, since \(-\frac18\ln\ell+\frac18\ln x=\frac18\ln \xi^{-1}(h_\ell^\sigma,\gamma)\). No equality between the two branches is assumed in this general generating-function formulation.

\subsection{Kramers--Wannier shifted interval \(XYR\)}

The \({\rm XYR}\) block is obtained from the ordinary XY block by the coefficient shift \(g_{a-b}\mapsto g_{1-a+b}\), or equivalently by the scalar representative \(\omega_{\rm XYR}(z)=z\omega_{\gamma,h}(z^{-1})\). Its Pfaffian representation uses the same interlaced convention as \(XY\). Therefore, the extensive density and correlation length are the same as in the ordinary XY problem. However, the finite constants and crossover functions are allowed to differ because the shifted block has a different boundary orientation and a different finite-\(\ell\) sign chamber.

The reason the same scalar bulk applies to \(XYR\) is that multiplication by \(z\) and inversion \(z\mapsto z^{-1}\) do not change the bulk root scale of the XY Laurent polynomial. They change the orientation of the interval block and hence the finite boundary data, but not the extensive density obtained from the common bulk branch. This is why \(f(u;h,\gamma)\), \(\xi(h,\gamma)\), and the logarithmic exponent are shared, while the constants below are kept independent.

In the off-critical regime,
\begin{equation}
\label{eq:sub_regimeI_XYR}
\ln P_\ell^{{\rm XYR},(\infty)}(u;h,\gamma)
=
\frac{\ell}{2}f(u;h,\gamma)
+
C_{\rm I}^{P,{\rm XYR}}(u;h,\gamma)
+
O(e^{-\ell/\xi(h,\gamma)}) ,
\end{equation}
with \(C_{\rm I}^{P,{\rm XYR}}=C_{\rm I}^{P,{\rm XY}}+B_{\rm I}^{\rm XYR}\). At criticality,
\begin{equation}
\label{eq:sub_regimeII_XYR}
\ln P_\ell^{{\rm XYR},(\infty)}(u;1,\gamma)
=
\frac{\ell}{2}f_c(u;\gamma)
-
\frac18\ln\ell
+
C_{\rm II}^{P,{\rm XYR}}(u;\gamma)
+
o(1),
\end{equation}
with \(C_{\rm II}^{P,{\rm XYR}}=C_{\rm II}^{P,{\rm XY}}+B_{\rm II}^{\rm XYR}\). In the crossover regime,
\begin{equation}
\label{eq:sub_regimeIII_XYR}
\ln P_\ell^{{\rm XYR},(\infty)}(u;h_\ell^\sigma,\gamma)
=
\frac{\ell}{2}f(u;h_\ell^\sigma,\gamma)
-
\frac18\ln\ell
+
C_{\rm II}^{P,{\rm XYR}}(u;\gamma)
+
\Phi_{\rm XYR}^{(\sigma)}(x;u,\gamma)
+
o(1),
\end{equation}
where
\begin{equation}
\label{eq:sub_Phi_XYR_split}
\Phi_{\rm XYR}^{(\sigma)}(x;u,\gamma)
=
\Phi_{\rm XY}^{(\sigma)}(x;u,\gamma)
+
B_{\rm cross}^{{\rm XYR},\sigma}(x;u,\gamma),
\qquad
B_{\rm cross}^{{\rm XYR},\sigma}(0;u,\gamma)=0 .
\end{equation}

On the Ising line \(\gamma=1\), the shifted block closes inside the TFI family through the usual self-duality relation \(h\mapsto h^{-1}\). Away from \(\gamma=1\), \(XYR\) remains a shifted anisotropic-XY interval block, and the boundary shifts above should be kept as independent data.

\subsection{Direct shifted interval \(XYL\)}

The \({\rm XYL}\) block uses the direct shifted representative \(\omega_{\rm XYL}(z)=z\omega_{\gamma,h}(z)\). In contrast with \(XYR\), its natural Pfaffian representation uses the snake convention. The key structural simplification is that the snake bulk is the ordinary XY Pfaffian kernel of size \(\ell-1\), supplemented by two boundary labels. Thus the \({\rm XYL}\) block has no independent bulk asymptotics.

Here, the use of the common scalar bulk is even more direct. The snake construction isolates an ordinary \(XY\) core of size \(\ell-1\). Therefore, the extensive density is fixed by the same \(f(u;h,\gamma)\), with a size-shift correction in the constant term. The only additional data are contained in the boundary Green function generated by the two boundary labels.

The bordered Pfaffian form is
\begin{equation}
\label{eq:sub_XYL_bordered}
P_\ell^{{\rm XYL},(\infty)}(u;h,\gamma)
=
\operatorname{Pf}
\begin{pmatrix}
0 & \beta_\ell^{\rm XYL} & (\mathbf p_\ell^{\rm XYL})^T\\
-\beta_\ell^{\rm XYL} & 0 & (\mathbf q_\ell^{\rm XYL})^T\\
-\mathbf p_\ell^{\rm XYL} & -\mathbf q_\ell^{\rm XYL} & \mathbf H_{\ell-1}^{\rm XY}(u)
\end{pmatrix},
\qquad
\mathbf H_{\ell-1}^{\rm XY}(u)=u\,\mathbf K_{\ell-1}^{\rm XY}(u).
\end{equation}
With \(a=0,\ldots,\ell-2\), the boundary entries are \(\beta_\ell^{\rm XYL}=1+u(-1)^\ell g_{-\ell}\), \(\mathbf p_a^{\rm XYL}=(1,\,-1+u(-1)^{a+1}g_{-a-1})^T\), and \(\mathbf q_a^{\rm XYL}=(1+u(-1)^{\ell+a}g_{a+1-\ell},\,-1)^T\). The \(2\times2\) bulk block of \(\mathbf H_{\ell-1}^{\rm XY}\), for \(d=a-b\), is \((\mathbf H_{\ell-1}^{\rm XY})_{ab}=\bigl(\begin{smallmatrix}\operatorname{sgn}(d)&\varepsilon_{ab}+u(-1)^dg_d\\ -\varepsilon_{ba}-u(-1)^dg_{-d}&\operatorname{sgn}(d)\end{smallmatrix}\bigr)\), where \(\varepsilon_{ab}=1\) for \(a\le b\) and \(\varepsilon_{ab}=-1\) for \(a>b\).

When \(\mathbf H_{\ell-1}^{\rm XY}\) is invertible, the Pfaffian Schur identity gives
\begin{equation}
\label{eq:sub_XYL_factorization}
\ln P_\ell^{{\rm XYL},(\infty)}(u;h,\gamma)
=
\ln P_{\ell-1}^{{\rm XY},(\infty)}(u;h,\gamma)
+
\ln\mathcal B_\ell^{\rm XYL}(u;h,\gamma),
\end{equation}
where
\begin{equation}
\label{eq:sub_XYL_boundary_factor}
\mathcal B_\ell^{\rm XYL}(u;h,\gamma)
=
\beta_\ell^{\rm XYL}
+
(\mathbf p_\ell^{\rm XYL})^T
\left(\mathbf H_{\ell-1}^{\rm XY}(u)\right)^{-1}
\mathbf q_\ell^{\rm XYL}.
\end{equation}
Thus the shift contributes through the scalar boundary Green function \(\mathcal B_\ell^{\rm XYL}\).

This exact factorization fixes the way in which the \({\rm XYL}\) constants are related to the ordinary XY constants. The shift from \(\ell\) to \(\ell-1\) in the core produces a subtraction of \(f/2\) in the off-critical constant and \(f_c/2\) in the critical constant. The remaining difference is a genuine boundary term.

In the off-critical regime,
\begin{equation}
\label{eq:sub_regimeI_XYL}
\ln P_\ell^{{\rm XYL},(\infty)}(u;h,\gamma)
=
\frac{\ell}{2}f(u;h,\gamma)
+
C_{\rm I}^{P,{\rm XYL}}(u;h,\gamma)
+
O(e^{-\ell/\xi(h,\gamma)}) ,
\end{equation}
where
\begin{equation}
\label{eq:sub_XYL_CI}
C_{\rm I}^{P,{\rm XYL}}(u;h,\gamma)
=
C_{\rm I}^{P,{\rm XY}}(u;h,\gamma)
-\frac12 f(u;h,\gamma)
+
B_{\rm I}^{\rm XYL}(u;h,\gamma),
\end{equation}
with \(B_{\rm I}^{\rm XYL}(u;h,\gamma)=\lim_{\ell\to\infty}\ln\mathcal B_\ell^{\rm XYL}(u;h,\gamma)\), whenever the limit exists. At criticality,
\begin{equation}
\label{eq:sub_regimeII_XYL}
\ln P_\ell^{{\rm XYL},(\infty)}(u;1,\gamma)
=
\frac{\ell}{2}f_c(u;\gamma)
-
\frac18\ln\ell
+
C_{\rm II}^{P,{\rm XYL}}(u;\gamma)
+
o(1),
\end{equation}
where
\begin{equation}
\label{eq:sub_XYL_CII}
C_{\rm II}^{P,{\rm XYL}}(u;\gamma)
=
C_{\rm II}^{P,{\rm XY}}(u;\gamma)
-\frac12 f_c(u;\gamma)
+
B_{\rm II}^{\rm XYL}(u;\gamma),
\end{equation}
with \(B_{\rm II}^{\rm XYL}(u;\gamma)=\lim_{\ell\to\infty,h=1}\ln\mathcal B_\ell^{\rm XYL}(u;1,\gamma)\).

In the crossover regime,
\begin{equation}
\label{eq:sub_regimeIII_XYL}
\ln P_\ell^{{\rm XYL},(\infty)}(u;h_\ell^\sigma,\gamma)
=
\frac{\ell}{2}f(u;h_\ell^\sigma,\gamma)
-
\frac18\ln\ell
+
C_{\rm II}^{P,{\rm XYL}}(u;\gamma)
+
\Phi_{\rm XYL}^{(\sigma)}(x;u,\gamma)
+
o(1),
\end{equation}
with
\begin{equation}
\label{eq:sub_Phi_XYL_split}
\Phi_{\rm XYL}^{(\sigma)}(x;u,\gamma)
=
\Phi_{\rm XY}^{(\sigma)}(x;u,\gamma)
+
B_{\rm cross}^{{\rm XYL},\sigma}(x;u,\gamma),
\qquad
B_{\rm cross}^{{\rm XYL},\sigma}(0;u,\gamma)=0 .
\end{equation}
Here \(B_{\rm cross}^{{\rm XYL},\sigma}\) is obtained from the scaling limit of \(\ln\mathcal B_\ell^{\rm XYL}(u;h_\ell^\sigma,\gamma)-B_{\rm II}^{\rm XYL}(u;\gamma)\).

\begin{table}[!ht]
\caption{Physical meaning and mathematical machinery of the three subsystem regimes.}
\label{tab:sub_regime_machinery}
\begin{ruledtabular}
\begin{tabular}{lll}
\rowcolor{blue!10}
Regime & Physical meaning & Mathematical machinery \\
\hline
\rowcolor{blue!3}
Off-critical, \(\ell/\xi\to\infty\) &
Finite correlation length, smooth bulk symbol &
Strong Szegő--Widom theorem \\
Critical, \(h=1\) &
Gapless interval, singular mode at \(q=\pi\) &
Fisher--Hartwig theorem \\
\rowcolor{blue!3}
Crossover, \(x=\ell/\xi=O(1)\) &
Interval size comparable to correlation length &
Block-Pfaffian scaling problem \\
\end{tabular}
\end{ruledtabular}
\end{table}

\begin{table}[!ht]
\caption{Universal and boundary-sensitive data for the three shifted XY interval blocks. The shared entries follow from the common XY bulk symbol and the common chamber \(h^2+\gamma^2>1\), \(0\le\gamma<1\). The boundary shifts encode the block-dependent \(O(1)\) data.}
\label{tab:sub_universal_boundary}
\begin{ruledtabular}
\begin{tabular}{llll}
\rowcolor{blue!10}
Quantity & \({\rm XY}\) & \({\rm XYR}\) & \({\rm XYL}\) \\
\hline
\rowcolor{blue!3}
Bulk density & \(f/2\) & Same & Same \\
Correlation length \(\xi(h,\gamma)\) & Same & Same & Same \\
\rowcolor{blue!3}
Critical point & \(h=1\) & Same & Same \\
Log coefficient & \(-1/8\) & Same & Same \\
\rowcolor{blue!3}
Off-critical constant & \(C_{\rm I}^{P,{\rm XY}}\) & \(C_{\rm I}^{P,{\rm XY}}+B_{\rm I}^{\rm XYR}\) & \(C_{\rm I}^{P,{\rm XY}}-f/2+B_{\rm I}^{\rm XYL}\) \\
Critical constant & \(C_{\rm II}^{P,{\rm XY}}\) & \(C_{\rm II}^{P,{\rm XY}}+B_{\rm II}^{\rm XYR}\) & \(C_{\rm II}^{P,{\rm XY}}-f_c/2+B_{\rm II}^{\rm XYL}\) \\
\rowcolor{blue!3}
Crossover & \(\Phi_{\rm XY}^{(\sigma)}\) & \(\Phi_{\rm XY}^{(\sigma)}+B_{\rm cross}^{{\rm XYR},\sigma}\) & \(\Phi_{\rm XY}^{(\sigma)}+B_{\rm cross}^{{\rm XYL},\sigma}\) \\
\end{tabular}
\end{ruledtabular}
\end{table}

\subsection{Final asymptotic forms}

For the three subsystem blocks \(\chi\in\{{\rm XY},{\rm XYR},{\rm XYL}\}\), the asymptotic structure can be summarized as
\begin{equation}
\label{eq:sub_final_limits_all}
\ln P_\ell^{\chi,(\infty)}
=
\begin{cases}
\dfrac{\ell}{2}f(u;h,\gamma)
+
C_{\rm I}^{P,\chi}(u;h,\gamma)
+
O(e^{-\ell/\xi}),
&
h\neq1,\quad x=\ell/\xi\to\infty,
\\[8pt]
\dfrac{\ell}{2}f_c(u;\gamma)
-
\dfrac18\ln\ell
+
C_{\rm II}^{P,\chi}(u;\gamma)
+
o(1),
&
h=1,
\\[8pt]
\dfrac{\ell}{2}f(u;h_\ell^\sigma,\gamma)
-
\dfrac18\ln\ell
+
C_{\rm II}^{P,\chi}(u;\gamma)
+
\Phi_\chi^{(\sigma)}(x;u,\gamma)
+
o(1),
&
x=\ell/\xi=O(1).
\end{cases}
\end{equation}
The matching relation is
\begin{equation}
\label{eq:sub_matching_constant_all}
C_{\rm I}^{P,\chi}(u;h,\gamma)
=
\frac18\ln \xi^{-1}(h,\gamma)
+
C_{\rm II}^{P,\chi}(u;\gamma)
+
\Phi_{\chi,\infty}^{(\sigma)}(u,\gamma)
+
o(1),
\qquad
h\to1^\sigma .
\end{equation}

The physical content of the three regimes is summarized in Table~\ref{tab:sub_regime_machinery}. The universal and block-dependent data of the three shifted interval representatives are summarized in Table~\ref{tab:sub_universal_boundary}. These tables emphasize the main point of the construction: the scalar bulk symbol controls the common density, correlation length, and leading logarithmic structure, whereas the \(O(1)\) constants and side-dependent crossover functions remain properties of the full block-Pfaffian problem.

\section{Subsystem SRE at \texorpdfstring{\(\alpha=1/2\)}{alpha=1/2}}
\label{app:subsystem-Mhalf-scaling}

In this Appendix, we apply the subsystem generating-function asymptotics to the stabilizer R\'enyi entropy at \(\alpha=1/2\). We first use the \(2^\ell\)-reference normalization. This convention isolates the Pfaffian numerator and keeps the Fisher--Hartwig logarithm explicit, which is the most useful form for a QFT-like scaling analysis. We then discuss the purity-normalized mixed-state SRE. In that normalization, the second R\'enyi entanglement entropy of the reduced interval cancels the critical logarithm, and the critical information is transferred to the finite constants and crossover functions.

Throughout this section \(\chi\in\{{\rm XY},{\rm XYR},{\rm XYL}\}\). The symbol \({\rm XY}\) denotes the ordinary interval block, \({\rm XYR}\) the Kramers--Wannier/cluster-dual shifted interval, and \({\rm XYL}\) the direct cluster-oriented shifted interval. Their Toeplitz matrices are
\(\mathbf G_\ell^{{\rm XY},(\infty)}=[g_{a-b}]\),
\(\mathbf G_\ell^{{\rm XYR},(\infty)}=[g_{1-a+b}]\), and
\(\mathbf G_\ell^{{\rm XYL},(\infty)}=[g_{a-b-1}]\), respectively. We denote the corresponding absolute-minor generating functions by \(P_\ell^{\chi,(\infty)}(u;h,\gamma)\). At \(\alpha=\frac12\), the relevant numerator is \(P_\ell^{\chi,(\infty)}(1;h,\gamma)\).

We set \(f_1(h,\gamma):=f(1;h,\gamma)\) and define the bulk density by
\(m_{\frac12}^{\rm bulk}(h,\gamma):=f_1(h,\gamma)-2\ln2\), with
\(m_{\frac12}^{\rm bulk}(1,\gamma)\) denoting its critical value. The inverse correlation length is \(\xi^{-1}(h,\gamma)\), and the crossover variable is \(x=\ell/\xi(h,\gamma)\). The numerator constants are abbreviated as
\(C_{\rm I}^{P,\chi}(h,\gamma):=C_{\rm I}^{P,\chi}(1;h,\gamma)\),
\(C_{\rm II}^{P,\chi}(\gamma):=C_{\rm II}^{P,\chi}(1;\gamma)\), and
\(\Phi_\chi^{(\sigma)}(x;\gamma):=\Phi_\chi^{(\sigma)}(x;1,\gamma)\).

The exact subsystem construction is Pfaffian and is naturally expressed in terms of a \(2\times2\) block symbol. The scalar symbol used in the scaling formulas should therefore not be interpreted as a literal replacement of the block-Pfaffian problem. Its role is more restricted: it fixes the common extensive density, the inverse correlation length, and the local Fisher--Hartwig singularity associated with the ordinary XY kernel. This is sufficient for organizing the leading scaling forms. The \(O(1)\) constants and the crossover functions remain block-dependent quantities. This distinction is especially important for \(XYR\) and \(XYL\), where the shifted intervals share the XY bulk but may differ by boundary Pfaffian data.

The three blocks have the same bulk density, correlation length, critical point, and Fisher--Hartwig logarithmic coefficient. The differences are boundary-sensitive. For the \({\rm XYR}\) block this comes from the coefficient shift \(g_{a-b}\mapsto g_{1-a+b}\). For the \({\rm XYL}\) block this follows from the exact snake factorization
\begin{equation}
\label{eq:app-plus-logP-factor-Mhalf}
\ln P_\ell^{{\rm XYL},(\infty)}(u;h,\gamma)
=
\ln P_{\ell-1}^{{\rm XY},(\infty)}(u;h,\gamma)
+
\ln\mathcal B_\ell^{\rm XYL}(u;h,\gamma),
\end{equation}
where \(\mathcal B_\ell^{\rm XYL}\) is the scalar boundary Pfaffian--Schur factor defined in Appendix~\ref{app:subsystem_toeplitz_pfaffian}. Thus the direct cluster-oriented interval has the same XY bulk, with a one-site size shift and a finite boundary Green function.

\paragraph*{\texorpdfstring{\(2^\ell\)}{2^ell}-normalized SRE.}

We define the \(2^\ell\)-reference-normalized subsystem SRE by
\begin{equation}
\label{eq:app-Mtilde-reference-def}
\widetilde M_{\frac12}^{\chi,(\infty)}(\ell;h,\gamma)
=
2\ln P_\ell^{\chi,(\infty)}(1;h,\gamma)
-
2\ell\ln2 .
\end{equation}
This is not yet the purity-normalized mixed-state SRE. Its role is to expose the raw Fisher--Hartwig logarithm of the Pfaffian numerator.

In the off-critical regime, \(h\neq1\) and \(x=\ell/\xi(h,\gamma)\to\infty\), the numerator asymptotics give
\begin{equation}
\label{eq:app-Mtilde-chi-regimeI}
\widetilde M_{\frac12}^{\chi,(\infty)}(\ell;h,\gamma)
=
\ell m_{\frac12}^{\rm bulk}(h,\gamma)
+
2C_{\rm I}^{P,\chi}(h,\gamma)
+
O(e^{-\ell/\xi(h,\gamma)}).
\end{equation}
The three blocks have the same density. Their differences appear only in \(O(1)\) terms. For the shifted \({\rm XYR}\) block,
\(C_{\rm I}^{P,{\rm XYR}}=C_{\rm I}^{P,{\rm XY}}+B_{\rm I}^{\rm XYR}\). For the direct \({\rm XYL}\) block, Eq.~\eqref{eq:app-plus-logP-factor-Mhalf} gives
\(C_{\rm I}^{P,{\rm XYL}}=C_{\rm I}^{P,{\rm XY}}-\frac12 f_1(h,\gamma)+B_{\rm I}^{\rm XYL}\), where
\(B_{\rm I}^{\rm XYL}(h,\gamma)=\lim_{\ell\to\infty}\ln\mathcal B_\ell^{\rm XYL}(1;h,\gamma)\). The term \(-f_1/2\) is the one-site shift from \(P_{\ell-1}^{{\rm XY},(\infty)}\).

At criticality, \(h=1\), the Pfaffian numerator has a Fisher--Hartwig contribution \(-\frac18\ln\ell\). Therefore the \(2^\ell\)-reference-normalized SRE has
\begin{equation}
\label{eq:app-Mtilde-chi-critical}
\widetilde M_{\frac12}^{\chi,(\infty)}(\ell;1,\gamma)
=
\ell m_{\frac12}^{\rm bulk}(1,\gamma)
-
\frac14\ln\ell
+
2C_{\rm II}^{P,\chi}(\gamma)
+
o(1).
\end{equation}
This is the normalization in which the critical Fisher--Hartwig logarithm remains visible. The logarithmic coefficient is common to \({\rm XY}\), \({\rm XYR}\), and \({\rm XYL}\). The constants are boundary-sensitive:
\(C_{\rm II}^{P,{\rm XYR}}=C_{\rm II}^{P,{\rm XY}}+B_{\rm II}^{\rm XYR}\), while
\(C_{\rm II}^{P,{\rm XYL}}=C_{\rm II}^{P,{\rm XY}}-\frac12 f_1(1,\gamma)+B_{\rm II}^{\rm XYL}\), with
\(B_{\rm II}^{\rm XYL}(\gamma)=\lim_{\ell\to\infty}\ln\mathcal B_\ell^{\rm XYL}(1;1,\gamma)\).

In the near-critical regime, let \(h=h_\ell^\sigma\) with \(\ell/\xi(h_\ell^\sigma,\gamma)=x\) and \(\sigma=\operatorname{sgn}(h_\ell^\sigma-1)\). The moving bulk form gives
\begin{equation}
\label{eq:app-Mtilde-chi-crossover}
\widetilde M_{\frac12}^{\chi,(\infty)}(\ell;h_\ell^\sigma,\gamma)
=
\ell m_{\frac12}^{\rm bulk}(h_\ell^\sigma,\gamma)
-
\frac14\ln\ell
+
2C_{\rm II}^{P,\chi}(\gamma)
+
2\Phi_\chi^{(\sigma)}(x;\gamma)
+
o(1).
\end{equation}
The crossover satisfies \(\Phi_\chi^{(\sigma)}(0;\gamma)=0\) and
\(\Phi_\chi^{(\sigma)}(x;\gamma)=\frac18\ln x+\Phi_{\chi,\infty}^{(\sigma)}(\gamma)+o(1)\) as \(x\to\infty\). Hence, in the \(2^\ell\)-reference-normalized SRE, the combination
\(-\frac14\ln\ell+2\Phi_\chi^{(\sigma)}(x;\gamma)\) becomes
\(\frac14\ln\xi^{-1}(h_\ell^\sigma,\gamma)+2\Phi_{\chi,\infty}^{(\sigma)}(\gamma)+o(1)\), reproducing the near-critical singularity of the off-critical numerator constant.

For the ordinary \(XY\) block, the large-\(x\) crossover tail can be tested directly by comparing the two approaches \(h_\ell^+\) and \(h_\ell^-\). Define the finite-\(\ell\) branch difference after subtracting the moving bulk and the critical normalization by
\(\Delta_\ell^{\rm XY}(x;\gamma):=\Phi_{\ell,{\rm XY}}^{(+)}(x;\gamma)-\Phi_{\ell,{\rm XY}}^{(-)}(x;\gamma)\). The observed collapse of \(\Delta_\ell^{\rm XY}(x;\gamma)\) at large \(x\) indicates that, for the ordinary interval, the two large-\(x\) tails approach the same limiting endpoint data. Equivalently, the \(XY\) matching constant can be taken branch-independent at the level of the limiting large-\(x\) tail, \(\Phi_{{\rm XY},\infty}^{(+)}(\gamma)=\Phi_{{\rm XY},\infty}^{(-)}(\gamma)\), within the scaling accuracy used here.

For the shifted blocks, it is useful to separate the XY crossover from boundary data. We write
\(\Phi_{\rm XYR}^{(\sigma)}=\Phi_{\rm XY}^{(\sigma)}+B_{\rm cross}^{{\rm XYR},\sigma}\) and
\(\Phi_{\rm XYL}^{(\sigma)}=\Phi_{\rm XY}^{(\sigma)}+B_{\rm cross}^{{\rm XYL},\sigma}\), with
\(B_{\rm cross}^{{\rm XYR},\sigma}(0;\gamma)=B_{\rm cross}^{{\rm XYL},\sigma}(0;\gamma)=0\). The logarithmic endpoint is unchanged; only the finite crossover function is boundary-sensitive.

Collecting the three regimes, the \(2^\ell\)-reference-normalized SRE satisfies
\begin{equation}
\label{eq:app-Mtilde-chi-regimes}
\widetilde M_{\frac12}^{\chi,(\infty)}(\ell;h,\gamma)
=
\begin{cases}
\ell m_{\frac12}^{\rm bulk}(h,\gamma)
+
2C_{\rm I}^{P,\chi}(h,\gamma)
+
O(e^{-\ell/\xi(h,\gamma)}),
&
h\neq1,\quad x=\ell/\xi\to\infty,
\\[2mm]
\ell m_{\frac12}^{\rm bulk}(1,\gamma)
-
\dfrac14\ln\ell
+
2C_{\rm II}^{P,\chi}(\gamma)
+
o(1),
&
h=1,
\\[2mm]
\ell m_{\frac12}^{\rm bulk}(h_\ell^\sigma,\gamma)
-
\dfrac14\ln\ell
+
2C_{\rm II}^{P,\chi}(\gamma)
+
2\Phi_\chi^{(\sigma)}(x;\gamma)
+
o(1),
&
x=\ell/\xi=O(1).
\end{cases}
\end{equation}
The matching condition is
\begin{equation}
\label{eq:app-Mtilde-chi-matching}
2C_{\rm I}^{P,\chi}(h,\gamma)
=
\frac14\ln \xi^{-1}(h,\gamma)
+
2C_{\rm II}^{P,\chi}(\gamma)
+
2\Phi_{\chi,\infty}^{(\sigma)}(\gamma)
+
o(1),
\qquad
h\to1^\sigma .
\end{equation}

The \(2^\ell\)-reference-normalized quantity \(\widetilde M_{\frac12}^{\chi,(\infty)}\) is the cleanest way to expose the critical anomaly of the stabilizer numerator. At the anisotropic critical line,
\(\widetilde M_{\frac12}^{\chi,(\infty)}(\ell;1,\gamma)=\ell m_{\frac12}^{\rm bulk}(1,\gamma)-\frac14\ln\ell+O(1)\). The logarithmic coefficient is common to the three blocks because they share the same local Fisher--Hartwig singularity at \(q=\pi\). The constants are boundary-sensitive.

\paragraph*{Ordinary \(XY\) scaling regimes.}
It is useful to isolate the ordinary \(XY\) result before introducing the shifted representatives. For \(\chi={\rm XY}\), the \(2^\ell\)-reference-normalized SRE has the following three regimes. In the massive regime, \(h\neq1\) and \(x=\ell/\xi(h,\gamma)\to\infty\),
\begin{equation}
\label{eq:app-Mtilde-XY-regimeI}
\widetilde M_{\frac12}^{{\rm XY},(\infty)}(\ell;h,\gamma)
=
\ell m_{\frac12}^{\rm bulk}(h,\gamma)
+
2C_{\rm I}^{P,{\rm XY}}(h,\gamma)
+
O(e^{-\ell/\xi(h,\gamma)}).
\end{equation}
At the anisotropic critical line, \(h=1\),
\begin{equation}
\label{eq:app-Mtilde-XY-critical}
\widetilde M_{\frac12}^{{\rm XY},(\infty)}(\ell;1,\gamma)
=
\ell m_{\frac12}^{\rm bulk}(1,\gamma)
-
\frac14\ln\ell
+
2C_{\rm II}^{P,{\rm XY}}(\gamma)
+
o(1).
\end{equation}
In the crossover regime, \(h=h_\ell^\sigma\), \(\ell/\xi(h_\ell^\sigma,\gamma)=x=O(1)\), and \(\sigma=\operatorname{sgn}(h_\ell^\sigma-1)\). For the ordinary \(XY\) block, the finite-\(\ell\) branch difference is absorbed into the first subleading correction, while the leading crossover is governed by a single limiting function. We therefore write
\begin{equation}
\label{eq:app-Mtilde-XY-crossover}
\widetilde M_{\frac12}^{{\rm XY},(\infty)}(\ell;h_\ell^\sigma,\gamma)
=
\ell m_{\frac12}^{\rm bulk}(h_\ell^\sigma,\gamma)
-
\frac14\ln\ell
+
2C_{\rm II}^{P,{\rm XY}}(\gamma)
+
2\Phi_{\rm XY}(x)
+
\frac{2\Psi_{\rm XY}^{\sigma}(x;\gamma)}{\sqrt{\ell}}
+
o(\ell^{-1/2}).
\end{equation}
Here \(x=O(1)\) defines the crossover scaling limit. The function \(\Phi_{\rm XY}(x)\) is the leading ordinary-\(XY\) crossover function after the moving bulk and the critical normalization have been subtracted. In this leading scaling function, the dependence on the side of the transition and on the anisotropy has dropped out; this dependence is retained in the finite-size correction \(\Psi_{\rm XY}^{\sigma}(x;\gamma)\).

Equivalently, defining the finite-\(\ell\) extracted crossover function by
\begin{equation}
\label{eq:app-Phi-ell-XY-def}
\Phi_{\ell,{\rm XY}}^{\sigma}(x;\gamma)
=
\frac12
\left[
\widetilde M_{\frac12}^{{\rm XY},(\infty)}(\ell;h_\ell^\sigma,\gamma)
-
\ell m_{\frac12}^{\rm bulk}(h_\ell^\sigma,\gamma)
+
\frac14\ln\ell
-
2C_{\rm II}^{P,{\rm XY}}(\gamma)
\right],
\end{equation}
one has
\begin{equation}
\label{eq:app-Phi-ell-XY-expansion}
\Phi_{\ell,{\rm XY}}^{\sigma}(x;\gamma)
=
\Phi_{\rm XY}(x)
+
\frac{\Psi_{\rm XY}^{\sigma}(x;\gamma)}{\sqrt{\ell}}
+
o(\ell^{-1/2}).
\end{equation}
Thus the branch difference satisfies
\begin{equation}
\label{eq:app-Delta-Phi-XY-scaling}
\Delta_\ell^{\rm XY}(x;\gamma)
:=
\Phi_{\ell,{\rm XY}}^{+}(x;\gamma)
-
\Phi_{\ell,{\rm XY}}^{-}(x;\gamma)
=
\frac{
\Psi_{\rm XY}^{+}(x;\gamma)
-
\Psi_{\rm XY}^{-}(x;\gamma)
}{\sqrt{\ell}}
+
o(\ell^{-1/2}).
\end{equation}
The endpoint behavior of the leading crossover function is fixed by matching:
\begin{equation}
\label{eq:app-Phi-XY-endpoints-Mhalf}
\Phi_{\rm XY}(0)=0,
\qquad
\Phi_{\rm XY}(x)
=
\frac18\ln x
+
\Phi_{{\rm XY},\infty}
+
o(1),
\qquad
x\to\infty .
\end{equation}
Consequently,
\begin{equation}
\label{eq:app-Mtilde-XY-matching}
2C_{\rm I}^{P,{\rm XY}}(h,\gamma)
=
\frac14\ln\xi^{-1}(h,\gamma)
+
2C_{\rm II}^{P,{\rm XY}}(\gamma)
+
2\Phi_{{\rm XY},\infty}
+
o(1),
\qquad
h\to1 .
\end{equation}
The branch dependence is preasymptotic for the ordinary \(XY\) interval and appears at order \(\ell^{-1/2}\). This branch simplification is not assumed for \(XYR\) or \(XYL\), where additional boundary data are present.

\paragraph*{Normalization determinant and second R\'enyi entanglement.}

For a physically reduced interval, or for a shifted Gaussian block entering a decimation formula, the corresponding Gaussian normalization is
\begin{equation}
\label{eq:app-normalization-det}
\mathcal N_\ell^{\chi,(\infty)}(h,\gamma)
=
\det\!\left[
\mathbf I_\ell+
\bigl(\mathbf G_\ell^{\chi,(\infty)}\bigr)^T
\mathbf G_\ell^{\chi,(\infty)}
\right].
\end{equation}
Equivalently, if \(s_a\) are the singular values of \(\mathbf G_\ell^{\chi,(\infty)}\), then
\(\ln\mathcal N_\ell^{\chi,(\infty)}=\sum_{a=1}^{\ell}\ln(1+s_a^2)\). This gives a stable finite-\(\ell\) evaluation.

The determinant has the physical interpretation
\begin{equation}
\label{eq:app-normalization-S2}
\mathcal N_\ell^{\chi,(\infty)}(h,\gamma)
=
2^\ell\operatorname{Tr}\bigl[(\rho_\ell^\chi)^2\bigr],
\qquad
\ln\mathcal N_\ell^{\chi,(\infty)}
=
\ell\ln2
-
S_2^\chi(\rho_\ell).
\end{equation}
For shifted blocks, this formula should be understood as the Gaussian normalization of the corresponding block contribution. The boundary shift changes finite terms but not the universal logarithmic coefficient.

The second R\'enyi entropy has the scaling structure
\begin{equation}
\label{eq:app-S2-chi-regimes}
S_2^\chi(\rho_\ell)
=
\begin{cases}
S_{2,{\rm sat}}^\chi(h,\gamma)+o(1),
&
h\neq1,\quad x=\ell/\xi\to\infty,
\\[1mm]
\dfrac18\ln\ell+s_2^\chi(\gamma)+o(1),
&
h=1,
\\[2mm]
\dfrac18\ln\ell+\mathcal S_2^{\chi,\sigma}(x;\gamma)+o(1),
&
x=\ell/\xi=O(1).
\end{cases}
\end{equation}
At criticality, the coefficient \(1/8\) is the Ising CFT value \(c/4\), with \(c=1/2\). In the crossover regime, the matching conditions are \(\mathcal S_2^{\chi,\sigma}(x;\gamma)\to s_2^\chi(\gamma)\) as \(x\to0\), and
\begin{equation}
\label{eq:app-S2-large-x}
\mathcal S_2^{\chi,\sigma}(x;\gamma)
=
-\frac18\ln x
+
\mathcal S_{2,\infty}^{\chi,\sigma}(\gamma)
+
o(1),
\qquad x\to\infty.
\end{equation}
The large-\(x\) behavior follows from the off-critical saturation
\(S_{2,{\rm sat}}^\chi(h,\gamma)\sim-\frac18\ln\xi^{-1}(h,\gamma)+\mathrm{const}\)
near criticality.

\paragraph*{Purity-normalized mixed-state SRE.}

The purity-normalized mixed-state SRE is
\begin{equation}
\label{eq:app-M-purity-def}
M_{\frac12}^{\chi,(\infty)}(\ell;h,\gamma)
=
2\ln P_\ell^{\chi,(\infty)}(1;h,\gamma)
-
2\ln\mathcal N_\ell^{\chi,(\infty)}(h,\gamma).
\end{equation}
Using \(\ln\mathcal N_\ell^{\chi,(\infty)}=\ell\ln2-S_2^\chi(\rho_\ell)\), this can be written as
\begin{equation}
\label{eq:app-M-purity-from-reference}
M_{\frac12}^{\chi,(\infty)}(\ell;h,\gamma)
=
\widetilde M_{\frac12}^{\chi,(\infty)}(\ell;h,\gamma)
+
2S_2^\chi(\rho_\ell).
\end{equation}

Substituting the reference-normalized SRE and the R\'enyi entropy regimes gives
\begin{equation}
\label{eq:app-M-purity-regimes}
M_{\frac12}^{\chi,(\infty)}(\ell;h,\gamma)
=
\begin{cases}
\ell m_{\frac12}^{\rm bulk}(h,\gamma)
+
2C_{\rm I}^{P,\chi}(h,\gamma)
+
2S_{2,{\rm sat}}^\chi(h,\gamma)
+
o(1),
&
h\neq1,\quad x\to\infty,
\\[2mm]
\ell m_{\frac12}^{\rm bulk}(1,\gamma)
+
2C_{\rm II}^{P,\chi}(\gamma)
+
2s_2^\chi(\gamma)
+
o(1),
&
h=1,
\\[2mm]
\ell m_{\frac12}^{\rm bulk}(h_\ell^\sigma,\gamma)
+
2C_{\rm II}^{P,\chi}(\gamma)
+
2\Phi_\chi^{(\sigma)}(x;\gamma)
+
2\mathcal S_2^{\chi,\sigma}(x;\gamma)
+
o(1),
&
x=O(1).
\end{cases}
\end{equation}

The critical logarithm has disappeared. The \(2^\ell\)-reference-normalized SRE contributes \(-\frac14\ln\ell\), while \(2S_2^\chi(\rho_\ell)\) contributes \(+\frac14\ln\ell\). Thus the purity-normalized reduced SRE has no residual \(\ln\ell\) term at the critical point.

The same cancellation occurs in the large-\(x\) crossover tail. Since \(2\Phi_\chi^{(\sigma)}(x;\gamma)\sim\frac14\ln x\) and \(2\mathcal S_2^{\chi,\sigma}(x;\gamma)\sim-\frac14\ln x\), the logarithmic \(x\)-dependence cancels. Equivalently, the near-critical \(\frac14\ln\xi^{-1}\) singularity of the numerator is canceled by the \(-\frac14\ln\xi^{-1}\) singularity of the entanglement normalization.

The \(2^\ell\)-reference-normalized SRE keeps the Pfaffian Fisher--Hartwig anomaly visible:
\(\widetilde M_{\frac12}^{\chi,(\infty)}(\ell;1,\gamma)=\ell m_{\frac12}^{\rm bulk}(1,\gamma)-\frac14\ln\ell+O(1)\). This is the natural convention for identifying the CFT/Fisher--Hartwig logarithm of the numerator. The purity-normalized mixed-state SRE subtracts the Gaussian norm of the reduced density matrix and removes the critical logarithm. Therefore, the absence of a logarithm in the purity-normalized SRE should not be interpreted as the absence of critical structure. Rather, the universal information is encoded in the finite constants and in the finite crossover combination \(2\Phi_\chi^{(\sigma)}(x;\gamma)+2\mathcal S_2^{\chi,\sigma}(x;\gamma)\). For \({\rm XY}\), \({\rm XYR}\), and \({\rm XYL}\), the logarithmic cancellation is the same, while the finite constants and crossover functions remain boundary-sensitive.

\end{document}